\documentclass[letterpaper,twocolumn,10pt]{article}

\usepackage{times}
\usepackage{fullpage}
\usepackage{float} %
\usepackage{bm} %
\usepackage[dvipsnames]{xcolor} %
\usepackage{booktabs}  %
\usepackage{dcolumn}   %
\usepackage{multirow}  %
\usepackage{rotating}
\usepackage{xspace}
\usepackage{hyphenat}  %
\usepackage{enumerate}
\usepackage{mfirstuc} %
\usepackage{tikz} %
\usepackage{xparse} %
\usepackage{listings}
\definecolor{dkgreen}{rgb}{0,.6,0}
\definecolor{dkblue}{rgb}{0,0,.6}
\definecolor{dkyellow}{cmyk}{0,0,.8,.3}
\usepackage[square,comma,numbers,sort&compress]{natbib}

\usepackage{wrapfig}
\usepackage{textcomp}
\usepackage{tabularx}
\usepackage{pifont}
\usepackage{url}

\usepackage{breakurl}

\usepackage{wrapfig}

\usepackage{ulem}
\usepackage{mfirstuc}

\usepackage{colortbl} %
\usepackage{array}    %

\usepackage{dblfloatfix}

\usepackage[font=small]{caption}
\usepackage{subcaption}

\usepackage[noend]{algpseudocode}
\algrenewcomment[1]{\hfill// #1}%
\algrenewcommand\algorithmicindent{1.25em}

\definecolor{darkgreen}{rgb}{0,0.75,0}

\algrenewcommand\algorithmicdo{:}
\algrenewcommand\algorithmicwhile{\textbf{while}}
\algrenewcommand\algorithmicfor{\textbf{for}}
\algrenewcommand\algorithmicforall{\textbf{for all}}
\algrenewcommand\algorithmicloop{\textbf{loop}}
\algrenewcommand\algorithmicrepeat{\textbf{repeat}}
\algrenewcommand\algorithmicuntil{\textbf{until}}
\algrenewcommand\algorithmicprocedure{\textbf{procedure}}
\algrenewcommand\algorithmicfunction{\textbf{function}}
\algrenewcommand\algorithmicif{\textbf{if}}
\algrenewcommand\algorithmicthen{:}
\algrenewcommand\algorithmicelse{\textbf{else}}
\algrenewcommand\algorithmicrequire{\textbf{Require}:}
\algrenewcommand\algorithmicensure{\textbf{Ensure}:}
\algrenewcommand\algorithmicreturn{\textbf{return}}
\algdef{SE}[WHILE]{While}{EndWhile}[1]{\algorithmicwhile\ #1\algorithmicdo}{\algorithmicend\ \algorithmicwhile}%
\algdef{SE}[FOR]{For}{EndFor}[1]{\algorithmicfor\ #1\algorithmicdo}{\algorithmicend\ \algorithmicfor}%
\algdef{S}[FOR]{ForAll}[1]{\algorithmicforall\ #1\algorithmicdo}%
\algdef{SE}[IF]{If}{EndIf}[1]{\algorithmicif\ #1\algorithmicthen}{\algorithmicend\ \algorithmicif}%
\algdef{C}[IF]{IF}{ElsIf}[1]{\algorithmicelse\ \algorithmicif\ #1\algorithmicthen}%
\algnotext{Function}                
\algnotext{EndFunction}
\algnotext{EndFor}
\algnotext{EndIf}
\algnotext{EndWhile}
\algnewcommand\Or{\textbf{or}\xspace}
\algnewcommand\myAnd{\textbf{and}\xspace}

\usepackage{amsmath,amscd}
\usepackage{amssymb}
\usepackage{amsfonts}
\usepackage{amsthm}

\usepackage[noeka]{mathrmletter}
\usepackage{tgtermes}

\newif\ifextended
\newif\iflongbatching
\newif\ifsubmission
\newif\ifelementary
\newif\ifwithdbproof
\newif\ifbuildanonapdx

\ifx\buildextended\undefined
\else
    \extendedtrue
\fi

\def\noeditingmarks{}

\ifx\noeditingmarks\undefined
   \newcommand{\pgwrapper}[3]{\begingroup \color{#1} #2: #3 \endgroup}
   \newcommand{\pgwrapperb}[1]{\textbf{#1}}
   \newcommand{\dangerwrapper}[1]{{\color{red}#1}}
\else
   \newcommand{\pgwrapperb}[1]{}
   \newcommand{\pgwrapper}[3]{}
   \newcommand{\dangerwrapper}[1]{}
\fi

\def\hn{\sffamily\selectfont}
\newcommand{\mpfont}{\hn\scriptsize}

\ifx\noeditingmarks\undefined
    \newcommand{\MPworker}[2]{\unskip{\color{#1}\vrule\vrule}{\marginpar{\raggedright\color{#1}\mpfont #2}}}
\else
    \newcommand{\MPworker}[2]{\unskip}
\fi

\ifx\noeditingmarks\undefined
    
\else
    
\fi

\def\t{\textit}

\newcommand{\sys}{\textsc{cobra}\xspace}
\newcommand{\Sys}{\sys}

\newcommand{\oneshot}{one-shot verification\xspace}

\newcommand{\CF}[1]{\xmakefirstuc{#1}}

\newcommand{\twitter}{C-Twitter\xspace}
\newcommand{\rubis}{C-RUBiS\xspace}

\makeatletter
\def\imod#1{\allowbreak\mkern10mu({\operator@font mod}\,\,#1)}
\makeatother


\def\compactify{\itemsep=0in \topsep=2pt \parsep=0.00in \partopsep=0pt
\leftmargin=2em}
\let\latexusecounter=\usecounter

  {\begin{list}{\labelitemi}{\itemsep3pt \topsep3pt \parsep0.00in
  \partopsep=3pt \leftmargin1em}}%
  {\end{list}}
\newenvironment{myitemize2}%
  {\begin{list}{\labelitemi}{\itemsep1pt \topsep2pt \parsep0.00in
  \partopsep=0pt \leftmargin1.2em}}%
  {\end{list}}
  {\begin{list}{\labelitemi}{\itemsep2pt \topsep2pt \parsep0.00in
  \partopsep=0pt \leftmargin1.2em}}%
  {\end{list}}
  {\begin{list}{\labelitemi}{\itemsep3pt \topsep3pt \parsep0.00in
  \partopsep=3pt \leftmargin1.5em}}%
  {\end{list}}
  {\begin{list}{\S3.3}{\itemsep3pt \topsep3pt \parsep0.00in
  \partopsep=3pt\leftmargin1em}}%
  {\end{list}}

\newenvironment{myenumerate}
  {\def\usecounter{\compactify\latexusecounter}
   \begin{enumerate}}
  {\end{enumerate}\let\usecounter=\latexusecounter}

\def\compactsortof{\itemsep=0in \topsep=2pt \parsep=0.00in \partopsep=0pt
\leftmargin=1.2em}

\def\compactsqueeze{\itemsep=0pt \topsep0pt \parsep=0ex \partopsep=0pt
\leftmargin=1.63em}

\ifextended
\def\compactRenum{\itemsep=1ex \topsep=1ex \parsep=0.00in \partopsep=0pt
\leftmargin=2.05em}
\else
\def\compactRenum{\itemsep=0in \topsep=2pt \parsep=0.00in \partopsep=0pt
\leftmargin=2.05em}
\fi

\ifx\normalpar\undefined
  
\else
  
\fi

\def\discretionaryslash{\discretionary{/}{}{/}}
{\catcode`\/\active
\gdef\URLprepare{\catcode`\/\active\let/\discretionaryslash
        \def~{\char`\~}}}%
\def\URL{\bgroup\URLprepare\realURL}%
\def\realURL#1{\tt #1\egroup}%

\NewDocumentCommand{\xrightarrows}{ O{}O{} }{%
\mathrel{%
\vcenter{\hbox{%
\begin{tikzpicture}
  \node[minimum width=0.5cm,minimum height=1ex,anchor=south,align=center] (a){\text{\vphantom{hg}#1}\\[0.5ex] \vphantom{hg}#2};
  \draw[->] ([yshift=0.25ex]a.west) -- ([yshift=0.25ex]a.east);
  \draw[<-] ([yshift=-0.25ex]a.west) -- ([yshift=-0.25ex]a.east);
\end{tikzpicture}
}}%
}%
}

\newcommand{\wrarrow}[1]{\xrightarrow[]{\text{wr(#1)}}}

\newcommand{\wrarrowx}{\xrightarrow[]{\text{wr}}}
\newcommand{\rwarrowx}{\xrightarrow[]{\text{rw}}}

\newcommand{\heading}[1]{
  \vspace{1ex}
  \noindent
  \textbf{#1}}

\newcommand{\constraint}{constraint\xspace}
\newcommand{\constraints}{constraints\xspace}
\newcommand{\trace}{history\xspace}

\newcommand{\traces}{histories\xspace}

\newcommand{\wrrels}{read-dependencies\xspace}

\newcommand{\rwrel}{anti-dependency\xspace}

\newcommand{\pscc}{P-SCC\xspace}
\newcommand{\polyscc}{poly-strongly connected component\xspace}
\newcommand{\polysccs}{poly-strongly connected components\xspace}

\newcommand{\polyg}{polygraph\xspace}

\newcommand{\pg}{precedence graph\xspace} %
\newcommand{\pgs}{precedence graphs\xspace}

\newcommand{\mypolyg}{\sys polygraph\xspace}
\newcommand{\wellformed}{not easily rejectable\xspace}
\newcommand{\notwellformed}{easily rejectable\xspace}
\newcommand{\adjacentwrites}{successive writes\xspace}
\newcommand{\adjacentwrite}{successive write\xspace}
\newcommand{\clientser}{strong session serializable\xspace}
\newcommand{\clientserty}{strong session serializability\xspace}
\newcommand{\depinfo}{extended history\xspace}
\newcommand{\E}{E}
\newcommand{\Po}{P}
\newcommand{\Qo}{Q}
\newcommand{\Pp}{P_p}
\newcommand{\Qp}{Q_p}
\newcommand{\del}{\ominus}
\newcommand{\wellff}{well-formed\xspace}
\newcommand{\solvedcon}{solved constraint\xspace}
\newcommand{\unsolvedcon}{unsolved constraint\xspace}
\newcommand{\solvedcons}{solved constraints\xspace}
\newcommand{\unsolvedcons}{unsolved constraints\xspace}
\newcommand{\delcand}{removable\xspace}

\newcommand{\vpscc}{pscc}
\newcommand{\obsolete}{obsolete\xspace}

\newcommand{\chain}{chain\xspace}
\newcommand{\chains}{\chain{s}\xspace}
\newcommand{\wfence}{Wfence\xspace}
\newcommand{\rfence}{Rfence\xspace}
\newcommand{\wfences}{Wfences\xspace}
\newcommand{\rfences}{Rfences\xspace}
\newcommand{\epochis}[1]{\t{epoch}[{#1}]}

\def\t{\textit}
\newcommand{\pluseq}{\mathrel{+}=}

\newcommand{\subeq}{\mathrel{-}=}

\newcommand{\vg}{\t{g}\xspace}
\newcommand{\vtx}{\t{tx}\xspace}
\newcommand{\ven}{\t{ent}\xspace}
\newcommand{\wwpairs}{\t{wwpairs}\xspace}
\newcommand{\readfrom}{\t{readfrom}\xspace}
\newcommand{\vconstraints}{\t{con}\xspace}

\newcommand{\vperkey}{\t{chains}}
\newcommand{\vkey}{\t{key}}
\newcommand{\vwlist}{\t{chain}\xspace}
\newcommand{\vhead}{\t{head}\xspace}
\newcommand{\vtail}{\t{tail}\xspace}

\newcommand{\edgeset}{\t{edge\_set}}

\newcommand{\vrtx}{\t{rtx}\xspace}
\newcommand{\vwtx}{\t{wtx}\xspace}

\newcommand{\vrm}{\t{tr}\xspace}

\newcommand{\vfe}{\t{fepoch}\xspace}

\newcommand{\epochagree}{\t{epoch}_{\textrm{agree}}}

\newcommand{\wone}{\vtx_\textrm{w1}}
\newcommand{\wtwo}{\vtx_\textrm{w2}}
\newcommand{\rone}{\vtx_\textrm{r}}

\newcommand{\rop}{\t{rop}\xspace}
\newcommand{\wop}{\t{wrop}\xspace}

\newcommand{\fkey}{\textrm{key}}
\newcommand{\ffirst}{\textrm{head}}
\newcommand{\flast}{\textrm{tail}}
\newcommand{\fepoch}{\textrm{epoch}\xspace}
\newcommand{\ffrozen}{\textrm{frozen}\xspace}
\newcommand{\fobsolete}{\textrm{\obsolete}\xspace}
\newcommand{\fdelcand}{\textrm{\delcand}\xspace}
\newcommand{\fdelcandcand}{\textrm{candidate}\xspace}

\newcommand{\rounds}{rounds\xspace}

\newcommand{\ConstructEncoding}{\mbox{\textrm{ConstructEncoding}}\xspace}
\newcommand{\CreateKnownGraph}{\mbox{\textrm{CreateKnownGraph}}\xspace}
\newcommand{\GenConstraint}{\mbox{\textrm{GenConstraints}}\xspace}

\newcommand{\PruneConstraintsRMW}{\mbox{\textrm{CombineWrites}}\xspace}
\newcommand{\PruneConstraintsTR}{\mbox{\textrm{Prune}}\xspace}
\newcommand{\BundleConstraints}{\mbox{\textrm{Coalesce}}\xspace}

\newcommand{\VerifySerializability}{\mbox{\textrm{VerifySerializability}}\xspace}

\newcommand{\ComputeStronglyConnectedComponents}{\mbox{\textrm{CalcStronglyConnectedComponents}}\xspace}
\newcommand{\TransitiveClosure}{\mbox{\textrm{TransitiveClosure}}\xspace}
\newcommand{\InferRWEdges}{\mbox{\textrm{InferRWEdges}}\xspace}
\newcommand{\GenChainToChainEdges}{\mbox{\textrm{GenChainToChainEdges}}\xspace}

\newcommand{\AssignEpoch}{\mbox{\textrm{AssignEpoch}}\xspace}

\newcommand{\CreateKnownGraphTwo}{\mbox{\textrm{CreateKnownGraph2}}\xspace}
\newcommand{\SafeDeletionTwo}{\mbox{\textrm{SafeDeletion2}}\xspace}

\newcommand{\EncodeSMT}{\mbox{\textrm{EncodeSMT}}\xspace}
\newcommand{\GenIndependentClusters}{\mbox{\textrm{GenPolySCCs}}\xspace}
\newcommand{\EncodeAndSolve}{\mbox{\textrm{EncodeAndSolve}}\xspace}

\newcommand{\GarbageCollection}{\mbox{\textrm{GarbageCollection}}\xspace}
\newcommand{\SetFrozen}{\mbox{\textrm{SetFrozen}}\xspace}
\newcommand{\SetObsolete}{\mbox{\textrm{Set\CF{\obsolete}}}\xspace}
\newcommand{\SetRemovable}{\mbox{\textrm{Set\CF{\delcand}}}\xspace}
\newcommand{\GenFrontier}{\mbox{\textrm{GenFrontier}}\xspace}

\newtheorem{theorem2}{Theorem} %
\newtheorem{lemma2}[theorem2]{Lemma}
\newtheorem{claim}[theorem2]{Claim}
\theoremstyle{acmdefinition}
\newtheorem{definition2}[theorem2]{Definition} %
\newtheorem{fact}[theorem2]{Fact}
\newtheorem{corollary2}[theorem2]{Corollary}

\usepackage[noheadfoot,left=0.75in,right=0.75in,top=1.0in,
            footskip=0.5in,bottom=1in,
            columnsep=0.33in
	    ]{geometry}

\usepackage{fancyhdr}

\usepackage{etoolbox}
\makeatletter
\patchcmd{\ttlh@hang}{\parindent\z@}{\parindent\z@\leavevmode}{}{}
\patchcmd{\ttlh@hang}{\noindent}{}{}{}
\makeatother

\begin{document}
\title{
  Detecting Incorrect Behavior of Cloud Databases as an Outsider
}
\author{Cheng Tan, Changgeng Zhao, Shuai Mu$^\star$, and
Michael Walfish \vspace{.5pc}\\
\fontsize{9.5}{11}\selectfont NYU Department of Computer Science, Courant
Institute \quad\quad\quad $^{\star}$Stony Brook University}

\date{}

\maketitle
\thispagestyle{empty}

\begin{abstract}
Cloud DBs offer strong properties, including serializability, sometimes
called the gold standard database correctness property. But cloud DBs
are complicated black boxes, running in a different administrative domain from
their clients; thus, clients might like to know whether the DBs are
meeting their contract. A core difficulty is that the underlying problem
here, namely \emph{verifying serializability}, is
NP-complete~\cite{papadimitriou79serializability}. Nevertheless, we
hypothesize that on real-world workloads, verifying serializability is
tractable, and we treat the question as a systems problem, for the first
time. We build \sys, which tames the underlying search problem by
blending a new encoding of the problem, hardware acceleration, and a
careful choice of a suitable SMT solver. \Sys also introduces
a technique to address the challenge of garbage collection 
in this context. \Sys improves over natural baselines by at least
10$\times$ in the problem size it can handle, while imposing
modest overhead on clients.
\end{abstract}

 \vspace{1ex} %

\section{Introduction and motivation}
\label{s:intro}

A new class of cloud databases has emerged, including products from
    Google~\cite{corbett13spanner,googlespanner,googledatastore},
    Amazon~\cite{verbitski17amazon,amazoneaurora},
    Microsoft Azure~\cite{azurecosmos},
    as well as their ``cloud-native'' open source
    alternatives~\cite{cockroachdb, yugabytedb, faunadb, foundationdb}.
Compared to earlier generations of NoSQL databases (such as
Facebook Cassandra,
  Google Bigtable, and
  Amazon S3),
members of the new class offer the same
scalability, availability, replication, and geodistribution but in addition support a powerful programming construct:
strong ACID transactions.  By ``strong'', we mean that the promised
isolation contract is \emph{serializability}~\cite{papadimitriou79serializability,
  bernstein87concurrency}:
all transactions appear to execute in a single,
sequential order.

Serializability is the ``gold standard'' isolation level~\cite{bailis14highly}, 
  and the one that many applications and programmers implicitly assume
  (in the sense that their code is incorrect if the database provides a
  weaker contract)~\cite{warszawski17acidrain}.%
\footnote{Of course, there is something stronger, namely strict
serializability~\cite{papadimitriou79serializability,bernstein79formal}
(ss), which reflects real-time constraints; while the ss property is
easier to \emph{verify}, it's harder to \emph{provide}
with dependable performance~\cite{lim17cicada}. For that reason,
many of these databases %
offer non-strict serializability,
hence our focus on that property.}
As the essential \emph{correctness contract}, serializability is the visible
  part of the entire iceberg (the cloud).
This has to do with how the cloud database is used~\cite{story1,story2}:
  a user (developer or administrator) deploys, for example, web servers
  as database clients.
And, based on whether the observed behaviors from clients are serializable,
  the user can deduce whether the cloud
  database has operated as expected.
In particular, if the database has satisfied serializability throughout
  one's observation, then the user knows  that the database maintains basic
  integrity: each value read derives from a valid write.
It also implies that the database has survived failures (if any)
  during this period.

A user can legitimately wonder 
whether cloud databases in fact
provide the promised contract. 
For one thing, users have no visibility into
a cloud database's internals. %
Any internal corruption---as could happen from misconfiguration, misoperation, compromise, or
adversarial control at any layer of the execution stack---can result in serializability violation.
And for another, one need not
adopt a paranoid stance (``the cloud as malicious adversary'') to
acknowledge
that it is
difficult, as a technical matter, to provide serializability \emph{and}
geo-distribution \emph{and} geo-replication \emph{and} high performance
under various failures~\cite{alquraan2018analysis,zheng2014torturing,ganesan2017redundancy}.
Doing so usually involves a consensus protocol that interacts with an
  atomic commit protocol~\cite{corbett13spanner,kraska13mdcc,
    mahmoud13low}---a complex, and hence potentially bug-prone,
    combination.

As serializability is a critical guarantee, verifying that the
  database is serializable is an important issue. 
Indeed, existing works~\cite{sumner11marathon,sinha10runtime,
  nagar18automated,xu05serializability,hammer08dynamic,zellag14consistency,
  brutschy17serializability}
  can verify serializability and/or other consistency anomalies.
However, these works share a limitation: they require
  ``inside information''---(parts of) the
  internal schedules of the database.
Crucially, such internal schedules are invisible to the clients of cloud databases.
This leads to our question: \emph{how can clients verify the
serializability of a cloud database without inside information?}

On the one hand,
this question has long been known to be intractable:
Papadimitriou proved its NP-completeness 40 years
  ago~\cite{papadimitriou79serializability}.
On the other hand, one of the remarkable aspects in the field of formal
  verification has %
  been the use of heuristics
  to ``solve'' problems whose general form is intractable; recent examples
  include~\cite{hawblitzel14ironclad,hawblitzel15ironfleet,fromherz19verified}.
This owes to major advances in solvers (advanced SAT and SMT
  solvers)~\cite{moskewicz01chaff,goldberg07berkmin,liang16exponential,
    stump02cvc,biere09handbook,de08z3,balyo17sat,bruttomesso08mathsat},
  coupled with an explosion of computing power.
Thus, our guiding intuition is that it ought to be possible to verify
  serializability (without inside information) in many real-world cases.

And so, we are motivated to treat the italicized question as a systems problem,
  for the first time.
Compared to prior works~\cite{setty18proving,xu05serializability,sinha10runtime,lu15existential}~(\S\ref{s:relwork}),
our underlying technical problem is different and computationally harder,
which consists of
(a)~positing an unmodifiable and black-box database, (b)~retaining the
database's throughput and latency, and (c)~checking
serializability, %
rather than a weaker property.

\bigskip

This paper describes a system called \sys. \Sys comprises a third-party
database  (which \sys does not modify); a set of (legacy)
database clients (which \sys modifies to link to a library); one or more
\emph{\trace collectors} that record requests and responses to the
database; and a \emph{verifier}. The \trace collectors periodically send
\emph{\trace{} fragments} to the verifier, which has to determine
whether the observed history is serializable. The deployer of \sys (also, the user of the cloud database)
defines the trust domain which encompasses database clients, collectors,
and the verifier; while the database is untrusted.
Section~\ref{s:setup} further details the setup.
\Sys's verifier solves two main
problems, outlined below.

\emph{1. Efficient witness search~(\S\ref{s:search}).} One can check
serializability by searching for an acyclic graph whose vertices are
transactions and whose edges obey certain constraints;
a constraint specifies that exactly one
of two edges must be in the searched-for graph. From this description, one suspects that a SAT/SMT
  solver~\cite{de08z3,yices,soos09extending,barrett11cvc4}
would be useful.
But complications arise. To begin with, encoding acyclicity in a SAT instance brings overhead~\cite{gebser14answer,
  gebser14sat,janota17quest} (we see this too;~\S\ref{subsec:oneshot}).
Instead, \sys uses a recent SMT solver, MonoSAT~\cite{bayless15sat}, that is
  well-suited to checking graph properties~(\S\ref{s:solving}).
However, even MonoSAT alone
  is too inefficient~(\S\ref{subsec:oneshot}).

To address this issue, \sys reduces the search problem size.
First, \sys introduces a new encoding that exploits common patterns in
real workloads, such as read-modify-write transactions, to
  efficiently infer ordering relationships from a
  \trace~(\S\ref{s:combiningwrites}--\S\ref{s:coalescing}).
(We prove that \sys's encoding is a valid reduction in Appendix~\ref{sec:appxa}.)
Second, \sys uses parallel
hardware (our implementation uses GPUs; \S\ref{s:impl}) to compute
\emph{all-paths reachability} over the known graph edges; then, \sys 
is able to efficiently resolve some of the constraints, by
testing whether a candidate edge would generate a cycle
with an existing path.

\emph{2. Garbage collection and scaling~(\S\ref{s:gc}).} \Sys's verifier
works in rounds. From round-to-round, however, the verifier must trim
history, otherwise verification would become too costly.  The challenge
is that the verifier seemingly needs to retain all history, because
serializability does not respect real-time ordering, so future
transactions can read from values that (in a real-time view) have been
overwritten~(\S\ref{s:truncationhard}).  To solve this problem, clients
issue periodic \emph{fence transactions}~(\S\ref{s:fence}). The fences
impose coarse-grained synchronization, creating a window from which
future reads, if they are to be serializable, are permitted to read.
This allows the verifier to discard transactions prior to the window.

We implement \sys~(\S\ref{s:impl}) and find~(\S\ref{s:eval}) that,
compared to our baselines, \sys delivers at least a 10$\times$
improvement in the problem size it can handle (verifying a history of 10k transactions in
14 seconds), while imposing minor throughput and latency overhead on
clients. End-to-end, on an ongoing basis, \sys can sustainably verify
1k--2.5k txn/sec on the workloads that we experiment with.

\sys's main limitations are: First, given the underlying problem
is NP-complete, theoretically there is no guarantee that \sys can terminate
(though all our experiments finish in reasonable time, \S\ref{s:eval}).
Second, range queries are not natively supported by \sys;
programmers need to add extra meta-data in database schemas
to help check serializability on range queries.

\begin{figure}
\centering
\includegraphics[width=0.38\textwidth]{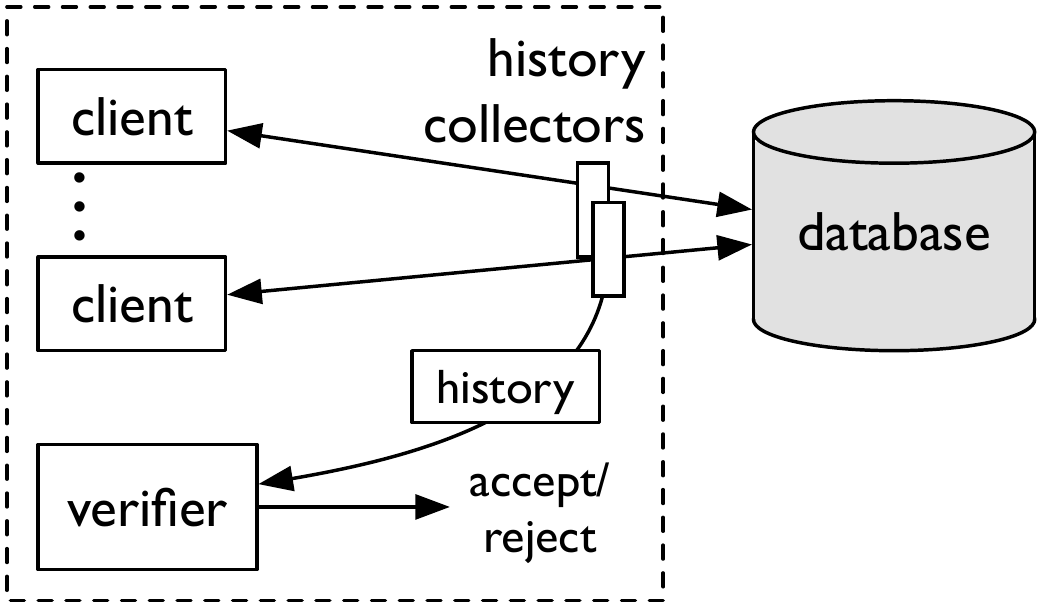}
\caption{Architecture of \sys. The trust domain is framed by dashed lines.}
\label{fig:arch}
\end{figure}
 
\section{Overview and background}
\label{s:background}
\label{s:setup}

Figure~\ref{fig:arch} depicts \sys's high-level architecture.
\textit{Clients} issue requests to a \emph{database} and receive
results. The database is untrusted: the results can be arbitrary.

Each client request is one of five operations: start, commit, abort (which
refer to \emph{transactions}), and read and write (which refer to \emph{keys}).
Each client is single-threaded: it waits to receive the result
  of the current request before issuing the next request.

A set of \emph{\trace collectors} sit between clients and the database,
and capture the requests that clients issue and the (possibly wrong)
results delivered by the database. This capture is a \emph{fragment of a \trace}. A
\emph{\trace} is a set of operations; it is the union of the fragments
from all collectors.

A \emph{verifier} retrieves history fragments from collectors 
and verifies whether the history is
\emph{serializable}; we defined this term loosely in the introduction and will
make it more precise below~(\S\ref{s:searchprelims}).
Verification proceeds in \emph{\rounds{}}; each round consists of a
witness search, the input to which is logically the output of the
previous round and new history fragments. 
Clients, history collectors, and the verifier are trusted.

\sys's architecture is relevant in real-world scenarios.
As an example, an enterprise web application uses a cloud
database for stability, performance, and fault-tolerance.
The end-users of this application
are geo-distributed employees of the enterprise.
To avoid confusion, note that
the employees are users of the application, and the
\emph{clients} here are the web
application servers, as clients of the database.

Database clients (the application) run on the enterprise's hardware (``on-premises'')
while the database runs on an untrusted cloud provider.
The verifier also runs on-premises.
In this setup, collectors can be middleboxes
situated at the edge of the enterprise and can thereby capture
the requests/responses between the clients and the database in the cloud.

The rest of this section defines the core problem more precisely and
gives the starting point for \sys's solution. Section~\ref{s:search}
describes \sys's techniques for a single instance of the problem while
Section~\ref{s:gc} describes the techniques needed to stitch rounds
together.

\subsection{Preliminaries}
\label{s:searchprelims}

First, assume that each value written to the database is unique; thus,
from the \trace, any read (in a transaction) can be associated with the
unique transaction that issued the corresponding write. \Sys discharges
this assumption with logic in the %
\sys client library~(\S\ref{s:impl}).

A \emph{\trace} is a set of read and write operations, each of which is
  associated with a transaction.
Each read operation must read from a particular write operation in the
  history. %
A \emph{\trace is serializable}
if it \emph{matches} a \emph{serial schedule}~\cite{papadimitriou79serializability}.
A \emph{schedule} is a total order of all operations in the \trace.
A \trace and schedule match each other if executing the operations following
  the schedule on a set of single-copy data produces the same read results
  as the \trace. (The write operations are assumed to have empty
    returns so are irrelevant in matching a history and a schedule.)
A \emph{serial} schedule means that the schedule does not have overlapping
transactions.
In addition to a serializable \trace, we also say a \emph{schedule is serializable}
  if the schedule is \emph{equivalent} to a serial schedule---executing the two
  schedules generates the same read results and leaves the data in the
  same final state.

A schedule implies an ordering for every pair of conflicting operations; two
  operations conflict if they are from different transactions and
  at least one is write.
These orderings (all of them) form
a set of \emph{dependencies} among the transactions.
For example, if an operation of a transaction $T_1$ writes a key, and later in
  the schedule, an operation of transaction $T_2$ writes the same key,
  the dependency set contains a dependency denoted as $T_1 \rightarrow T_2$.

From a schedule and its dependency set, one can construct a
\emph{precedence graph} that has a vertex for every transaction in the
schedule and a directed edge for every dependency implied by the
schedule. An important fact is that if the precedence graph is acyclic,
a serial schedule that is equivalent to the original schedule can be
derived, by topologically sorting the precedence graph.

\subsection{Verification problem statement}
\label{subsec:problemstatement}

Based on the immediately preceding fact, the question of \emph{whether a
\trace is serializable} can be converted to \emph{whether the \trace
matches a schedule whose precedence graph is acyclic}. So, the core
problem is to identify such a precedence graph, or assert that none
exists.

Note that this question would be straightforward if the database
revealed its actual schedule (thus ruling out any other possible
schedule): one could construct that schedule's precedence graph, and
test it for acyclicity. Indeed, this is the problem of testing
\emph{conflict-serializability}~\cite{weikum01transactional}. Our
problem, however, is testing 
\emph{view-serializability}~\cite{yannakakis84serializability}.\footnote{Confusingly,
in works targeting conflict serializability, the term ``history''
implies dependency information among conflicting transactions, and
refers to what we call a ``schedule''. Even more confusingly, a database
that \emph{claims to implement} conflict-serializability can, in our
context, be
\emph{tested} only for view-serializability, as the internal
scheduling choices are not exposed.}
In our context, where the database is a black box~(\S\ref{s:intro},
\S\ref{s:background}), we have to (implicitly)
find schedules that match the \trace, and test those schedules'
precedence graphs for acyclicity.  Intuitively, we will conduct this
search by first listing all edges that must exist---for example, a
transaction reads from another's write---and then consider the edges
between every other pair of conflicting transactions (operations) as
\emph{possibilities}.

\begin{figure*}
\includegraphics[width=\textwidth]{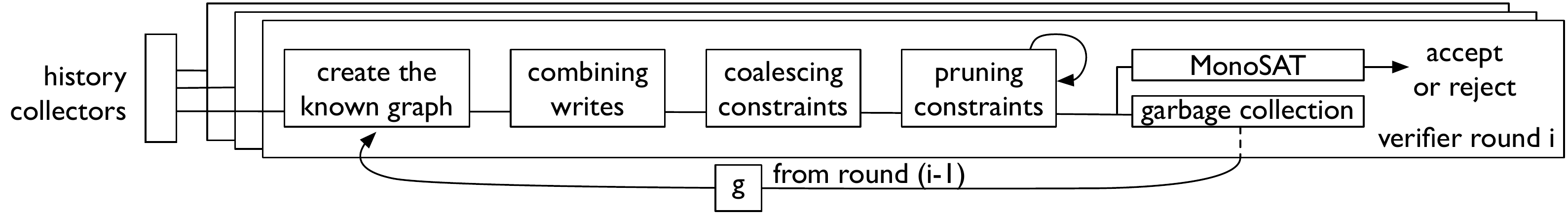}
\caption{The verifier's process, within a round and across rounds.}
\label{fig:veriprocess}
\vspace{-2ex}
\end{figure*}
 
\subsection{Starting point: Intuition and brute force}
\label{subsec:bruteforce}
\label{s:startingpoint}
\label{s:bruteforce}

This section describes a brute-force solution, which serves as the
starting point for \sys and gives intuition.  The approach relies on a
data structure called a \emph{\polyg}~\cite{papadimitriou79serializability},
which captures all possible \pgs when some of the dependencies are
unknown.

In a \polyg, vertices ($V$) are transactions and edges ($E$) are
\wrrels.  A set $C$, which we call \textit{\constraints}, indicates
possible (but unknown) dependencies. Here is an example \polyg:

\vspace{1ex}
\hspace{.75in}\includegraphics[width=0.25\textwidth]{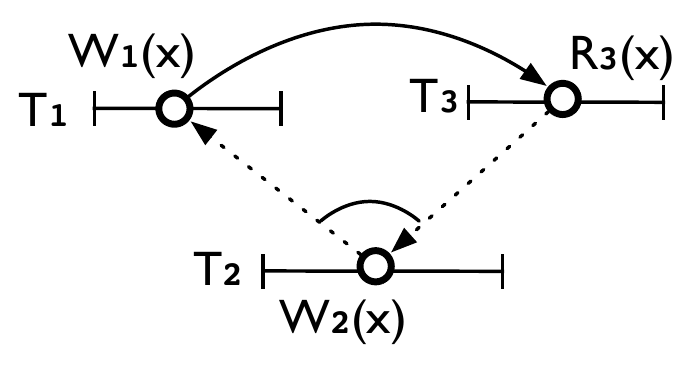}

It has three vertices $V=\{T_1, T_2, T_3\}$, one known edge $E=\{ (T_1, T_3) \}$ from $W_1(x) \wrarrow{x} R_3(x)$,
and one \constraint $\langle\, (T_3,\,T_2),\ (T_2,\,T_1)\, \rangle$
which is shown as two dashed arrows connected by an arc.
This \constraint captures the fact that
$T_3$ cannot happen in between $T_1$ and $T_2$,
because $T_3$ reads $x$ from $T_1$; and $T_2$ which writes $x$
either happens before $T_1$ or after $T_3$.
But it is unknown which option is the truth.

Formally, a \emph{\polyg} $P = (V,\ E,\ C)$ is a directed graph $(V,\ E)$ together with a set 
of \emph{bipaths}, $C$; that is, pairs of edges---not necessarily in
$E$---of the form $\langle (v,\,u),\ (u,\,w) \rangle$ such that $(w,\,v)
\in E$. A bipath of that form can be read as ``either $u$ happened after $v$,
or else $u$ happened before $w$''.

Now, define
\emph{the \polyg $(V,\,E,\,C)$ associated with a \trace}, as
follows~\cite{weikum01transactional}:

\begin{myitemize2}

  \item $V$ are all committed transactions in the \trace

  \item $E = \{ (T_i,\,T_j)\, |\, T_j \textrm{ reads from } T_i \}$;
  that is, $T_i \wrarrow{x} T_j$, for some $x$.

  \item $C = \{ \langle\,(T_j,\, T_k),\  (T_k,\, T_i)\, \rangle \mid
    (T_i \wrarrow{x} T_j)\, \land\,\\ 
    \hspace*{3em}(T_k\textrm{\ writes to $x$}) \land T_k \ne T_i \land T_k \ne T_j\}$.

\end{myitemize2}

The edges in $E$ capture a class of
dependencies~(\S\ref{s:searchprelims}) that are evident from
the \trace, known as WR dependencies (a transaction writes a key, and
another transaction reads the value written to that key).
The third bullet describes how uncertainty is encoded into
constraints. Specifically, for each WR 
dependency in the \trace, all other transactions that write the same key
either happen before the given write or else after the given read.

A \pg is called \textit{compatible} with a \polyg
if: the \pg has the same nodes and known edges in the \polyg, and
the \pg chooses one edge out of each \constraint.
 Formally,
a \pg $(V', E')$ is \emph{compatible} with a \polyg $(V,\,E,\,C)$
if: $V=V'$,
$E \subseteq E'$, and
$\forall \langle e_1,\,e_2 \rangle \in C,\ 
(e_1 \in E' \land e_2 \notin E') \lor (e_1 \notin E' \land e_2
\in E')$.

A crucial fact is: there exists an acyclic \pg that is compatible
 with
the polygraph
 associated to a \trace if and only if that \trace is
serializable~\cite{papadimitriou79serializability,weikum01transactional}.
This yields a brute-force approach for verifying
serializability: first,
     construct a \polyg from a \trace; second, %
     search for a compatible \pg that is acyclic.
However, not only does this approach need to consider $|C|$
binary choices ($2^{|C|}$ possibilities) but also $|C|$ is massive: it is a sum of quadratic
terms, specifically
$\sum_{k \in \mathcal{K}} p_\textit{k} \cdot (q_\textit{k}-1)$,
where $\mathcal{K}$ is the set of keys in the \trace,
and each $p_k,q_k$ are (respectively) the number of reads and writes
of key $k$.

\section{Verifying serializability in \sys}
\label{s:search}
\label{sec:verifyserializable}
\label{subsec:reduce}
\label{s:reduce}
\label{s:ourencoding}

Figure~\ref{fig:veriprocess} depicts the verifier and the major
components of verification. This section covers one round of
verification. As a simplification, assume that the round runs in a
vacuum; Section~\ref{s:gc} discusses how rounds are linked.

\Sys uses an SMT solver geared to graph
properties, specifically MonoSAT~\cite{bayless15sat} (\S\ref{s:solving}).
Yet, despite MonoSAT's power, encoding the problem as in
Section~\ref{s:bruteforce} would generate too much work for
it~(\S\ref{subsec:oneshot}). 

\Sys refines that encoding in several ways. It introduces \emph{write
combining}~(\S\ref{s:combiningwrites}) and
\emph{coalescing}~(\S\ref{s:coalescing}). These techniques are motivated
by common patterns in workloads, and efficiently extract restrictions
(on the search space) that are available in the history.  \Sys's
verifier also does its own inference~(\S\ref{s:pruning}), prior to
invoking the solver. This is motivated by observing that (a)~having
all-pairs reachability information (in the ``known edges'') yields quick
resolution of many constraints, and (b)~computing that information is
amenable to acceleration on parallel hardware such as GPUs (the
computation is iterated matrix multiplication; \S\ref{s:impl}).

Figure~\ref{fig:algocodemain} depicts the algorithm that constructs
\sys's encoding and shows how the techniques combine.  Note that \sys
relies on a generalized notion of \constraints. Whereas previously a
\constraint was a pair of edges, now a \constraint is a pair of
\emph{sets of edges}. Meeting a \constraint $\langle A, B \rangle$ means
including \textit{all} edges in $A$ and excluding \textit{all} in $B$,
or vice versa. More formally, we say that a \pg $(V',E')$ is
\emph{compatible} with a known graph $G=(V,E)$ and generalized
\constraints $C$ if: $V=V'$, $E \subseteq E'$, and $\forall \langle
A,\,B \rangle \in C, (A \subseteq E' \land B \cap E' = \emptyset) \lor
(A \cap E' = \emptyset \land B \subseteq E')$.
  
We prove the validity of \sys's encoding in Appendix~\ref{sec:appxa}.
Specifically we prove that \textit{there exists
an acyclic graph that is compatible with the constraints constructed by
\sys on a given \trace if and only if the \trace is serializable}.

\begin{figure*}[htb!]

\newcommand{\accept}{{\small\textsc{accept}}}
\newcommand{\reject}{{\small\textsc{reject}}}
\newcommand{\mtab}{\hspace{\algorithmicindent}}
\newcommand{\mmtab}{\mtab\mtab}
\newcommand{\mmmmtab}{\mmtab\mmtab}

\footnotesize
\rule{\linewidth}{.08em}

\begin{minipage}[t][][t]{.48\textwidth} 
\begin{algorithmic}[1]
  \Procedure{\ConstructEncoding}{\t{\trace}} \label{li:audit}
    \State $\vg,\,\readfrom,\,\wwpairs  \gets \CreateKnownGraph(\t{\trace})$
    \State $\vconstraints  \gets \GenConstraint(\vg,\,\readfrom,\,\wwpairs)$
    \State $\vconstraints,\,\vg \gets \PruneConstraintsTR(\vconstraints,\,\vg)$ // \S\ref{s:pruning}, executed one or more times
    \State \Return \vconstraints, \vg \label{li:encodingreturn}%
  \EndProcedure

  \State

  \Procedure{\CreateKnownGraph}{\t{\trace}}
    \State $\vg \gets$ empty graph \mmtab\mmtab\mmtab\mtab\,// the konwn graph
    \State $\wwpairs \gets$ map $\{\langle \textrm{Key}, \textrm{Tx}\rangle \to \textrm{Tx}\}$  \,\mtab// consecutive writes \label{li:wwisamap}
    \State // map from a write Tx to a set of read Txs that read this write
    \State $\readfrom \gets$ map $\{\langle \textrm{Key}, \textrm{Tx}\rangle \to
                                    \textrm{Set}\langle\textrm{Tx}\rangle\}$
    \For {transaction \vtx in \t{\trace}} \label{li:loopall}
        \State $\vg.\textrm{Nodes} \pluseq \vtx$
        \For {read operation \rop in \vtx}
            \State $\vg.\textrm{Edges} \pluseq (\rop.\textrm{read\_from\_tx},\ \vtx)$ \mtab// add wr-edge\label{li:addwr}
            \State $\readfrom[\langle \rop.\textrm{key},\,\rop.\textrm{read\_from\_tx} \rangle] \pluseq \vtx$
              \label{li:getreadfrom}
        \EndFor
        \State
        \State // detect RMW (read-modify-write) transactions
        \For {all Keys $\vkey$ that are both read and written by $\vtx$} \label{li:rwsamekey}
           \State \rop $\gets$ the operation in $\vtx$ that reads $\vkey$
                \If {$\wwpairs[\langle \vkey,\,\rop.\textrm{read\_from\_tx} \rangle] \neq \t{null}$} \label{li:detectrmws}
                  \State \reject \label{li:reject1} \mtab  // two \adjacentwrites, not serializable
                \EndIf
                \State $\wwpairs[\langle \vkey,\,\rop.\textrm{read\_from\_tx} \rangle] \gets \vtx$
                  \label{li:getwwpairs}
        \EndFor
    \EndFor
    \State
    \State \Return $\vg,\ \readfrom,\ \wwpairs$
  \EndProcedure

  \State

  \Procedure{\GenConstraint}{$\vg,\,\readfrom,\,\wwpairs$}
    \State // each key maps to set of chains; each chain is an ordered list
    \State $\vperkey \gets$ empty map \{$\textrm{Key} \to \textrm{Set}\langle \textrm{List} \rangle$\}
    \For {transaction \vtx in \vg} \label{li:write2chain1}
        \For {write \wop in \vtx} \label{li:write2chain2}
            \State \vperkey[\wop.key] $\pluseq$ [ \vtx ] \mtab// one-element list 
              \label{li:preparecombine} \label{li:preaprecombine}
        \EndFor
    \EndFor
    \State \label{li:beforecombine}
    \State $\PruneConstraintsRMW(\vperkey,\,\wwpairs)$ \, \, \,// \S\ref{s:combiningwrites}  \label{li:invokecombine}
    \State $\InferRWEdges(\vperkey,\, \readfrom,\, \vg)$ // infer anti-dependency%
    \State
    \State $\vconstraints \gets$ empty set
    \For {$\langle \vkey, \textit{chainset} \rangle$ in $\vperkey$}
    \For {every pair $\{\vwlist_i,\,\vwlist_j\}$ in $\textit{chainset}$} \label{li:chainpairs}
                \State $\vconstraints \pluseq
                \BundleConstraints(\vwlist_i,\,\vwlist_j,\,\vkey,\,\readfrom)$ // \S\ref{s:coalescing}\label{li:genconstraint}
    \EndFor
    \EndFor
  \State
  \State \Return \vconstraints
  \EndProcedure
\algstore{linenum}
\end{algorithmic}

\end{minipage}
\hspace{1ex}
\begin{minipage}[t][][t]{.50\textwidth} 

\begin{algorithmic}[1]
\algrestore{linenum}

  \Procedure{\PruneConstraintsRMW}{$\vperkey,\, \wwpairs$}
    \For {$\langle \vkey,\,\vtx_1,\,\vtx_2 \rangle$ in \wwpairs}
    \label{li:looprmw}  \label{li:isrmw} 
        \State // By construction of \wwpairs, $\vtx_1$ is the write immediately
        \State // preceding $\vtx_2$ on $\vkey$. Thus, we can sequence all writes
        \State // prior to $\vtx_1$ before all writes after $\vtx_2$, as follows:
        \State $\vwlist_1 \gets$ the list in $\vperkey[\vkey]$ whose last elem is $\vtx_1$  \label{li:getchain1}
        \State $\vwlist_2 \gets$ the list in $\vperkey[\vkey]$ whose first elem is $\vtx_2$ 
        \State $\vperkey[\vkey] \mathrel{\setminus}= \{\vwlist_1,\,\vwlist_2\}$ %
         \label{li:rmoldchains}
        \State $\vperkey[\vkey] \pluseq \textrm{concat}(\vwlist_1, \vwlist_2)$ %
          \label{li:concatenateww}
    \EndFor
  \EndProcedure

  \State

    \Procedure{\InferRWEdges}{$\vperkey,\, \readfrom,\, \vg$} \label{li:rwrels}
    \For{$\langle \vkey, \textit{chainset} \rangle$ in \vperkey}
        \For {$\vwlist$ in $\textit{chainset}$}
            \For {$i$ in $[0,\, \textrm{length}(\vwlist) - 2]$}
                \For {$\vrtx$ in $\readfrom[\langle \vkey, \vwlist[i]\rangle]$}
                  \State \textbf{if} ($\vrtx \ne \vwlist[i\textrm{+}1]$):
                  $\vg.\textrm{Edges} \pluseq (\vrtx,\,\vwlist[i\textrm{+}1])$ \label{li:addrw} %
                \EndFor
            \EndFor
        \EndFor
    \EndFor
  \EndProcedure

  \State

  \Procedure{\BundleConstraints}{$\vwlist_1,\,\vwlist_2,\,\vkey,\,\readfrom$} \label{li:coalesce}
    \State $\edgeset_1 \gets \GenChainToChainEdges(\vwlist_1,\,\vwlist_2,\,\vkey,\,\readfrom)$\label{li:c2c1}
    \State $\edgeset_2 \gets \GenChainToChainEdges(\vwlist_2,\,\vwlist_1,\,\vkey,\,\readfrom)$\label{li:c2c2}
    \State \Return $\langle \edgeset_1,\, \edgeset_2 \rangle$ \label{li:conset1set2}
  \EndProcedure

  \State
  \Procedure{\GenChainToChainEdges}{$\vwlist_i,\,\vwlist_j,\,\vkey,\,\readfrom$}
    \If {$\readfrom[\langle \vkey,\,\vwlist_i.\flast \rangle] = \emptyset$}  \label{li:noreadstart}
        \State $\edgeset \gets \{ (\vwlist_i.\flast,\,\vwlist_j.\ffirst) \}$ \label{li:noread}
        \State \Return \edgeset
    \EndIf
    \State
    \State $\edgeset \gets$ empty set
    \For {\vrtx in $\readfrom[\langle \vkey,\,\vwlist_i.\flast \rangle]$} \label{li:wmid}
      \State $\edgeset \pluseq (\vrtx,\, \vwlist_j.\ffirst)$ \label{li:wend}
    \EndFor
    \State \Return \edgeset
  \EndProcedure

  \State

  \Procedure{\PruneConstraintsTR}{$\vconstraints,\,\vg$} \label{li:prunefunc}
  \State // $\vrm$ is the transitive closure (reachability of every two nodes) of $\vg$
  \State $\vrm \gets \TransitiveClosure(\vg)$ \,// standard algorithm; see \cite[Ch.25]{clrs} \label{li:transitiveclosure}
  \For {\t{c} =$\langle \edgeset_1,\,\edgeset_2 \rangle$ in \vconstraints} \label{li:beginprune}
    \If {$\exists (\vtx_i,\, \vtx_j) \in \edgeset_1 \ s.t.\ \vtx_j \rightsquigarrow \vtx_i$ in $\vrm$}
       \label{li:conflict1}
      \State $\vg.\textrm{Edges} \gets \vg.\textrm{Edges} \cup \edgeset_2$  \label{li:addprune1}
      \State $\vconstraints \subeq \t{c}$
    \ElsIf {$\exists (\vtx_i,\, \vtx_j) \in \edgeset_2 \ s.t.\ \vtx_j \rightsquigarrow \vtx_i$ in $\vrm$}
        \label{li:conflict2}
      \State $\vg.\textrm{Edges} \gets \vg.\textrm{Edges} \cup \edgeset_1$ \label{li:addprune2}
      \State $\vconstraints \subeq \t{c}$ \label{li:endprune}
    \EndIf
  \EndFor
  \State \Return \vconstraints,\,\vg
  \EndProcedure
\end{algorithmic}
\end{minipage}
\vspace*{-1ex}
\rule{\linewidth}{.08em}
\caption{ \sys's procedure for converting a \trace into a constraint
satisfaction problem~(\S\ref{subsec:reduce}). After this procedure, \sys feeds the results (a
graph of known edges $G$ and set of constraints $C$) to a constraint solver~(\S\ref{s:solving}),
which searches for a graph that includes the known edges
from $G$, meets the constraints in $C$, and is acyclic. We prove the
algorithm's validity in Appendix~\ref{sec:appxa}.}
\label{fig:algocodemain}
\vspace*{-1ex}
\end{figure*}
 
\subsection{Combining writes}
\label{s:combiningwrites}
\label{s:writecombining}

\Sys exploits the read-modify-write (RMW) pattern, in which a
transaction reads a key and then writes the same key. The pattern is
common in real-world scenarios, for example shopping: in one
transaction, get the number of an item in stock, decrement, and write
back the number.
\Sys uses RMWs to impose order on writes; this reduces the 
orderings that the verification procedure would otherwise have to
consider. Here is an example:

\vspace{1ex}
\includegraphics[width=0.38\textwidth]{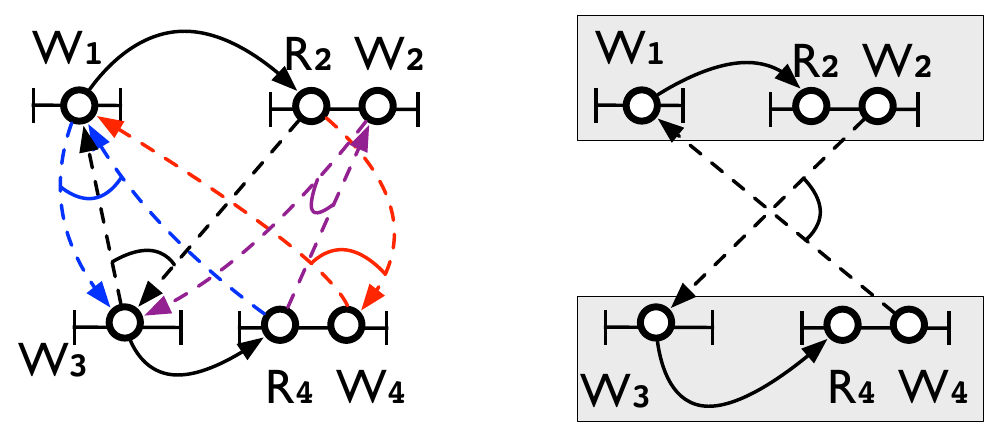}
\vspace{1ex}

There are four transactions, all operating on the same key. Two of the
transactions are RMW, namely $R_2,W_2$ and $R_4,W_4$. On the left is the
basic polygraph~(\S\ref{subsec:bruteforce}); it has four \constraints (each in
a different color), which
are derived from considering WR dependencies. %

\Sys, however, infers
\textit{\chains{}}. A single \chain comprises a \emph{sequence of
transactions whose write operations are consecutive}; in the figure, a
\chain is indicated by a shaded area. Notice that the only ordering
possibilities exist at the granularity of \chains (rather than individual
writes); in the example, the two possibilities of course are
$[W_1,\,W_2] \to [W_3,\,W_4]$ and $[W_3,\,W_4] \to [W_1,\,W_2]$.  This
is a reduction in the possibility space;
for instance, the original version considers the possibility that $W_3$ is
immediately prior to $W_1$ (the upward dashed black arrow), but \sys
``recognizes'' the impossibility of that.

To construct \chains, \sys initializes every write as a one-element
\chain (Figure~\ref{fig:algocodemain}, line~\ref{li:preaprecombine}).
Then, \sys consolidates chains: for each RMW transaction $t$ and the
transaction $t'$ that contains the prior write, \sys concatenates the
chain containing $t'$ and the chain containing
$t$~(lines~\ref{li:getwwpairs} and
\ref{li:looprmw}--\ref{li:concatenateww}).

Note that if a transaction $t$, which is \emph{not} an RMW, reads from a
transaction $u$, then $t$ requires an edge to $u$'s successor (call it
$v$); otherwise, $t$ could appear in the \pg downstream of $v$, which
would mean $t$ actually read from $v$ (or even from a later write), which
does not respect history.  \sys creates the $t \to v$ edge (known as an
\emph{\rwrel} in the literature~\cite{adya99weak}) in \InferRWEdges
(Figure~\ref{fig:algocodemain}, line~\ref{li:rwrels}).

\subsection{Coalescing \constraints}
\label{s:coalescing}

This technique exploits the fact that, in many real-world workloads,
there are far more reads than writes. At a high level, \sys combines all 
reads that read-from the same write. We give an example and then generalize.

\vspace{1ex}

\includegraphics[width=0.4\textwidth]{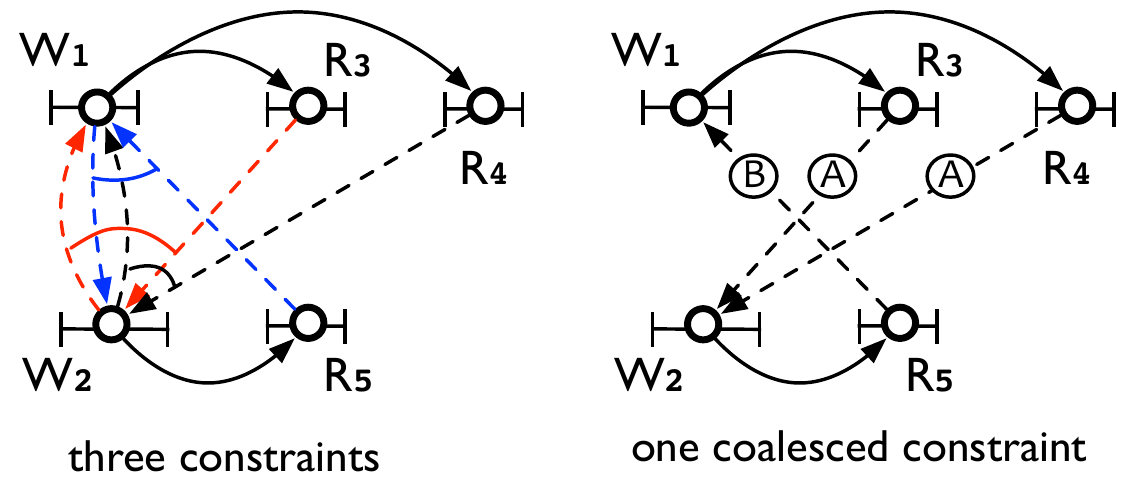}

\noindent In the above figure, there are five single-operation
transactions, to the same key.
On the left is the basic \polyg~(\S\ref{subsec:bruteforce}), which contains three \constraints;
each is in a different color.
Notice that all three \constraints involve the question:
which write happened first, $W_1$ or $W_2$?

One can represent the possibilities as a %
\constraint $\langle A',\,B'\rangle$
where 
$A'=\{(W_1,W_2),(R_3,\,W_2),\,(R_4,\,W_2)\}$ and
$B'=\{(W_2,W_1),(R_5,\,W_1)\}$.
In fact,
\sys does not include $(W_1,W_2)$ 
because there is a known edge $(W_1,R_3)$, which, together with
$(R_3,W_2)$ in $A'$, implies the ordering $W_1 \to R_3 \to W_2$, so there
is no need to include $(W_1,W_2)$. Likewise, \sys does not include $(W_2,W_1)$
on the basis of the known edge $(W_2,R_5)$. So \sys includes the \constraint 
$\langle A,B \rangle=
    \langle \{(R_3,\,W_2),\,(R_4,\,W_2)\},
    \{(R_5,\,W_1)\}
    \rangle$
in the figure.

To construct \constraints using the above reductions, \sys does the
following. Whereas the brute-force approach uses all reads and their
prior writes~(\S\ref{subsec:bruteforce}), \sys considers particular
\emph{pairs of writes}, and creates \constraints from these writes and
their following reads. The particular pairs of writes are the first and
last writes from all pairs of chains pertaining to that key.
In more detail, given two chains, $\vwlist_i,\vwlist_j$, \sys constructs
a \constraint $c$ by
(i)~creating a set of edges $\textit{ES}_1$ that point from reads of
$\vwlist_i.\flast$ to $\vwlist_j.\ffirst$
(Figure~\ref{fig:algocodemain}, lines~\ref{li:wmid}--\ref{li:wend});
this is why \sys does not include the $(W_1,W_2)$ edge above. If there
are no such reads, $\textit{ES}_1$ is $\vwlist_i.\flast \to
\vwlist_j.\ffirst$ (Figure~\ref{fig:algocodemain},
line~\ref{li:noread}); (ii)~building another edge set $\textit{ES}_2$
that is the other way around (reads of $\vwlist_j.\flast$ point to
$\vwlist_i.\ffirst$, etc.), and (iii)~setting $c$ to be $\langle
\textit{ES}_1, \textit{ES}_2\rangle$ (Figure~\ref{fig:algocodemain},
line~\ref{li:conset1set2}).

\subsection{Pruning \constraints}
\label{s:pruning}

Our final technique leverages the information that is encoded in paths
in the known graph.
This technique culls irrelevant possibilities en masse (\S\ref{subsec:oneshot}).
The underlying logic of the
technique is almost trivial. The interesting
aspect here is that the technique is enabled by a design decision to
accelerate the computation of reachability on parallel
hardware~(\S\ref{s:impl} and
Figure~\ref{fig:algocodemain}, line~\ref{li:transitiveclosure}); this
can be done since the computation is iterated (Boolean) matrix
multiplication.
Here is an example:

\begin{center}
\vspace{-2ex}
\includegraphics[width=0.18\textwidth]{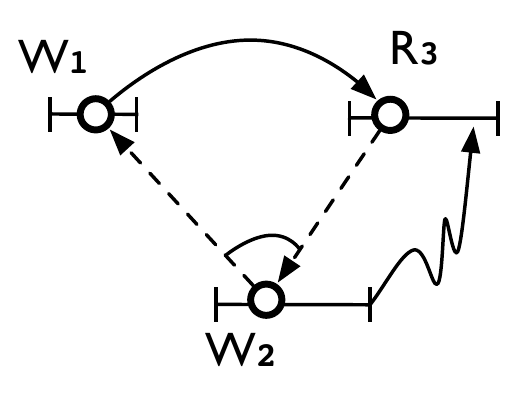}
\vspace{-3ex}
\end{center}

The \constraint is $\langle (R_3,W_2),(W_2,W_1)\rangle$. Having
precomputed reachability, \sys knows that the first choice cannot hold,
as it creates a cycle with the path $W_2 \rightsquigarrow R_3$; \sys
thereby concludes that the second choice holds.  Generalizing, if \sys
determines that an edge in a \constraint generates a
cycle, \sys throws away both components of the entire 
\constraint and adds all the other edges to the known graph
(Figure~\ref{fig:algocodemain},
lines~\ref{li:beginprune}--\ref{li:endprune}).
In fact, \sys does pruning multiple times, if
necessary~(\S\ref{s:impl}).

\subsection{Solving}
\label{s:solving}

The remaining step is to search for an acyclic graph that is compatible
with the known graph and \constraints, as computed in
Figure~\ref{fig:algocodemain}.  \Sys does this by leveraging a
constraint solver. However, traditional solvers do not perform well on
this task because encoding the acyclicity of a graph as a set of SAT formulas is expensive
(a claim by Janota et al.~\cite{janota17quest}, which we also observed,
using their acyclicity encodings on Z3~\cite{de08z3}; \S\ref{subsec:oneshot}).

\Sys instead leverages MonoSAT, which is a particular kind of SMT
solver~\cite{biere09handbook} that includes SAT modulo
\emph{monotonic} theories~\cite{bayless15sat}. This solver efficiently
encodes and checks graph properties, such as acyclicity.

\Sys represents a verification problem instance (a graph $G$ and
\constraints $C$) as follows. 
\Sys creates a Boolean variable $E_{(i,j)}$ for each
edge: True means the $i$th node has an edge to the $j$th
node; False means there is no such edge. \Sys sets all the
edges in $G$ to be True.
For the constraints $C$, recall that each \constraint $\langle A,
B\rangle$ is a pair of sets of edges, and represents a mutually
exclusive choice to include either all edges in $A$ or else all edges in
$B$. \Sys encodes this in the natural way:
$
     ((\forall e_a \in A, e_a) \land (\forall e_b \in B, \lnot e_b)) 
\lor\ ((\forall e_a \in A, \lnot e_a) \land (\forall e_b \in B, e_b)).
$
Finally, \sys enforces the acyclicity of the searched-for compatible graph 
(whose candidate edges are given by the known edges and the constrained
edge variables)
by invoking a primitive provided by the solver.

\heading{\sys vs. MonoSAT.} One might ask: if \sys's encoding makes MonoSAT faster, why use
MonoSAT? Can we take the domain knowledge further?
Indeed, in the limiting case, \sys could
re-implement the solver! However, 
MonoSAT, as an SMT solver, seamlessly leverages many
prior optimizations. %
One way to think about
the decomposition of function in \sys is that \sys's preprocessing
exploits some of the structure created by the problem of verifying
serializability, %
whereas the solver is exploiting residual structure
common to many graph problems. %

\section{Garbage collection and scaling}
\label{s:gc}

\Sys verifies periodically, in rounds. There are two motivations for
rounds.
First, new history is continually produced, of course.
Second, there are limits on the maximum problem size
(in terms of number of transactions) that the verifier can
handle~(\S\ref{s:evalscaling}); breaking the task into
rounds keeps each solving task manageable.

In the first round, a verifier starts with nothing and creates a graph
from \textsc{\CreateKnownGraph}, then does verification. After that, the
verifier receives more client \traces; it reuses the graph from the last
round (the $\vg$ in \textsc{\ConstructEncoding}, Figure~\ref{fig:algocodemain}, line~\ref{li:encodingreturn}),
and adds new nodes and
edges to it from the new \trace fragments
received~(Figure~\ref{fig:veriprocess}).

The technical problem is to keep the input to verification bounded.
So the question \sys must answer is: which transactions
can be deleted safely from \trace? Below, we describe the
challenge~(\S\ref{s:truncationhard}), the core mechanism of fence
transactions~(\S\ref{s:fence}), and how the verifier deletes
safely~(\S\ref{s:safetruncation}).
Due to space restrictions, we only describe the general rules and insights.
A complete specification and correctness proof are
  in Appendix~\ref{sec:appxb}.

\subsection{The challenge}
\label{s:truncationhard}

The core challenge is that past transactions can be relevant to future
verifications, even when those transactions' writes have been
overwritten. Here is an example:

{
\centering
\includegraphics[width=0.47\textwidth]{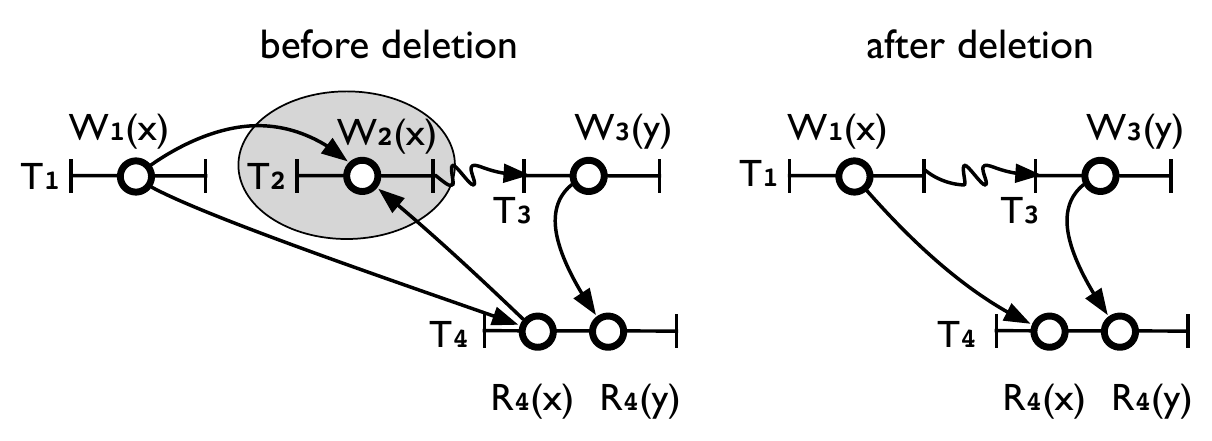}
}

\noindent Suppose a verifier saw three transactions ($T_1,\,T_2,\,T_3$) and wanted
to remove $T_2$ (the shaded transaction) from consideration in future
verification rounds. Later, the verifier observes a new transaction
$T_4$ that violates serializability by reading from $T_1$ and $T_3$.  To
see the violation, notice that $T_2$ is logically subsequent to $T_4$,
which generates a cycle ($T_4 \rwarrowx T_2 \rightsquigarrow T_3
\wrarrowx T_4$). Yet, if we remove $T_2$, there is no cycle.  Hence,
removing $T_2$ is not safe: future verifications would fail to detect
certain kinds of serializability violations.

Note that this does not require malicious or exotic behavior from the
database. For example, consider an underlying database that uses
multi-version values and is geo-replicated: a client can retrieve a
stale version from a local replica. 

Finally, we note that if the database told the verifier when values are
permanently overwritten, the verifier could use this information to
delete safely~\cite{hadzilacos89deleting,farzan08monitoring}. But in
our setup~(\S\ref{s:setup}), the verifier does not get that information.

\subsection{Fence transactions and epochs}
\label{s:fence}
\label{s:guarantee}

\Sys addresses this challenge by introducing \emph{fence transactions} that
  impose a coarse-grained ordering on all transactions; the verifier can
  then discard ``old'' transactions suggested by fence transactions.
A fence transaction is a transaction that reads-and-writes a single
key named ``fence'' (a dedicated key that is used by fence transactions only).
Each client issues fence transactions periodically (for example,
  every 20 transactions).

The fence transactions are designed to divide transactions into
  different \emph{epochs} in the serial schedule.
What prevents the database from defeating the
point of fences by placing all of the fence transactions at the
beginning of a notional serial schedule? The answer is that \sys
\emph{requires} that the database's serialization order not violate the order
of transactions issued by a given client (which, recall, are
single-threaded and block; \S\ref{s:setup}).  Production databases
are supposed to respect this requirement; doing otherwise would violate
causality.  With this, the epoch ordering is naturally
intertwined with the rest of the workload.

Given the preceding requirement, the verifier adds ``client-order
  edges'' to the set of known edges in $g$ (the verifier knows the client order from
  the history collector). 
The verifier also assigns an \emph{epoch number} to each transaction. 
To do so, the verifier traverses the known graph~($g$), locates all the fence
  transactions, chains them into a list based on RMW relation~(\S\ref{s:reduce}),
  and assigns their position in the list as their epoch numbers.
Then, the verifier scans the graph again, and for each normal transaction on a client
  that is between fences with epoch $i$ and epoch $j$ ($j > i$), the verifier
  assigns the normal transaction with an epoch number $j-1$.

During the scan, assume the largest epoch number that has been seen or
  surpassed by every client is $epoch_{agree}$, then we have the following
  guarantee.

\textbf{Guarantee}.
For any transaction $T_i$ whose epoch %
  $\le (epoch_{agree}-2)$, and for any transaction (including future ones)
  $T_j$ whose epoch $\ge epoch_{agree}$, the known graph $g$ contains
  a path $T_i \rightsquigarrow T_j$.

To see why the guarantee holds, consider the problem in three parts.
First, for the fence transaction with epoch number $epoch_{agree}$ (denoted as
  $F_{ea}$), $g$ must have a path $F_{ea} \rightsquigarrow T_j$.
Second, for the fence transaction with epoch number $(epoch_{agree}-1)$ (denoted as
$F_{ea-1}$), $g$ must have a path as $T_i \rightsquigarrow F_{ea-1}$.
Third, $F_{ea-1} \rightarrow F_{ea}$ in $g$.

The guarantee suggests that no future transaction (with epoch
  $\ge epoch_{agree}$) can be a direct predecessor of such $T_i$, otherwise
  a cycle will appear in the \polyg. 
We can extend this property to use in garbage collection.
In particular, if all predecessors of $T_i$ have epoch number
  $\le (epoch_{agree}-2)$, we call $T_i$ a \emph{frozen} transaction,
  referring to the fact that no future transaction can be its (transitive)
  predecessor.

\subsection{Safe garbage collection}
\label{s:safetruncation}

\sys's garbage collection algorithm targets frozen transactions---as
  they are guaranteed to be no descendants of future transactions.
Of all frozen transactions, the verifier needs to keep those which have the
  most recent writes to some key (because they might be read by future transactions). 
If there are multiple writes to the same key and the verifier cannot distinguish
  which is the most recent one, the verifier keeps them all. 
Meanwhile, if a future transaction reads from a deleted transaction (which is a serializability violation---stale read),
  the verifier detects this (the verifier maintains tombstones for the deleted transaction ids)
  and rejects the history.

One would think the above approach is enough, as we did during developing
  the garbage collection algorithm. %
However, this turns out to be insufficient,
which we illustrate using an example below.

{
  \vspace{1ex}
  \centering
  \includegraphics[width=0.42\textwidth]{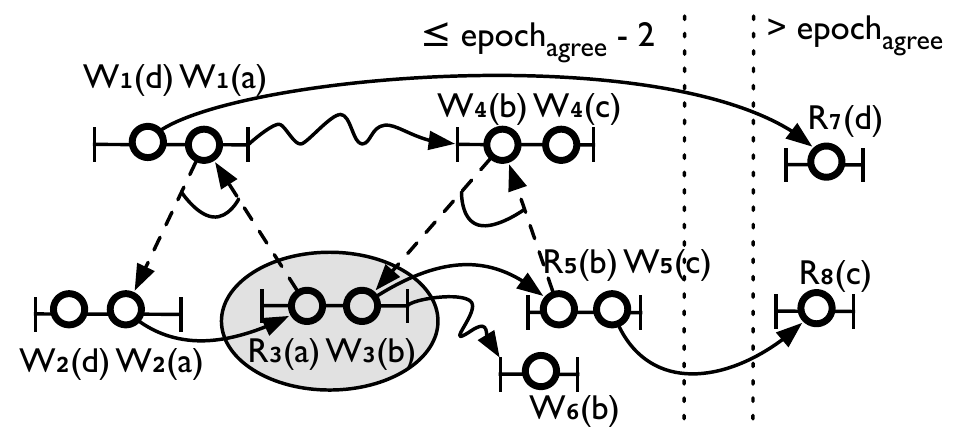}
}

In this example, the shaded transaction ($T_3$; transaction ids indicated by
  operation subscripts)
is frozen and is not the most recent write to any key. %
However, with the two future transactions ($T_7$ and $T_8$), deleting the shaded
  transaction results in failing to detect cycles in the \polyg.

To see why, consider operations on key $c$: $W_4(c)$, $W_5(c)$, and $R_8(c)$.
By the epoch guarantee (\S\ref{s:guarantee}), both $T_4$ and $T_5$
happen before $T_8$. Plus, $R_8(c)$ reads from $W_5(c)$,
hence $W_4(c)$ must happen before $W_5(c)$ (otherwise, $R_8(c)$ should have read from $W_4(c)$).
In which case, the constraint $\langle (T_5,  T_4),\, (T_4, T_3)\rangle$ is solved
($T_5 \to T_4$ conflicts with the fact that $W_4(c)$ happens before $W_5(c)$; hence, $T_4 \to T_3$ is chosen).
Similarly, because of $R_7(d)$, the other constraint is solved and $T_3 \to T_1$.
With these two solved constraints, there is a cycle ($T_1 \rightsquigarrow T_4 \to T_3 \to T_1$).
Yet, if the verifier deletes $T_3$, such cycle would be undetected.

The reason for the prior undetected cycle is that the future transaction may ``finalize'' some
  constraints from the past, causing cycles whereas in the past the constraints
  were ``chosen'' in a different way.
To prevent cases like this, \sys's verifier keeps
  transactions that are involved in any potentially cyclic constraints.

\section{Implementation}
\label{s:impl}

\begin{figure}
\footnotesize
\begin{tabular*}{\columnwidth}{@{\extracolsep{\fill}} l  l @{}}
\toprule
\Sys component     &  LOC written/changed \\
\midrule
\sys client library & \\
\hspace{2ex} \trace recording&   620 lines of Java\\
\hspace{2ex} database adapters   &   900 lines of Java\\
\sys verifier &\\
\hspace{2ex} data structures and algorithms & 2k lines of Java\\
\hspace{2ex} GPU optimizations & 550 lines of CUDA/C++ \\
\hspace{2ex} \trace parser and others & 1.2k lines of Java\\
\bottomrule
\end{tabular*}
\caption{Components of \sys implementation.}
\label{fig:impl}
\end{figure}

The components of \sys's implementation are listed in
Figure~\ref{fig:impl}.
Our implementation includes a client library and a verifier.
\Sys's client library wraps other database
libraries: JDBC, Google Datastore library, and RocksJava.
It enforces the assumption of uniquely written
values~(\S\ref{s:searchprelims}), by adding a unique id to a client's
writes, and stripping them out of reads.  It also issues fence
transactions~(\S\ref{s:fence}). %
Finally, in our current implementation, we simulate history
collection~(\S\ref{s:setup}) by collecting histories in this library;
future work is to move this function to a proxy.

For the verifier, we discuss two aspects of pruning~(\S\ref{s:pruning}).
First, the verifier iterates the pruning logic within a round, stopping
when either it finds nothing more to prune or else when it reaches a
configurable maximum number of iterations (to bound the verifier's
work); a better implementation would stop when the cost of the marginal
pruning iteration exceeds the improvement in the solver's running
time brought by this iteration.

The second aspect is GPU acceleration. Recall that pruning 
works by computing the transitive closure of the known edges~(Fig.~\ref{fig:algocodemain}, line~\ref{li:transitiveclosure}). \Sys
uses the standard algorithm: repeated squaring of the Boolean
adjacency matrix~\cite[Ch.25]{clrs} as long as the matrix keeps
changing, up to $\log |V|$ matrix multiplications.
($\log{|V|}$ is the worst
case and occurs when two nodes are connected by a ($\geq |V|/2+1$)-step
path; at least in our experiments, this case does not arise much.)
The execution platform 
is cuBLAS~\cite{cublas} (a dense linear algebra library on GPUs)
and cuSPARSE~\cite{cusparse}
(a sparse linear algebra library on GPUs),
which contain matrix multiplication routines.

\Sys includes several optimizations. It invokes a specialized
routine for triangular matrix multiplication.
(\sys first tests the graph for acyclicity, and then indexes the
vertices according to a topological sort, creating a triangular matrix.)
\sys also exploits sparse matrix multiplication (cuSPARSE), and moves to 
ordinary (dense) matrix multiplication when the density of the matrix
exceeds a threshold
(chosen to be $\geq5$\% of the matrix elements are
non-zero, the empirical cross-over point that we observed).

Whenever \sys's verifier detects a serializable violation, it creates a
certificate with problematic transactions.
The problematic transactions are either a cycle in
the known graph detected by \sys's algorithm, or
a minimal unsatisfiable core (a set of unsatisfiable
clauses that translates to problematic transactions)
produced by the SMT solver.
\section{Experimental evaluation}
\label{s:eval}

We answer three questions:

\begin{myitemize2}

\item What are the verifier's costs and limits, and how do these compare
to baselines?

\item What is the verifier's end-to-end, round-to-round
\emph{sustainable capacity}? This determines the offered load (on the
actual database) that the verifier can support.

\item How much runtime overhead (in terms of throughput and latency) does \sys
impose for clients? And what are \sys's storage and network overheads?

\end{myitemize2}

\heading{Benchmarks and workloads.}
We use four benchmarks:

\begin{myitemize2}

\item \textit{TPC-C}~\cite{tpcc} is a standard.
A warehouse has 10 districts with 30k customers.
There are five types of transactions (frequencies in 
parentheses): new order (45\%), payment (43\%), order status (4\%), delivery
(4\%), and stock level (4\%).
In our experiments, each client randomly chooses a warehouse
and a district, and issues a transaction based on the frequencies above.

\item \textit{\CF{\twitter}}~\cite{twitter} is a simple clone of Twitter,
according to Twitter's own description~\cite{twitter}.
It allows users to tweet a new post, follow/unfollow other users,
show a timeline (the latest tweets from followed users).
Our experiments include a thousand users. Each user tweets 140-word posts and
follows/unfollows other users based on Zipfian distribution
($\alpha=100$).

\item \textit{\CF{\rubis}}~\cite{rubis,amza02specification},
simulates bidding systems like eBay~\cite{rubis}.
Users can register accounts, register items,
bid for items, and comment on items.
We initialize the market with 20k users and 200k items.

\item \textit{BlindW} is a microbenchmark we wrote to demonstrate \sys's
performance in extreme scenarios. It
creates a set of keys,
and runs random read-only and write-only transactions on them.
In our experiments, every transaction has eight operations, and
there are 10k keys in total.
This benchmark has two variants:
(1) \textit{BlindW-RM} represents a read-mostly workload that contains 90\% read-only
transactions; and (2) \textit{BlindW-RW} represents a read-write
workload, evenly divided between read-only and write-only transactions.

\end{myitemize2}

\noindent\textbf{Databases and setup.}
We evaluate \sys on Google Cloud Datastore~\cite{googledatastore},
PostgreSQL~\cite{postgresql,ports12serializable}, and RocksDB~\cite{rocksdb,dong17optimizing}.
They represent three database environments---cloud, local, and
co-located.
In our experimental setup, clients interact with Google Cloud Datastore
through the wide-area Internet,
and connect to a local PostgreSQL server through a local 1Gbps network.

In the cloud and local database setups,
clients run on two machines with a 3.3GHz Intel i5-6600 (4-core) CPU,
16GB memory, a 250GB SSD, and Ubuntu 16.04.
In the local database setup, a PostgreSQL server runs on a machine
with a 3.8GHz Intel Xeon E5-1630 (8-core) CPU, 32GB memory,
a 1TB disk, and Ubuntu 16.04.
In the co-located setup, the same machine hosts the 
client threads and RocksDB threads, which all run in the same process.
We use a \textit{p3.2xlarge} Amazon EC2 instance as the verifier,
with an NVIDIA Tesla V100 GPU, a 8-core CPU, and 64GB memory.
\subsection{\CF{\oneshot}}
\label{subsec:oneshot}

In this section, we consider ``\oneshot'',
the original serializability verification problem:
a verifier gets a history and decides whether
that history is serializable.
In our setup, clients record \traces fragments and store them as files;
a verifier reads them from the local file system.
In this section, the database is RocksDB (PostgreSQL gives similar results;
Google Cloud Datastore limits the throughput for a fresh database instance
which causes some time-outs).

\heading{Baselines.}
We have two baselines:

\begin{myitemize2}

\item Z3~\cite{de08z3}: %
  we encode the serializability verification problem into a set of SAT formulas,
  where edges are Boolean variables.
  We use \emph{binary labeling}~\cite{janota17quest} to express
  acyclicity, requiring $\Theta(|V|^2)$ SAT formulas.

\item MonoSAT~\cite{bayless15sat} (with the ``brute force'' encoding):
we implement the original \polyg~(\S\ref{s:bruteforce}), directly encode
the \constraints (without the techniques of~\S\ref{s:search}), and feed
them to MonoSAT.

\end{myitemize2}

\begin{figure}
    \centering
    \includegraphics[width=0.9\linewidth]{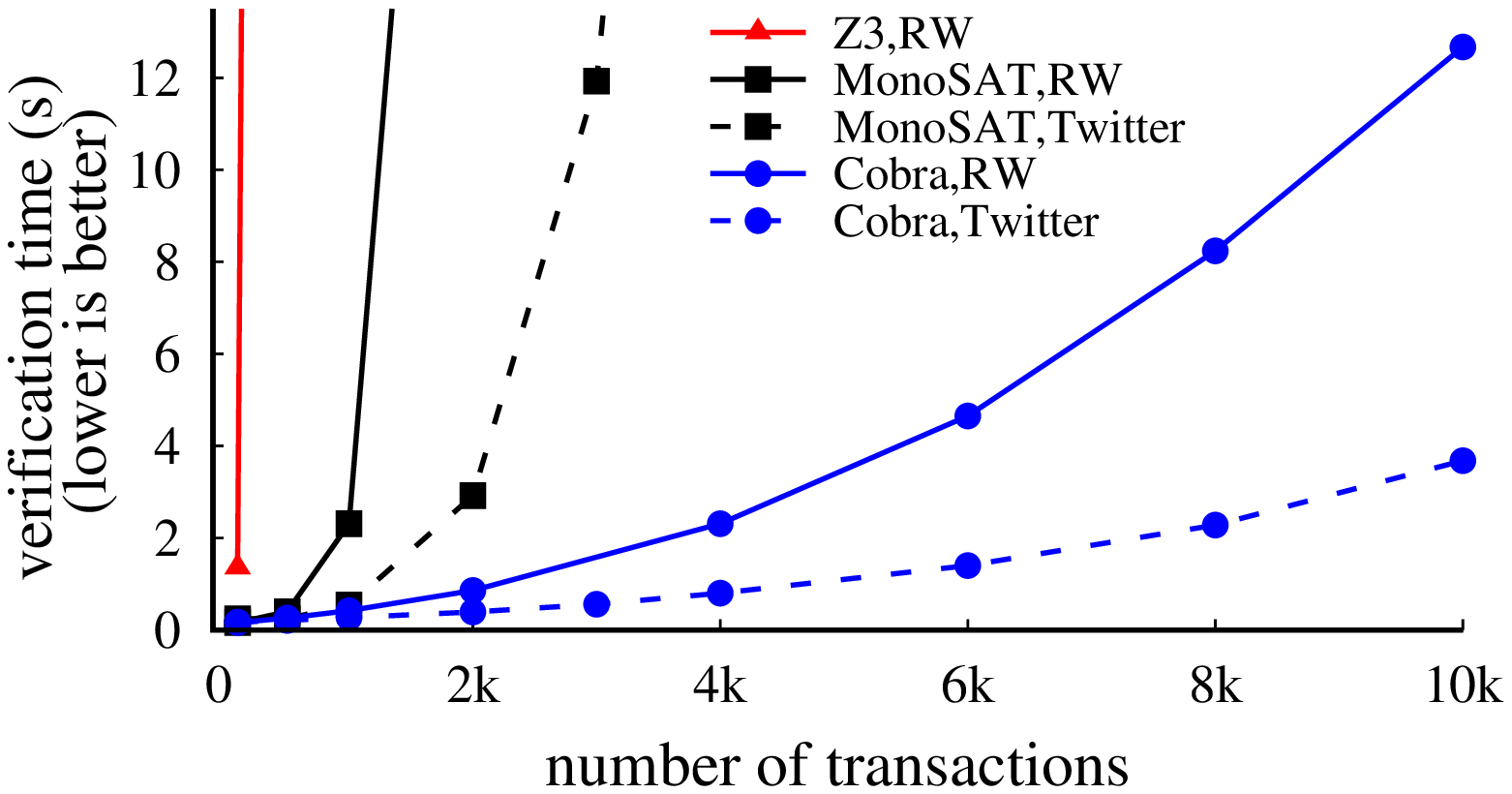}
    \caption{\Sys's running time is faster than MonoSAT's and Z3's on
    the BlindW-RW workload (solid lines) 
    and the Twitter workload (dashed lines). The same holds on the other benchmarks (not
    depicted). Verification runtime grows superlinearly.}
    \label{fig:timevsnumtxn}
    \vspace{-3ex}
\end{figure}

\heading{Verification runtime vs. number of transactions.}
We compare \sys to other baselines, on the various
workloads. There are 24 clients. We vary the total number of
transactions in the workload, and measure the total verification time.
Figure~\ref{fig:timevsnumtxn} depicts the results on two benchmarks.
On all five benchmarks,
\sys does better than
MonoSAT which does better than Z3.\footnote{\label{fn:holycow}As a
special case, there is, for TPC-C, an alternative
that beats MonoSAT and Z3
and has the same performance as \sys. Namely, add edges that be inferred from RMW
operations in history to a candidate graph (without constraints, and so
missing a lot of dependency information), topologically sort it, and
check whether the result matches history; if not, repeat. This process
has even worse order complexity than the one in \S\ref{s:bruteforce},
but it works for TPC-C because that workload has \emph{only} RMW
transactions, and thus the candidate graph \emph{is} (luckily) a
precedence graph.}

\begin{figure}
    \centering
    \includegraphics[width=\linewidth]{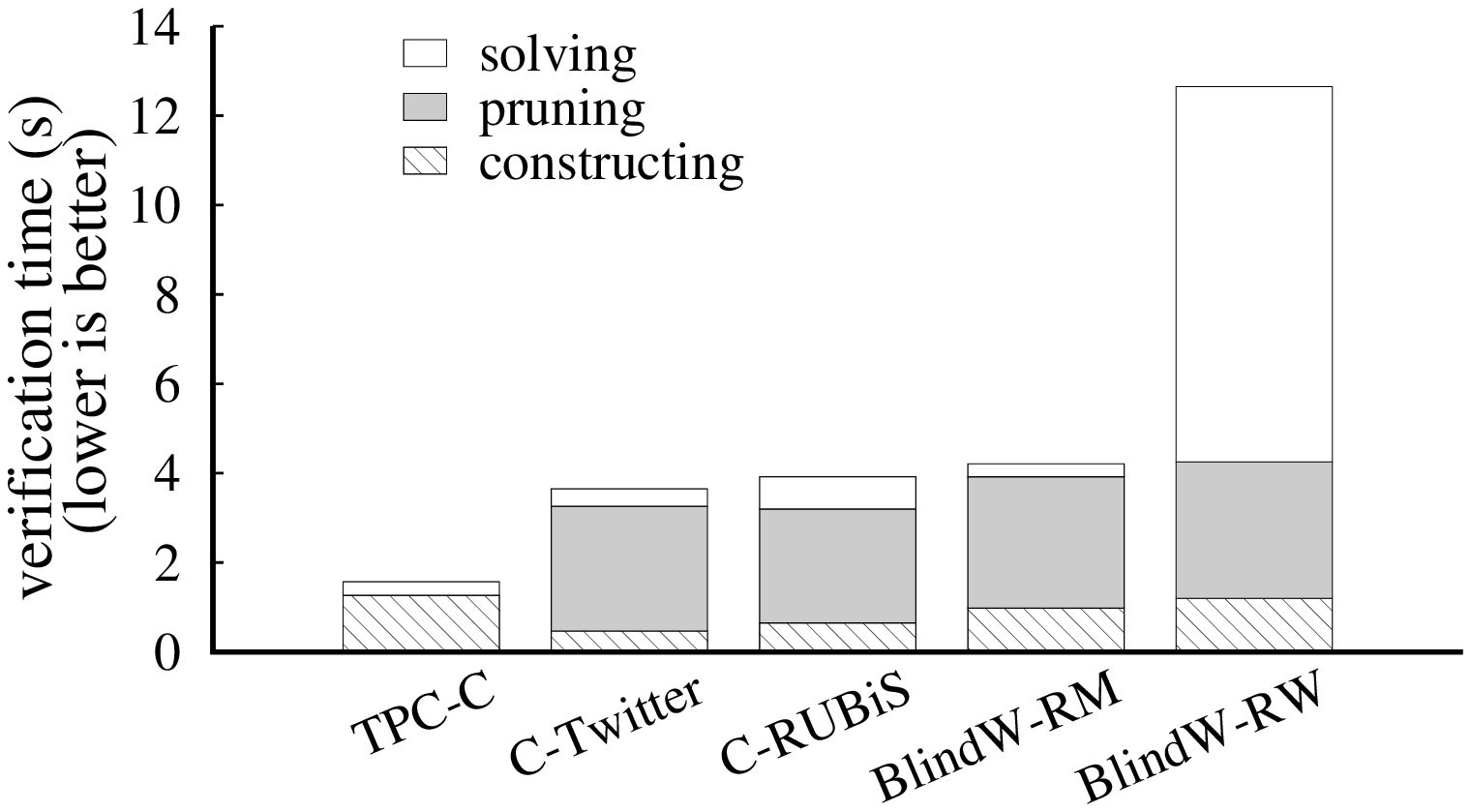}
    \caption{Decomposition of \sys runtime, on 10k-transaction workloads. 
    In benchmarks with RMWs only (the left one), there is no constraint, so
    \sys doesn't do pruning (see also footnote~\ref{fn:holycow},
    page~\pageref{fn:holycow});
    in benchmarks with many reads and RMWs (the middle three),
    the dominant component is pruning not solving, because \sys's own logic
    can identify concrete dependencies; in benchmarks with many blind writes
    (the last one), solving is a much larger contributor because \sys is
    not able to eliminate as many constraints.}
    \label{fig:veribreakdown}
\end{figure}

\heading{Detecting serializability violations.}
In order to investigate \sys's performance on an unsatisfiable instance: does
it trigger an exhaustive search, at least on the real-world workloads we found?
We evaluate \sys on five real-world workloads that are known to have
serializability violations. \sys detects them in reasonable time.
Figure~\ref{fig:violations} shows the results.

\begin{figure}
\footnotesize
\begin{tabular*}{\columnwidth}{@{\extracolsep{\fill}} l @{\hskip 3pt}l @{\hskip 3pt}l l@{}}
\toprule
Violation & Database & \#Txns & Time  \\
\midrule
G2-anomaly~\cite{yuga-g2} & YugaByteDB 1.3.1.0 & 37.2k & 66.3s\\
Disappear writes~\cite{yuga-disw} & YugaByteDB 1.1.10.0 & 2.8k &  5.0s\\
G2-anomaly~\cite{cock-g2} & CockroachDB-beta 20160829 & 446 & 1.0s\\
Read uncommitted~\cite{cock-blog} & CockroachDB 2.1 & 20$\star$ & 1.0s\\
Read skew~\cite{fauna-page} & FaunaDB 2.5.4 & 8.2k & 11.4s\\
\bottomrule
\end{tabular*}
\caption{Serializability violations that \sys checks.
``Violation'' describes  the phenomena that clients experience.
``Database'' is the database (with version number) that generates the violation.
``\#Txns'' is the size of the violation history.
``Time'' is the runtime for \sys to detect such violation.\\
$\star$ The bug report only provides the history snippet that violates
serializability.}
\label{fig:violations}
\vspace{2ex}
\end{figure}

\heading{Decomposition of \sys's verification runtime.} We measure the
wall clock time of \sys's verification on our setup, broken into three
stages: \emph{constructing}, which includes creating the graph of known
edges, combining writes, and creating
\constraints~(\S\ref{s:combiningwrites}--\S\ref{s:coalescing});
\emph{pruning}~(\S\ref{s:pruning}), which includes the time taken by the
GPU; and \emph{solving}~(\S\ref{s:solving}), which includes the time
spent within MonoSAT. We experiment with all benchmarks, with 10k
transactions. Figure~\ref{fig:veribreakdown} depicts the
results.

\begin{figure}
    \centering
    \includegraphics[width=\linewidth]{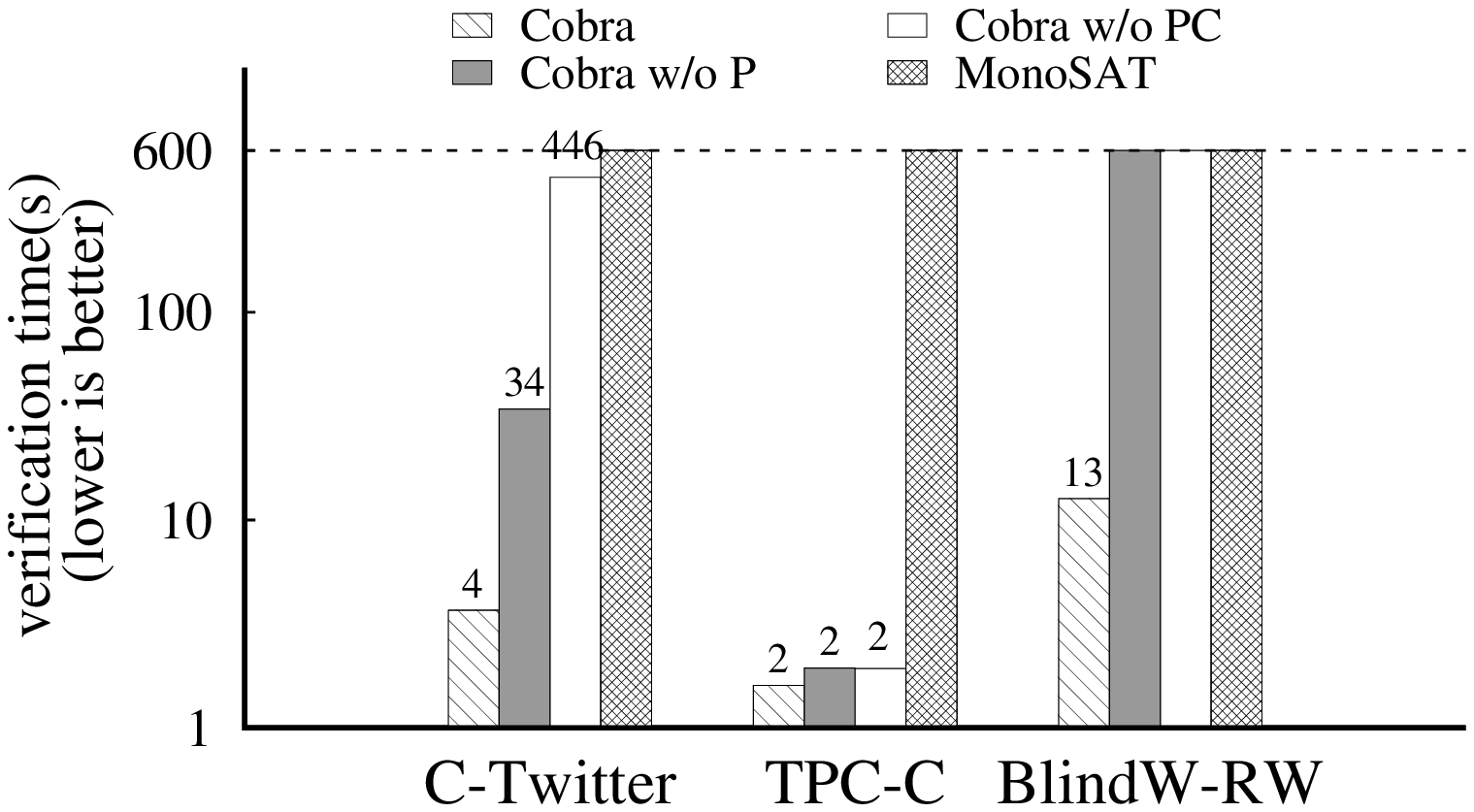}
    \caption{Differential analysis on different workloads. Log-scale,
    with values above bars. On C-Twitter, each of \sys's
    components contributes meaningfully. 
    On TPC-C, combining writes exploits the RMW pattern and
    solves all the constraints
    (see also footnote~\ref{fn:holycow}, page~\pageref{fn:holycow}).
    On the other hand, pruning is essential for BlindW-RW.}
    \label{fig:diffanalysis}
\end{figure}

\heading{Differential analysis.}
We experiment with four variants: \sys itself; \sys without
pruning~(\S\ref{s:pruning}); \sys without pruning and
coalescing~(\S\ref{s:coalescing}), which is equivalent to MonoSAT plus
write combining~(\S\ref{s:writecombining}); and the MonoSAT baseline.
We experiment with three benchmarks, with 10k transactions.
Figure~\ref{fig:diffanalysis} depicts the results.

\subsection{Scaling}
\label{s:evalscaling}

We want to know: what offered load (to the database) can \sys support on
an ongoing basis? To answer this question, we must quantify \sys's
\emph{verification capacity}, in txns/second. This depends on the
characteristics of the workload, the number of transactions one
round~(\S\ref{s:gc}) verifies ($\#tx_r$), and the average time for one
round of verification ($t_{r}$). Note that the variable here is
$\#tx_r$; $t_r$ is a function of that choice. So the
\textit{verification capacity} for a particular workload is defined as:
$\max_{\#tx_r}(\#tx_r/t_r)$.

To investigate this quantity, we run our benchmarks on RocksDB
with 24 concurrent clients, a fence transaction every 20
transactions. We generate a 100k-transaction history
ahead of time. For that same history, we vary $\#tx_r$, plot
$\#tx_r/t_r$, and choose the optimum.

\begin{figure}[]
    \centering
    \includegraphics[width=\linewidth]{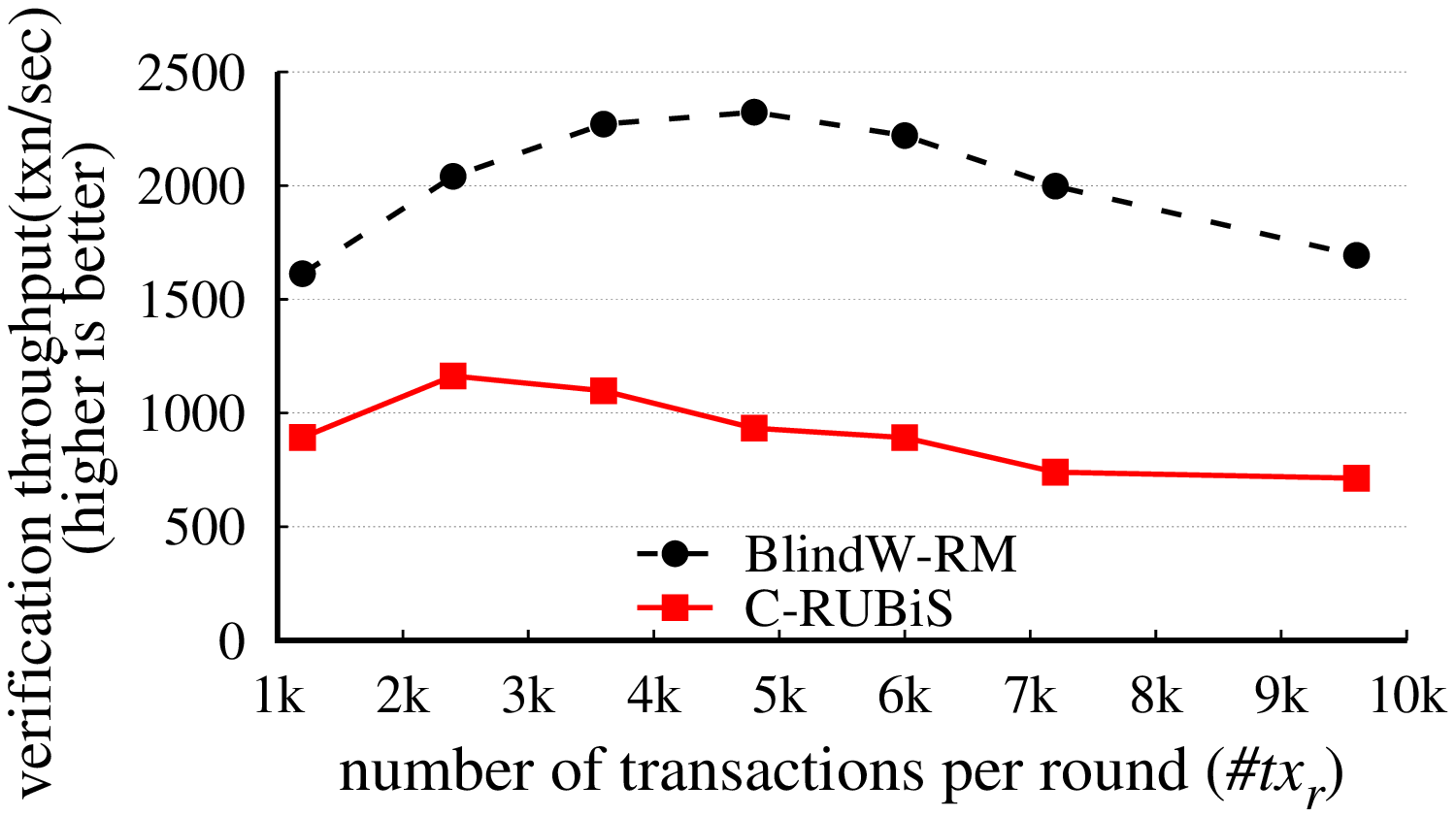}
    \caption{Verification throughput vs. round size ($\#tx_r$).
    The verification capacity for BlindW-RM (the dashed line) is 2.3k txn/sec
    when $\#tx_r$ is 5k; the capacity for C-RUBiS (the solid line) is 1.2k txn/sec
    when $\#tx_r$ is 2.5k.
    }
    \label{fig:scaling}
    \vspace{2ex}
\end{figure}

\begin{figure*}[t!]
  \centering
  \begin{subfigure}[b]{0.32\linewidth}
    \centering\includegraphics[width=\linewidth]{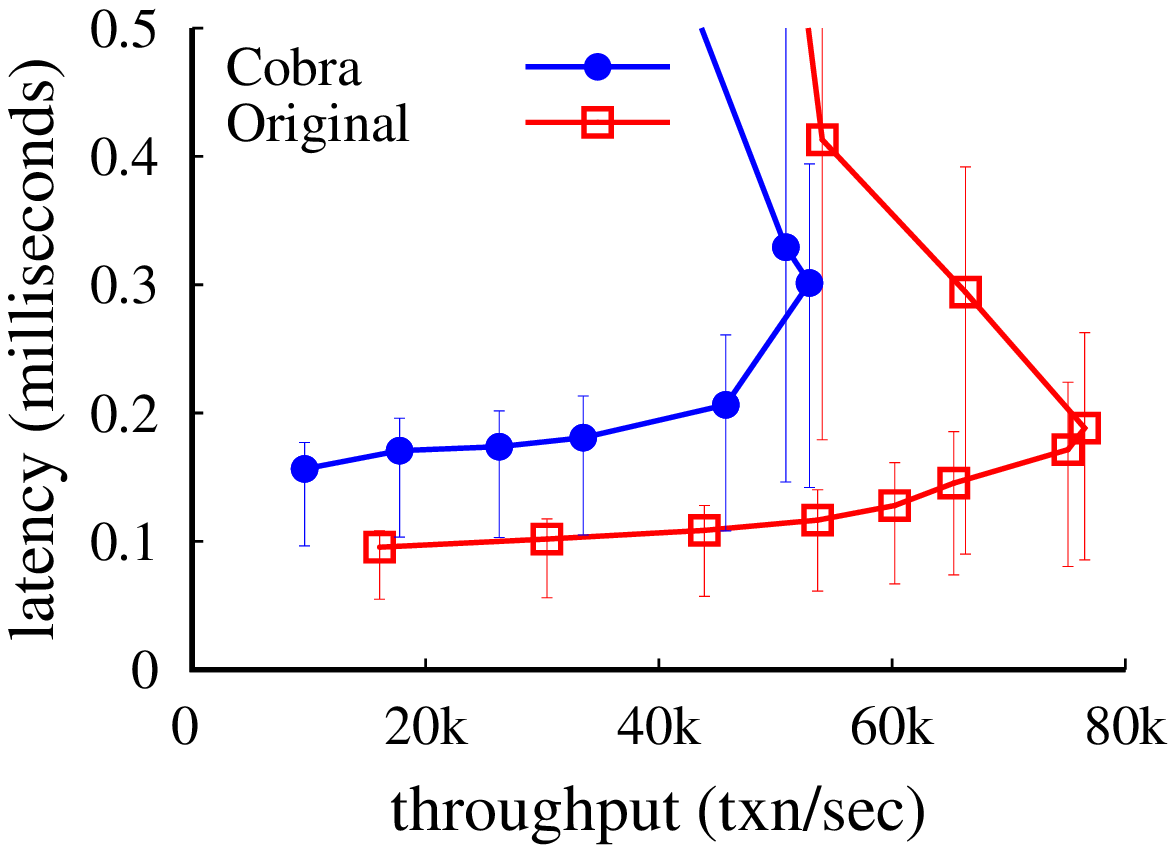}
    \caption{\label{fig:thpt-lat-rocksdb} RocksDB}
  \end{subfigure}%
  \begin{subfigure}[b]{0.32\linewidth}
    \centering\includegraphics[width=\linewidth]{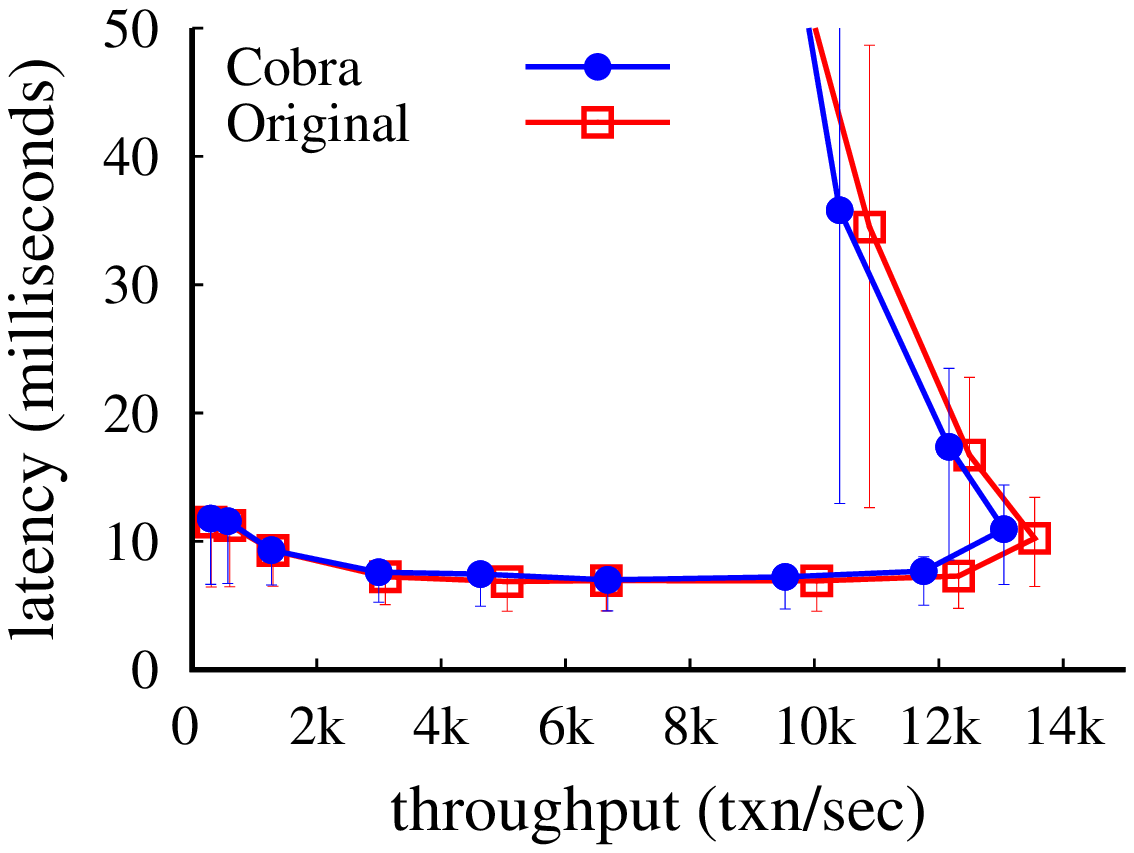}
    \caption{\label{fig:thpt-lat-postgres} PostgreSQL}
  \end{subfigure}
  \begin{subfigure}[b]{0.32\linewidth}
    \centering\includegraphics[width=\linewidth]{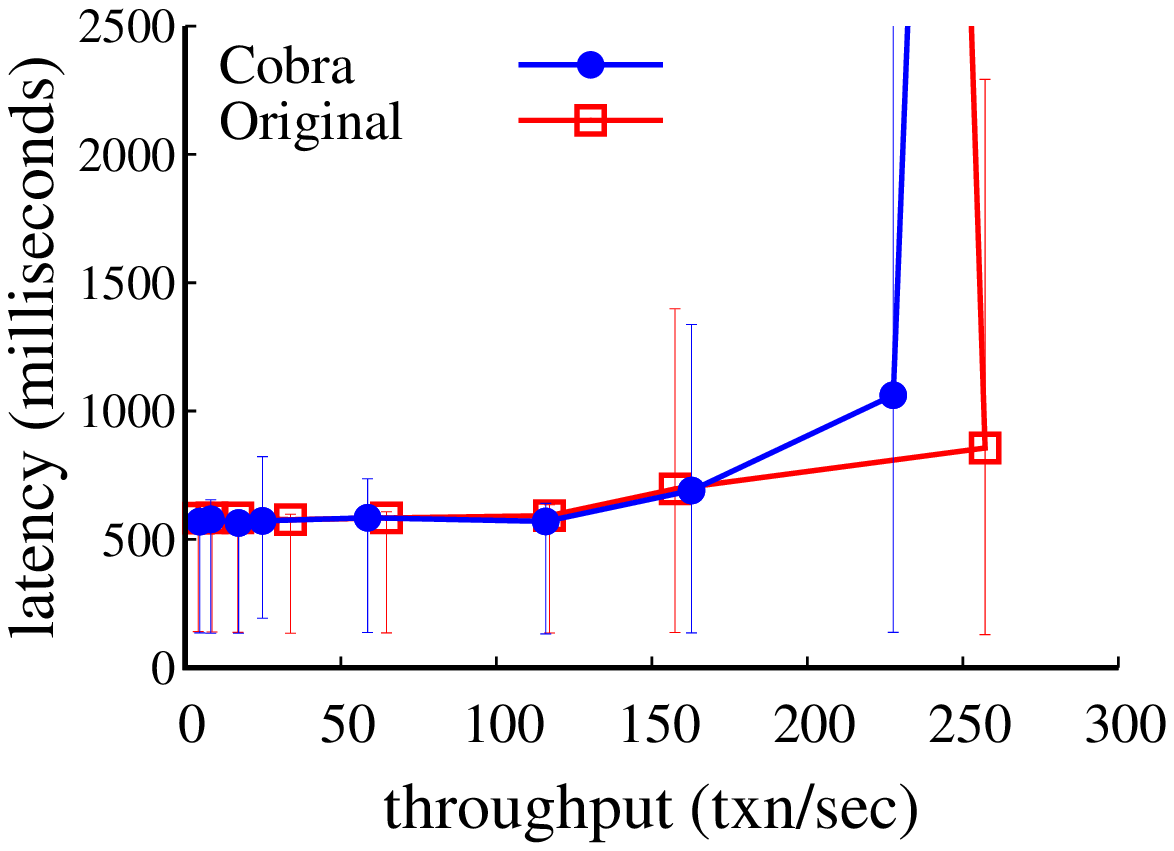}
    \caption{\label{fig:thpt-lat-google} Google Datastore}
  \end{subfigure}
  \caption{Throughput and latency, for C-Twitter benchmark.
On the left is the in-process setup;
90th percentile latency increases 64\%, with 31\% throughput penalty,
an artifact of \trace collection (disk bandwidth contention between
clients and the DB).
In the middle is the local setup (PostgreSQL), where \sys imposes minor overhead.
Finally, on the right is the cloud setup; 
there is an artifact here too: the
throughput penalty reflects a ceiling imposed by the cloud service
for a fresh DB instance.}
\label{fig:thpt-lat}
\vspace{4ex}
\end{figure*}
 
Figure~\ref{fig:scaling} depicts the results.
When $\#tx_r$ is smaller, \sys wastes cycles on redundant verification;
when $\#tx_r$ is larger, \sys suffers from a problem size that is too
large (recall that verification time increases superlinearly;
\S\ref{subsec:oneshot}).
For different workloads, the optimal
choices of $\#tx_r$ are different.

In workload BlindW-RW, \sys runs out of GPU memory. The reason is
that due to many blind writes in this workload, \sys
is unable to garbage collect enough transactions and
fit the remaining history into the GPU memory.
Our future work is to investigate this case and design a more efficient (in
terms of deleting more transactions) algorithm.

\subsection{\sys online overheads}

The baseline in this section is the legacy system;
that is, clients use the unmodified database library (for example, JDBC),
with no recording of history.

\heading{Throughput latency analysis.}
We evaluate \sys's client-side throughput and latency in the three 
setups, tuning the number of clients (up to 256) to saturate the databases.
Figure~\ref{fig:thpt-lat} depicts the results.

\newcommand{\numopchengorig}{99k\xspace}
\newcommand{\numopchengcobra}{111k\xspace}
\newcommand{\numopoverheadcheng}{11.94\%\xspace}
\newcommand{\numoptpccorig}{530k\xspace}
\newcommand{\numoptpcccobra}{568k\xspace}
\newcommand{\numopoverheadtpcc}{7.25\%\xspace}
\newcommand{\numoptwitterorig}{70k\xspace}
\newcommand{\numoptwittercobra}{82k\xspace}
\newcommand{\numopoverheadtwitter}{17.13\%\xspace}
\newcommand{\numopycsborig}{35k\xspace}
\newcommand{\numopycsbcobra}{46k\xspace}
\newcommand{\numopoverheadycsb}{31.41\%\xspace}
\newcommand{\numoprubisorig}{69k\xspace}
\newcommand{\numoprubiscobra}{82k\xspace}
\newcommand{\numopoverheadrubis}{17.32\%\xspace}

\newcommand{\networkchengcobra}{3352.0\xspace}
\newcommand{\networkchengorig}{3124.6\xspace}
\newcommand{\networkoverheadcheng}{227.4\xspace}
\newcommand{\networkpercentcheng}{7.28\%\xspace}
\newcommand{\networkycsbcobra}{1126.4\xspace}
\newcommand{\networkycsborig}{1050.9\xspace}
\newcommand{\networkoverheadycsb}{75.4\xspace}
\newcommand{\networkpercentycsb}{7.18\%\xspace}
\newcommand{\networktwittercobra}{6865.8\xspace}
\newcommand{\networktwitterorig}{6572.9\xspace}
\newcommand{\networkoverheadtwitter}{292.9\xspace}
\newcommand{\networkpercenttwitter}{4.46\%\xspace}
\newcommand{\networkrubiscobra}{2481.9\xspace}
\newcommand{\networkrubisorig}{2374.4\xspace}
\newcommand{\networkoverheadrubis}{107.5\xspace}
\newcommand{\networkpercentrubis}{4.53\%\xspace}
\newcommand{\networktpcccobra}{3688.4\xspace}
\newcommand{\networktpccorig}{3610.2\xspace}
\newcommand{\networkoverheadtpcc}{78.2\xspace}
\newcommand{\networkpercenttpcc}{2.17\%\xspace}

\newcommand{\tracesizechengcobra}{245.5\xspace}
\newcommand{\tracesizeycsbcobra}{60.2\xspace}
\newcommand{\tracesizetwittercobra}{200.7\xspace}
\newcommand{\tracesizerubiscobra}{148.9\xspace}
\newcommand{\tracesizetpcccobra}{1380.8\xspace}

{\small

\begin{figure}[]
\begin{tabular}{lccc}
\toprule
\multirow{2}{*}{workload} & \multicolumn{2}{c}{network overhead} & \trace \\ \cline{2-3}
 & traffic & percentage & size \\ \midrule
BWrite-RW & \networkoverheadcheng KB & \networkpercentcheng & \tracesizechengcobra KB \\
C-Twitter & \networkoverheadtwitter KB & \networkpercenttwitter & \tracesizetwittercobra KB \\
C-RUBiS & \networkoverheadrubis KB & \networkpercentrubis & \tracesizerubiscobra KB \\
TPC-C & \networkoverheadtpcc KB & \networkpercenttpcc & \tracesizetpcccobra KB \\ \bottomrule
\end{tabular}
\caption{Network and storage overheads per one thousand transactions. The network overheads
comes from fence transactions and the metadata (transaction ids and write ids)
added by \sys's client library.
}
\label{fig:cost}
\end{figure}

}
 
\vspace{1ex}
\heading{Network cost and \trace size.}
We evaluate the network traffic on the client side by tracking the number of bytes
  sent over the NIC.
We measure the \trace size by summing sizes of the \trace files.
Figure~\ref{fig:cost} summarizes.

\section{Related work}
\label{s:relwork}

Below, we cover many works that wish to verify or enforce the
correctness of storage, some with very similar motivations to ours. As
stated earlier~(\S\ref{s:intro}), our problem statement is differentiated by combining requirements:
(a)~a black box database, (b)~performance and concurrency approximating that
of a \sys-less system, and (c)~checking view-serializability.

\heading{Isolation testing and Consistency testing.}
Serializability is a particular 
\emph{isolation} level in a transactional system---the I in ACID transactions.
Because checking view-serializability is
  NP-complete~\cite{papadimitriou79serializability}, to the best of our
  knowledge, all works testing serializability prior to \sys
  are checking conflict-serializability where the write-write ordering is known. 
Sinha et al.~\cite{sinha10runtime} record the ordering of operations in a modified
  software transactional memory library to reduce the search space in
  checking serializability;
this work uses the polygraph data
structure~(\S\ref{s:bruteforce}).
The idea of recording order to help test serializability
  has also been used in detecting data races in multi-threaded
  programs~\cite{xu05serializability, hammer08dynamic,sumner11marathon}.

In shared memory systems and systems that offer replication (but do not
necessarily support transactions), there is an analogous correctness
contract,
namely \emph{consistency}. (Confusingly,
the ``C(onsistency)'' in ACID transactions refers to something else,
namely %
semantic invariants~\cite{bailis14linearizability}.) Example consistency models
are linearizability~\cite{herlihy90linearizability}, sequential
consistency~\cite{lamport79how}, and eventual
consistency~\cite{petersen97flexible}.

Testing for these consistency models is an analogous problem to ours.
In both cases, one searches for a schedule
  that fits the ordering constraints of both the model and
  the history~\cite{golab11analyzing}. 
As in checking serializability, the computational complexity of checking
  consistency decreases if a stronger model is targeted (for example, linearizability vs.
  sequential consistency)~\cite{gibbons97testing}, or
  if more ordering information can be (intrusively)
  acquired (by opening black boxes)~\cite{wing93testing}.

Concerto~\cite{arasu17concerto} uses \emph{deferred
  verification}, allowing it to exploit
  an offline memory checking algorithm~\cite{blum94checking} to 
  check online the sequential consistency of a highly concurrent
  key-value store.
Concerto's design achieves orders-of-magnitude performance
  improvement compared to Merkle tree-based
  approaches~\cite{merkle87digital,blum94checking},
  but it also requires modifications of the database. 
(See elsewhere~\cite{li06dynamic,dwork09efficient} for
related algorithms.)

A body of work examines cloud storage
  consistency~\cite{anderson10what,aiyer08consistability,
    lu15existential,liu14consistency}.
These works rely on extra ordering information
  obtained through techniques like loosely- or well-synchronized
  clocks~\cite{anderson10what,golab11analyzing,aiyer08consistability,
    lu15existential,kim15caelus}, or client-to-client
  communication~\cite{shraer10venus,liu14consistency}. 
As another example, a gateway that sequences the requests
  can ensure consistency by enforcing
  ordering~\cite{shraer10venus,popa11enabling,sinha18veritasdb,jain13trustworthy}. 
Some of \sys's techniques are reminiscent of these works, such as
  its use of precedence graphs~\cite{anderson10what,golab11analyzing}.
However, a substantial difference is that \sys neither
modifies the ``memory'' (the database) to get information about the
actual internal schedule nor depends on external synchronization.
\sys of course exploits epochs for safe deletion~(\S\ref{s:gc}),
but this is a performance optimization, not core to the
verification task, and invokes
standard database interfaces.

\heading{Execution integrity.}
Our problem relates to the broad category of \emph{execution
integrity}---ensuring that a module in another administrative domain is
executing as expected. For example,
Orochi~\cite{tan17efficient} is an end-to-end audit that gives a
verifier assurance that a given web application, including its database,
is executing according to the code it is allegedly running.
Orochi operates in a setting reminiscent of the one
  that we consider in this paper, in which there are collectors and an
  untrusted cloud service.
Verena~\cite{karapanos16verena} operates in a similar model (but
makes fewer assumptions, in that its hash server and
application server are mutually distrustful); Verena uses authenticated
data structures and a careful placement of function to guarantee to the
\emph{deployer} of a given web service, backed by a database,
that the delivered web pages are correct.
Orochi and Verena require that the database is \emph{strictly} serializable, they 
provide end-to-end verification of a full stack, but they cannot treat
that stack as a black box.
\Sys is the other way around: it of course tolerates (non-strict)
serializability, its verification purview is limited to the database,
but it treats the database as a black box.

Other examples of execution integrity include
AVM~\cite{haeberlen10accountable} and Ripley~\cite{vikram09ripley},
which involve checking an untrusted module by re-executing the inputs to
it. These systems likewise are ``full stack'' but ``not black box.''
 
Another approach is to use trusted components.
For example, Byzantine fault tolerant (BFT)
  replication~\cite{castro99practical} (where the assumption is that a
  super-majority is not faulty)
and TEEs (trusted execution environments, comprising 
TPM-based
  systems~\cite{sailer04design,seshadri05pioneer,mccune08flicker,
    mccune10trustvisor,parno11bootstrapping,sirer11logical,chen13towards,
    hawblitzel14ironclad} and SGX-based systems~\cite{baumann14shielding,
    schuster2015vc3,arnautov16scone,hunt16ryoan,
    shinde17panoply,aublin18libseal,krahn18pesos,sinha18veritasdb}) ensure
that the right code is running.
However, this does not ensure that the code itself is right; concretely, if a
database violates serializability owing to an implementation bug, neither BFT 
nor SGX hardware helps.

There is also a class of systems that uses complexity-theoretic and
cryptographic
mechanisms~\cite{zhang15integridb,zhang17vsql,braun13verifying,setty18proving}.
None of these works handle systems of realistic scale, and only one of
them~\cite{setty18proving} handles concurrent workloads.  An exception
is Obladi~\cite{crooks18obladi}, which remarkably provides ACID
transactions atop an ORAM abstraction by exploiting a trusted proxy that
carefully manages the interplay between concurrency control and the ORAM
protocol; its performance is surprisingly good (as cryptographic-based
systems go) but still pays 1-2 orders of magnitude overhead in
throughput and latency.

\heading{SMT solver on detecting serializability violations.}
Several works~\cite{nagar18automated,brutschy18static,brutschy2017serializability}
propose using SMT solvers to detect serializability violations under weak consistency.
The underlying problem is different: they focus on encoding
\emph{static programs} and the \emph{correctness criterion} of weak consistency;
while \sys focuses on how to encode \emph{histories} more efficiently.
The most relevant work~\cite{sinha11predicting} is to use SMT solvers to
permutate all possible interleaving for a concurrent program and search for
serializablility violations.
Despite of different setups,
our baseline implementation on Z3 (\S\ref{subsec:oneshot})
has similar encoding and similar performance
(verifying hundreds of transactions in tens of seconds).

\heading{Correctness testing for distributed systems.}
There is a line of research on testing the correctness of distributed systems
under various failures, including network partition~\cite{alquraan2018analysis},
power failures~\cite{zheng2014torturing},
and storage faults~\cite{ganesan2017redundancy}.
In particular, Jepsen~\cite{jepsenweb} is a black-box testing framework
(also an analysis service)
that has successfully detected massive amount of correctness bugs
in some production distributed systems.
\sys is complementary to Jepsen, providing the ability to check serializability
of black-box databases.

\heading{Definitions and interpretations of isolation levels.}
\sys of course uses precedence graphs, which are
a common tool for reasoning about isolation
  levels~\cite{papadimitriou79serializability,bernstein79formal,adya99weak}. 
However,
isolation levels can be interpreted via other means such as excluding
  anomalies~\cite{berenson95critique} and client-centric
  observations~\cite{crooks17seeing};
 it remains an open and
  intriguing question whether the other definitions would yield a more
  intuitive and more easily-implemented encoding and algorithm than the
  one in \sys.

\section*{Acknowledgments}

Sebastian Angel, Miguel Castro, Byron Cook, Andreas Haeberlen,
Dennis Shasha, Ioanna Tzialla, Thomas Wies, and Lingfan Yu
made helpful comments and gave useful pointers. This work was supported
by NSF grants CNS-1423249 and CNS-1514422, ONR grant N00014-16-1-2154,
and AFOSR grants FA9550-15-1-0302 and FA9550-18-1-0421.

\frenchspacing

\begin{flushleft}
\footnotesize
\setlength{\parskip}{0pt}
\setlength{\itemsep}{0pt}
\bibliographystyle{abbrv}
\bibliography{conferences-long-with-abbr,bibs}
\end{flushleft}

\newpage
\appendix
\newpage
\section{The validity of \sys's encoding}
\label{sec:appxa}

Recall the crucial fact in Section~\ref{subsec:bruteforce}: an acyclic
precedence graph that is compatible with a \polyg constructed from a
\trace exists iff that \trace is
serializable~\cite{papadimitriou79serializability}. 
 In this section, we
establish the analogous statement for \sys's encoding.  We do this by
following the contours of Papadimitriou's proof of the
baseline statement~\cite{papadimitriou79serializability}. However \sys's
algorithm requires that we attend to additional details, complicating
the argument somewhat.

\subsection{Definitions and preliminaries}
\label{subsec:defandpreliminaries}

In this section, we define the terms used in our main argument (\S\ref{subsec:mainproof}):
\emph{history}, \emph{schedule}, \textit{\sys polygraph}, and \emph{\chains}.

\heading{History and schedule.}
The description of histories and schedules below restates what is in
section~\ref{s:searchprelims}.

A \textit{history} is a set of read and write operations, each of which 
belongs to a transaction.\footnote{The term ``history''~\cite{papadimitriou79serializability} was originally
defined on a fork-join parallel program schema. We have adjusted the definition
to fit our setup (\S\ref{s:setup}).}
Each write operation in the history has a key and a value as its arguments;
  each read operation has a key as argument, and a value as its result.
The result of a read operation is the same as the value argument of a particular
  write operation;
  we say that this read operation \textit{reads from} this write operation.
We assume each value is unique and can be associated to the corresponding write;
  in practice, this is guaranteed by \sys's client library described in
  Section~\ref{s:impl}.
We also say that a transaction $\vtx_i$ \textit{reads (a key $k$) from} another transaction $\vtx_j$
if: $\vtx_i$ contains a read $\rop$, $\rop$ reads from write $\wop$ on $k$, and 
$\vtx_j$ contains the write $\wop$.

A \textit{schedule} is a total order of all operations in a history.
A \textit{serial schedule}
means that the schedule does not have overlapping transactions.
A history \textit{matches} a schedule if: they have the same operations, and executing the
  operations in schedule order on a single-copy set of data results in
  the same read results as in the history.
So a read reading-from a write %
indicates that this write is the read's \textit{most recent write} (to this key)
in any matching schedule.

\begin{definition2}[Serializable history]
A \emph{serializable history} is a history that matches a serial schedule.
\label{def:serializable}
\end{definition2}

\heading{\mypolyg.}
In the following, we define a \mypolyg; this is a helper notion for the known graph
($\vg$ in the definition below)
and generalized \constraints
($\vconstraints$ in the definition below)
mentioned in Section~\ref{subsec:reduce}.

\begin{definition2}[\mypolyg]
Given a history $h$,
a \emph{\mypolyg} $Q(h) = (\vg,\,\vconstraints)$ where $\vg$ and $\vconstraints$
are generated by \textsc{\ConstructEncoding} from Figure~\ref{fig:algocodemain}.
\label{def:mypolyg}
\end{definition2}

We call a directed graph $\hat{g}$
\textit{compatible} with a \mypolyg $Q(h) = (\vg,\,\vconstraints)$,
if $\hat{g}$ has the same vertices as $\vg$,
includes the edges from $\vg$,
and selects one edge set from each constraint in $\vconstraints$.

\begin{definition2}[Acyclic \mypolyg]
A \mypolyg $Q(h)$ is \emph{acyclic} if there exists an acyclic graph
that is compatible with $Q(h)$.
\end{definition2}

\heading{Chains.}
When constructing a \mypolyg from a history,
function \textsc{\PruneConstraintsRMW} in
\sys's algorithm
(Figure~\ref{fig:algocodemain})
produces \textit{\chains}.
One \chain is an ordered list of transactions, associated to a key $k$,
that (supposedly) contains a sequence of consecutive writes (defined below in
Definition~\ref{def:consecutivewrites}) on key $k$.
In the following, we will first define what is a sequence of consecutive writes
and then prove that a \chain is indeed such a sequence.

\begin{definition2}[\CF{\adjacentwrite}]
In a history, a transaction $\vtx_i$ is a \emph{\adjacentwrite}
of another
transaction $\vtx_j$ on a key $k$,
if (1) both $\vtx_i$ and $\vtx_j$ write to $k$ and
   (2) $\vtx_i$ reads $k$ from $\vtx_j$.
\label{def:adjacentwrite}
\end{definition2}

\begin{definition2}[A sequence of consecutive writes]
\emph{A sequence of consecutive writes} on a key $k$ of length $n$
is a list of transactions $[\vtx_1,\ldots,\vtx_n]$
for which $\vtx_i$ is a \adjacentwrite of $\vtx_{i-1}$ on $k$,
for $1<i\leq n$.
\label{def:consecutivewrites}
\end{definition2}

Although the overall problem of detecting serializability is NP-complete~\cite{papadimitriou79serializability},
there are \emph{local} malformations, which
immediately indicate that a history is not serializable. We capture two
of them in the following definition:

\begin{definition2}[An \notwellformed history]
\emph{An \notwellformed history} $h$ is a history that
either (1) contains a transaction that has multiple \adjacentwrites on one key,
or (2) has a cyclic known graph $\vg$ of $Q(h)$.
\label{def:notwellformed}
\end{definition2}

An \notwellformed history is not serializable.
First, if a history has condition (1) in the above definition, there exist at
least two transactions that are \adjacentwrites of the same transaction (say
$\vtx_i$) on some key $k$. And, these two \adjacentwrites cannot be ordered in
a serial schedule, because whichever is scheduled later would
read $k$ from the other rather than from $\vtx_i$.
Second, if there is a cycle in the known graph,
this cycle must include multiple transactions
(because there are no self-loops, since we assume that transactions
never read keys after writing to them). The members of this cycle 
cannot be ordered in a serial schedule.

\begin{lemma2}
\sys rejects \notwellformed histories.
\label{lemma:algoreject}
\end{lemma2}

\begin{proof}
\sys (the algorithm in Figure~\ref{fig:algocodemain} and the constraint solver)
detects and rejects \notwellformed histories as follows.
(1) If a transaction has multiple \adjacentwrites on the same key in $h$,
\sys's algorithm explicitly detects this case.
The algorithm checks,
for transactions reading and writing the same key (line~\ref{li:rwsamekey}),
whether multiple of them read this key from the same transaction (line~\ref{li:detectrmws}).
If so, the transaction being read has multiple \adjacentwrites,
hence the algorithm rejects (line~\ref{li:reject1}).
(2) If the known graph has a cycle,
\sys detects and rejects this history when checking acyclicity in the
constraint solver.
\end{proof}

On the other hand, if a history is not \notwellformed,
we want to argue that each \chain produced by the algorithm
is a sequence of consecutive writes.

\begin{claim}
If \sys's algorithm makes it to line~\ref{li:beforecombine} (immediately before \textsc{\PruneConstraintsRMW}),
then from this line on,
any transaction writing to a key $k$ appears in exactly one \chain on $k$.
\label{claim:onechain}
\end{claim}

\begin{proof}
Prior to line~\ref{li:beforecombine},
\sys's algorithm loops over all the write operations
(line~\ref{li:write2chain1}--\ref{li:write2chain2}),
creating a \chain for each one (line~\ref{li:preparecombine}).
As in the literature~\cite{papadimitriou79serializability, weikum01transactional},
we assume that each transaction writes to a key only once.
Thus, any \vtx writing to a key $k$ has exactly one write operation to $k$
and hence appears in exactly one \chain on $k$ in line~\ref{li:beforecombine}.

Next, we argue that \textsc{\PruneConstraintsRMW} preserves this
invariant. 
This suffices to prove the claim, because after
line~\ref{li:beforecombine}, only \textsc{\PruneConstraintsRMW} updates
\chains (variable $\vperkey$ in the algorithm).

The invariant is preserved by \textsc{\PruneConstraintsRMW} because each
of its loop iterations splices two \chains on the same key into a new
\chain (line~\ref{li:concatenateww}) and deletes the two old \chains
(line~\ref{li:rmoldchains}). From the perspective of a transaction
involved in a splicing operation, its old chain on key  $k$ has been
destroyed, and it has joined a new one on key $k$, meaning that the 
number of chains it belongs to on key $k$ is unchanged: the number
remains 1.
\end{proof}

One clarifying fact %
is that a transaction can appear in multiple \chains
on different keys,
because a transaction can write to multiple keys.

\begin{claim}
If \sys's algorithm does not reject in line~\ref{li:reject1},
then after \textsc{\CreateKnownGraph},
for any two distinct entries 
$\ven_1$ and $\ven_2$ (in the form of $\langle \vkey,\,\vtx_i,\,\vtx_j \rangle$)
in $\wwpairs$:
if $\ven_1.\fkey = \ven_2.\fkey$,
then $\ven_1.\textrm{tx}_i \neq \ven_2.\textrm{tx}_i$
and $\ven_1.\textrm{tx}_j \neq \ven_2.\textrm{tx}_j$.
\label{claim:distincttuples}
\end{claim}

\begin{proof}
First, we prove $\ven_1.\textrm{tx}_i \neq \ven_2.\textrm{tx}_i$.
In \sys's algorithm, line~\ref{li:getwwpairs} is the only point
  where new entries are inserted into $\wwpairs$.
Because of the check in line~\ref{li:detectrmws}--\ref{li:reject1},
  the algorithm guarantees that a new entry will not be inserted into $\wwpairs$
  if an existing entry has the same $\langle \vkey,\, \vtx_i \rangle$.
Also, existing entries are never modified.
Thus, there can never be two
  entries in $\wwpairs$ indexed by the same $\langle \vkey,\, \vtx_i \rangle$.

Second, we prove $\ven_1.\textrm{tx}_j \neq \ven_2.\textrm{tx}_j$.
As in the literature~\cite{papadimitriou79serializability, weikum01transactional},
we assume that one transaction reads a key at most once.\footnote{In our
implementation, this
assumption is guaranteed by \sys's client library~(\S\ref{s:impl}).}
As a consequence, the body of the loop in line~\ref{li:rwsamekey},
including line~\ref{li:getwwpairs}, is
executed at most once for
each (key,tx) pair. Therefore,
there cannot be two entries in $\wwpairs$ that match 
$\langle \vkey,\textrm{\_}, \vtx \rangle$.
\end{proof}

\begin{claim}
In one iteration
of \textsc{\PruneConstraintsRMW} (line~\ref{li:looprmw}),
for $\ven_i = \langle \vkey,\,\vtx_1,\,\vtx_2 \rangle$ retrieved from $\wwpairs$,
there exist $\vwlist_1$ and $\vwlist_2$,
such that $\vtx_1$ is the tail of $\vwlist_1$ and $\vtx_2$ is the head of $\vwlist_2$.
\label{claim:tx1head}
\end{claim}

\begin{proof}
  Invoking Claim~\ref{claim:onechain}, denote the \chain on $\vkey$ that
  $\vtx_1$ is in as $\vwlist_i$; similarly, denote $\vtx_2$'s chain as
  $\vwlist_j$.

  Assume to the contrary that $\vtx_1$ is not the tail of $\vwlist_i$.
  Then there is a transaction $\vtx'$ next to $\vtx_1$ in $\vwlist_i$.
  But the only way for two transactions ($\vtx_1$ and $\vtx'$) to appear adjacent
  in a \chain is through the concatenation in line~\ref{li:concatenateww},
  and that requires an entry $\ven_j = \langle \vkey,\,\vtx_1,\,\vtx'\rangle$ in $\wwpairs$.
  Because $\vtx'$ is already in $\vwlist_i$ when the current iteration
  happens, 
  $\ven_j$ must have been retrieved in some prior iteration.
  Since $\ven_i$ and $\ven_j$ appear in different iterations,
  they are two distinct entries in $\wwpairs$.
  Yet, both of them are indexed by $\langle \vkey,\,\vtx_1 \rangle$,
  which is impossible, by Claim~\ref{claim:distincttuples}.

  Now
  assume to the contrary that $\vtx_2$ is not the head of $\vwlist_j$.
  Then $\vtx_2$ has an immediate predecessor $\vtx'$ in $\vwlist_j$.
  In order to have $\vtx'$ and $\vtx_2$ appear adjacent in $\vwlist_j$, there
  must be an entry $\ven_k = \langle \vkey,\,\vtx',\,\vtx_2 \rangle$ in $\wwpairs$.
  Because $\vtx'$ is already in $\vwlist_j$ when the current iteration
  happens,
  $\ven_k$ must have been retrieved in an earlier iteration.
  So, $\ven_k = \langle \vkey,\,\vtx',\,\vtx_2 \rangle$
  and $\ven_i = \langle \vkey,\,\vtx_1,\,\vtx_2 \rangle$ are distinct entries in $\wwpairs$,
  which is impossible,
  by Claim~\ref{claim:distincttuples}.
\end{proof}

\begin{lemma2}
If $h$ is not \notwellformed,
every \chain is a sequence of consecutive writes
after \textsc{\PruneConstraintsRMW}.
\label{lemma:longestchain}
\end{lemma2}

\begin{proof}
  Because $h$ is not \notwellformed, it doesn't contain any transaction
  that has multiple \adjacentwrites. Hence, \sys's algorithm does not reject
  in line~\ref{li:reject1} and can make it to \textsc{\PruneConstraintsRMW}.

  At the beginning (immediately before \textsc{\PruneConstraintsRMW}),
  all chains are single-element lists (line~\ref{li:preparecombine}).
  By Definition~\ref{def:consecutivewrites}, each chain is a sequence of
  consecutive writes with only one transaction.

  Assume that, before loop iteration $t$, each \chain is
  a sequence of consecutive writes. We show that after iteration $t$ (before
  iteration $t+1$), \chains are still sequences of consecutive writes.

  If $t \leq \textrm{size}(\wwpairs)$, then in line~\ref{li:looprmw},
  \sys's algorithm gets an entry
  $\langle \vkey,\,\vtx_1,\,\vtx_2\rangle$ from $\wwpairs$,
  where $\vtx_2$ is $\vtx_1$'s \adjacentwrite on $\vkey$.
  Also, we assume one transaction does not read from itself
  (as in the literature~\cite{papadimitriou79serializability, weikum01transactional}),
  and since $\vtx_2$ reads from $\vtx_1$, $\vtx_1 \ne \vtx_2$.
  Then, the algorithm references the \chains that they are in: $\vwlist_1$ and $\vwlist_2$.

  First, we argue that $\vwlist_1$ and $\vwlist_2$ are distinct \chains.
  By Claim~\ref{claim:onechain}, no transaction can appear in
  two \chains on the same key, so $\vwlist_1$ and $\vwlist_2$
  are either distinct \chains or the same \chain.
  Assume they are the same \chain ($\vwlist_1 = \vwlist_2$).
  If $\vwlist_1$ ($=\vwlist_2$) is a single-element \chain,
  then $\vtx_1$ (in $\vwlist_1$) is $\vtx_2$ (in $\vwlist_2$),
  a contradiction to $\vtx_1 \neq \vtx_2$.

  Consider the case that $\vwlist_1$ ($=\vwlist_2$) contains multiple transactions.
  Because $\vtx_2$ reads from $\vtx_1$,
  there is an edge $\vtx_1 \to \vtx_2$ (generated from line~\ref{li:addwr})
  in the known graph of $Q(h)$.
  Similarly, because $\vwlist_1$ is a sequence of consecutive writes
  (the induction hypothesis), any transaction
  $\vtx$ in $\vwlist_1$ reads from its immediate prior transaction,
  hence there is an edge from this prior transaction to $\vtx$.
  Since every pair of adjacent transactions in $\vwlist_1$ has such an edge,
  the head of $\vwlist_1$ has a path to the tail of $\vwlist_1$.
  Finally, by Claim~\ref{claim:tx1head}, $\vtx_2$ is the head of $\vwlist_2$ and
  $\vtx_1$ is the tail of $\vwlist_1$,
  as well as $\vwlist_1 = \vwlist_2$,
  there is a path $\vtx_2 \rightsquigarrow \vtx_1$.
  Thus, there is a cycle ($\vtx_1 \to \vtx_2 \rightsquigarrow \vtx_1$) in the known graph,
  so $h$ is \notwellformed, a contradiction.

  Second, we argue that the concatenation of $\vwlist_1$ and $\vwlist_2$, denoted as $\vwlist_{1+2}$,
  is a sequence of consecutive writes.
  Say the lengths of $\vwlist_1$ and $\vwlist_2$ are $n$ and $m$ respectively.
  Since $\vwlist_1$ and $\vwlist_2$ are distinct sequences of consecutive writes,
  all transactions in $\vwlist_{1+2}$ are distinct and
  $\vwlist_{1+2}[i]$ reads from $\vwlist_{1+2}[i-1]$ for $i \in \{2,
  \ldots, n+m \} \setminus \{n+1\}$.
  For $i=n+1$, the preceding also holds, because
  $\vtx_1$ is $\vwlist_1$'s tail ($=\vwlist_{1+2}[n]$),
  $\vtx_2$ is $\vwlist_2$'s head ($=\vwlist_{1+2}[n+1]$),
  and $\vtx_2$ is the \adjacentwrite of $\vtx_1$ ($\vtx_2$ reads from
  $\vtx_1$).
Thus,
   $\vwlist_{1+2}$ is a sequence of consecutive writes, 
 according to Definition~\ref{def:consecutivewrites}.
  
  If $t > \textrm{size}(\wwpairs)$ and the loop ends, then \chains don't change.
  As they are sequences of consecutive writes after the final step (when $t = \textrm{size}(\wwpairs)$),
  they still are after \textsc{\PruneConstraintsRMW}.
\end{proof}

In the following, when we refer to \chains, we mean the state of \chains
after executing
\textsc{\PruneConstraintsRMW}.

\subsection{The main argument}
\label{subsec:mainproof}

In this section, the two theorems (Theorem~\ref{theorem:=>} and~\ref{theorem:<=})
together prove the validity of \sys's encoding.

\begin{theorem2}
If a history $h$ is serializable, then $Q(h)$ is acyclic.
\label{theorem:=>}
\end{theorem2}

\begin{proof}

Because $h$ is serializable, there exists a serial schedule $\hat{s}$ that
$h$ matches.

\begin{claim}
For any transaction $\vrtx$ that reads from a transaction $\vwtx$ in $h$,
$\vrtx$ appears after $\vwtx$ in $\hat{s}$.
\label{claim:wrorder}
\end{claim}

\begin{proof}
This follows from the definitions given at the start of the section: if
$\vrtx$ reads from $\vwtx$ in $h$, then there is a read operation $\rop$
in $\vrtx$ that reads from a write operation $\wop$ in $\vwtx$.  Thus,
as stated earlier and by definition of matching, $\rop$ appears later
than $\wop$ in $\hat{s}$.  Furthermore, by definition of serial
schedule, transactions don't overlap in $\hat{s}$.  Therefore, all of
$\vrtx$ appears after all of $\vwtx$ in $\hat{s}$.
\end{proof}

\begin{claim}
For any pair of transactions ($\vrtx$, $\vwtx$) where $\vrtx$ reads a key $k$
from $\vwtx$ in $h$,
no transaction $\vwtx'$ that writes to $k$
can appear between $\vwtx$ and $\vrtx$ in $\hat{s}$.
\label{claim:nointervening}
\end{claim}

\begin{proof}
  Assume to the contrary that there exists $\vwtx'$ that appears in between
  $\vwtx$ and $\vrtx$ in $\hat{s}$.
  By
   Claim~\ref{claim:wrorder},
  $\vrtx$ appears after $\vwtx$ in $\hat{s}$.
  Therefore, $\vwtx'$ appears in $\hat{s}$ before $\vrtx$ and after $\vwtx$.
  Thus, in $\hat{s}$, 
  $\vrtx$ does not return the value of $k$ written by $\vwtx$.
  But in $h$, $\vrtx$ returns the value of $k$ written by $\vwtx$. Thus, $\hat{s}$
  and $h$ do not match, a contradiction.
\end{proof}

In the following, we use $\vhead_k$ and $\vtail_k$ as shorthands to
represent, respectively,
the head transaction and the tail transaction of $\vwlist_k$.
And, we denote that $\vtx_i$ appears before $\vtx_j$ in $\hat{s}$ as $\vtx_i <_{\hat{s}} \vtx_j$.

\begin{claim}
For any pair of \chains $(\vwlist_i,\,\vwlist_j)$ on the same key $k$,
if $\vhead_i <_{\hat{s}} \vhead_j$,
then (1) $\vtail_i  <_{\hat{s}} \vhead_j$ 
and (2) for any transaction $\vrtx$ that reads $k$ from $\vtail_i$,
$\vrtx <_{\hat{s}} \vhead_j$.
\label{claim:headtail}
\end{claim}

\begin{proof}
First, we prove $\vtail_i <_{\hat{s}} \vhead_j$.
If $\vhead_j \in \vwlist_j$ then $\vhead_j \not\in \vwlist_i$, by
Claim~\ref{claim:onechain}.
If $\vwlist_i$ has only one transaction (meaning $\vhead_i = \vtail_i$),
then $\vtail_i = \vhead_i <_{\hat{s}} \vhead_j$.

Next, if $\vwlist_i$ is a multi-transaction \chain, it can be written as
\[
  \vtx_1,\ \cdots \ \vtx_p,\, \vtx_{p+1},\ \cdots\ \vtx_n.
\]
By Lemma~\ref{lemma:longestchain}, 
$\vwlist_i$ is a sequence of consecutive writes on $k$,
so each transaction reads $k$ from its prior transaction in $\vwlist_i$.
Then, by Claim~\ref{claim:wrorder},
$\vtx_p <_{\hat{s}} \vtx_{p+1}$, for $1 \le p < n$.
Now, assume to the contrary that $\vhead_j <_{\hat{s}} \vtail_i$ ($=\vtx_n$).
Then, by the given, $\vtx_1(=\vhead_i) <_{\hat{s}} \vhead_j <_{\hat{s}} \vtx_n$.
Thus,
for some  $1 \le p < n$, we have $\vtx_p <_{\hat{s}} \vhead_j <_{\hat{s}}
\vtx_{p+1}$. 
But this is a contradiction, because $\vtx_{p+1}$ reads $k$ from $\vtx_p$,
and thus by Claim~\ref{claim:nointervening},
$\vhead_j$ cannot appear between them in $\hat{s}$.

Second, we prove that any transaction $\vrtx$ that reads $k$ from $\vtail_i$ 
appears before $\vhead_j$ in $\hat{s}$. 
Assume to the contrary that $\vhead_j <_{\hat{s}} \vrtx$. We have from
the first half of the claim that $\vtail_i <_{\hat{s}} \vhead_j$. 
Thus, $\vhead_j$ appears between $\vtail_i$ and $\vrtx$ in $\hat{s}$,
which is again a contradiction, by Claim~\ref{claim:nointervening}.
\end{proof}

Now we prove that $Q(h)$ is acyclic by constructing
a compatible graph $\hat{g}$ and proving $\hat{g}$ is acyclic.
We have the following fact from function
\textsc{\BundleConstraints}.

\begin{fact}
In \textsc{\BundleConstraints},
each \constraint $\langle A,\,B \rangle$ %
is generated from a pair of \chains
$(\vwlist_1,\,\vwlist_2)$ on the same key $k$ in line~\ref{li:coalesce}.
All edges in edge set $A$ point to $\vhead_2$,
and all edges in $B$ point to $\vhead_1$.
This is because
all edges in $A$ have the form either
$(\vrtx, \vhead_2)$ 
or $(\vtail_1, \vhead_2)$; see lines~\ref{li:noread}
and~\ref{li:wmid}--\ref{li:wend}.
Similarly by swapping $\vwlist_1$ and $\vwlist_2$ (line~\ref{li:c2c1} and \ref{li:c2c2}),
edges in $B$ point to $\vhead_1$.
\label{fact:constraint}
\end{fact}

We construct graph $\hat{g}$ as follows:
first, let $\hat{g}$ be the known graph of $Q(h)$.
Then, for each constraint $\langle A,\,B \rangle$ in $Q(h)$, 
and letting $\vhead_1$ and $\vhead_2$ be defined as in
Fact~\ref{fact:constraint},
add $A$ to $\hat{g}$ if $\vhead_1 <_{\hat{s}} \vhead_2$,
and otherwise add $B$ to $\hat{g}$.
This process results in a directed graph $\hat{g}$.

Next, we show that
all edges in $\hat{g}$ are a subset of the total ordering in $\hat{s}$; this
implies $\hat{g}$ is acyclic.

First, the edges in the known graph (line~\ref{li:addwr} and~\ref{li:addrw})
are a subset of the total ordering given by $\hat{s}$.
Each edge added in line~\ref{li:addwr} represents
that the destination vertex
reads from the source vertex in $h$.
By Claim~\ref{claim:wrorder}, this ordering holds in $\hat{s}$.
As for the edges in line~\ref{li:addrw}, they are added to capture
the fact that a read
operation (in transaction \vrtx)
that reads from a write (in transaction $\vwlist[i]$) is sequenced before the
next write on the same key
(in transaction $\vwlist[i+1]$), an ordering that also holds in $\hat{s}$. (This
is known as an \emph{\rwrel} in the literature~\cite{adya99weak}.)
If this ordering doesn't hold in $\hat{s}$, then $\vwlist[i+1]
<_{\hat{s}} \vrtx$, and thus $\vwlist[i] <_{\hat{s}} \vwlist[i+1]
<_{\hat{s}} \vrtx$, which contradicts Claim~\ref{claim:nointervening}.

Second, consider the edges in $\hat{g}$ that come from
\constraints.
Take a \constraint $\langle A, B \rangle$ generated from \chains
$(\vwlist_1,\,\vwlist_2)$ on the same key. If $\vhead_1 <_{\hat{s}} \vhead_2$, then by
Fact~\ref{fact:constraint}
and construction of $\hat{g}$,
all added edges have the form
$(\vtail_1, \vhead_2)$ or $(\vrtx, \vhead_2)$, where
$\vrtx$ reads from $\vtail_1$.
By Claim~\ref{claim:headtail}, the source vertex
of these edges appears prior to $\vhead_2$ in $\hat{s}$;
thus, these edges respect the ordering in $\hat{s}$.
When 
$\vhead_2 <_{\hat{s}} \vhead_1$, 
the foregoing argument works the same, with appropriate relabeling.
Hence, all \constraint edges chosen in $\hat{g}$ are a subset of the
total ordering given by $\hat{s}$.
This completes the proof.
\end{proof}

\bigskip

\begin{theorem2}
If $Q(h)$ is acyclic, then the history $h$ is serializable.
\label{theorem:<=}
\end{theorem2}

\begin{proof}

Given that $Q(h)$ is acyclic, \sys accepts $h$.
Hence, by Lemma~\ref{lemma:algoreject}, $h$ is not \notwellformed.
And, by Lemma~\ref{lemma:longestchain}, each \chain (after
\textsc{\PruneConstraintsRMW}) is a sequence of consecutive writes.

Because $Q(h)$ is acyclic, there must exist an acyclic directed graph $q$ that is
compatible with $Q(h)$.

\begin{claim}
If $\vtx_i$ appears before $\vtx_j$ in a \chain $\vwlist_k$,
then graph $q$ has $\vtx_i \rightsquigarrow \vtx_j$.
\label{claim:wwinchain}
\end{claim}

\begin{proof}
Because $\vwlist_k$ is a sequence of consecutive writes,
a transaction $\vtx$ in $\vwlist_k$ reads from its immediate predecessor in $\vwlist_k$,
hence there is an edge in the known graph (generated by line~\ref{li:addwr})
from the predecessor to $\vtx$.
Because every pair of adjacent transactions in $\vwlist_k$ has such an
edge and $\vtx_i$ appears before $\vtx_j$ in $\vwlist_k$,
$\vtx_i \rightsquigarrow \vtx_j$ in $Q(h)$'s known graph.
As $q$ is compatible with $Q(h$),
such a path from $\vtx_i$ to $\vtx_j$ also exists in $q$.
\end{proof}

\begin{claim}
For any \chain $\vwlist_i$ (on a key $k$) and
any transaction $\vwtx_j \not\in \vwlist_i$ that writes to $k$,
graph $q$ has either:
(1) paths from $\vtail_i$ and transactions that read \vkey from $\vtail_i$ (if any)
to $\vwtx_j$,
or
(2) paths from $\vwtx_j$ to all the transactions in $\vwlist_i$.
\label{claim:wwtwochains}
\end{claim}

\begin{proof}
  Call the \chain that $\vwtx_j$ is in $\vwlist_j$.
  By Claim~\ref{claim:onechain},
  $\vwlist_j$ exists and $\vwlist_j \neq \vwlist_i$.

  For $\vwlist_i$ and $\vwlist_j$,
  $Q(h)$ has a \constraint $\langle A,\,B \rangle$
  that is generated from them (line~\ref{li:genconstraint}).
  This is because $\vwlist_i$ and $\vwlist_j$ touch the same key $k$,
  and \sys's algorithm creates one \constraint for every
  pair of \chains on the same key (line~\ref{li:chainpairs}).
  (We assume $\vwlist_i$ is the first argument of function
  \textsc{\BundleConstraints} and $\vwlist_j$ is the second.)

  First, we argue that the edges in edge set $A$
  establish $\vtail_i \rightsquigarrow \vhead_j$
  and $\vrtx \rightsquigarrow \vhead_j$ ($\vrtx$ reads $k$ from $\vtail_i$)
  in the known graph;
  and $B$ establishes $\vtail_j \rightsquigarrow \vhead_i$.
  Consider edge set $A$. There are two cases: (i)
  there are reads $\vrtx$ reading from $\vtail_i$, and (ii) there is no such read.
  In case (i),
  the algorithm adds $\vrtx \to \vhead_j$ for every $\vrtx$ reading-from $\vtail_i$
  (line~\ref{li:wmid}--\ref{li:wend}).
  And $\vrtx \to \vhead_j$ together with
  the edge $\vtail_i \to \vrtx$ (added in line~\ref{li:addwr})
  establish $\vtail_i \rightsquigarrow \vhead_j$.
  In case (ii),
  \sys's algorithm adds an edge $\vtail_i \to \vhead_j$ to $A$
  (line~\ref{li:noread}), and there is no $\vrtx$ in this case.
  Similarly, by switching $i$ and $j$
  in the above reasoning (except we don't care about the reads in this case),
  edges in $B$ establish $\vtail_j \rightsquigarrow \vhead_i$.

  Second, because $q$ is compatible with $Q(h)$,
  it either (1) contains $A$:
  \begin{align*}
    \vtail_i/\vrtx &\rightsquigarrow \vhead_j &\text{[proved in the first half]}\\
                   &\rightsquigarrow \vwtx_j  &\text{[Claim~\ref{claim:wwinchain}; $\vwtx_j \in \vwlist_j$]}
  \end{align*}
  or else (2) contains $B$:
  \begin{align*}
    \vwtx_j  & \rightsquigarrow \vtail_j  &\text{[Claim~\ref{claim:wwinchain}; $\vwtx_j \in \vwlist_j$]}\\
             & \rightsquigarrow \vhead_i &\text{[proved in the first half]}\\
             &\rightsquigarrow \vtx  &\text{[Claim~\ref{claim:wwinchain}; $\vtx \in \vwlist_i$]}
  \end{align*}
    The
  argument still holds if $\vwtx_j = \vhead_j$ in case (1): remove the
  second step in (1). Likewise, if
  $\vwtx_j = \vtail_j$ in case (2), remove the first step in (2). 
\end{proof}

\begin{claim}
For any pair of transactions $(\vwtx,\,\vrtx)$
where $\vrtx$ reads a key $k$ from $\vwtx$
and any other transaction $\vwtx'$ that writes to $k$,
graph $q$ has either $\vwtx' \rightsquigarrow \vwtx$ or $\vrtx \rightsquigarrow \vwtx'$.
\label{claim:wrandw}
\end{claim}

\begin{proof}

By Claim~\ref{claim:onechain},
$\vwtx$ must appear in some \chain $\vwlist_i$ on $k$.
Each of the three transactions ($\vwtx$, $\vrtx$, and $\vwtx'$)
has two possibilities relative to $\vwlist_i$:

\begin{myenumerate}
  \item $\vwtx$ is either the tail or non-tail of $\vwlist_i$.
  \item $\vrtx$ is either in $\vwlist_i$ or not.
  \item $\vwtx'$ is either in $\vwlist_i$ or not.
\end{myenumerate}

In the following, we enumerate all combinations
of the above possibilities and prove the claim
in all cases.

\begin{myitemize2}

\item $\vwtx = \vtail_i$.

Then, $\vrtx$ is not in $\vwlist_i$. (If $\vrtx$ is in $\vwlist_i$, its enclosing transaction
would have to be subsequent to $\vwtx$ in $\vwlist_i$, which is a
contradiction, since $\vwtx$ is last in the chain.)

  \begin{myitemize2}

    \item $\vwtx' \in \vwlist_i$.

      Because $\vwtx$ is the tail, $\vwtx'$ appears before $\vwtx$ in $\vwlist_i$.
      Thus, $\vwtx' \rightsquigarrow \vwtx$ in $q$ (Claim~\ref{claim:wwinchain}).

    \item $\vwtx' \not\in \vwlist_i$.

      By invoking Claim~\ref{claim:wwtwochains} for $\vwlist_i$ and
      $\vwtx'$,
      $q$ either has (1) paths from each read ($\vrtx$ is one of them)
      reading-from $\vtail_i$ ($=\vwtx$) to $\vwtx'$,
      therefore $\vrtx \rightsquigarrow \vwtx'$.
      Or else $q$ has (2) paths from $\vwtx'$
      to every transaction in $\vwlist_i$,
      and $\vwtx \in \vwlist_i$, thus $\vwtx' \rightsquigarrow \vwtx$.

  \end{myitemize2}

\item $\vwtx \not= \vtail_i \land \vwtx \in \vwlist_i$.

  \begin{myitemize2}

    \item $\vrtx \in \vwlist_i$.

      Because $\vwlist_i$ is a sequence of consecutive writes on $k$
      (Lemma~\ref{lemma:longestchain})
      and $\vrtx$ reads $k$ from $\vwtx$,
      $\vrtx$ is the \adjacentwrite of $\vwtx$. Therefore, $\vrtx$ appears
      immediately after $\vwtx$ in $\vwlist_i$.

      \begin{myitemize2}

        \item $\vwtx' \in \vwlist_i$.

          Because $\vrtx$ appears immediately after $\vwtx$ in $\vwlist_i$,
          $\vwtx'$ either appears before $\vwtx$ or after $\vrtx$.
          By Claim~\ref{claim:wwinchain}, there is
          either $\vwtx' \rightsquigarrow \vwtx$ or $\vrtx \rightsquigarrow \vwtx'$ in $q$.

        \item $\vwtx' \not\in \vwlist_i$.

          By invoking Claim~\ref{claim:wwtwochains} for $\vwlist_i$ and $\vwtx'$,
          $q$ either has (1) $\vtail_i \rightsquigarrow \vwtx'$,
          together with $\vrtx \rightsquigarrow \vtail_i$ (or $\vrtx = \vtail_i$) by Claim~\ref{claim:wwinchain},
          therefore $\vrtx \rightsquigarrow \vwtx'$.
          Or else $q$ has (2) $\vwtx' \rightsquigarrow \vwtx$
          ($\vwtx'$ has a path to every transaction in $\vwlist_i$, and $\vwtx \in \vwlist_i$).

      \end{myitemize2}

    \item $\vrtx \not\in \vwlist_i$.

      If $\vrtx \not\in \vwlist_i$, because of \textsc{\InferRWEdges}
      (line~\ref{li:rwrels}), $\vrtx$ has an edge (in the known graph, hence in $q$)
      to the transaction that immediately follows $\vwtx$ in $\vwlist_i$,
      denoted as $\vwtx^*$
      (and $\vwtx^*$ must exist because $\vwtx$ is not the tail of the \chain).

      \begin{myitemize2}

        \item $\vwtx' \in \vwlist_i$.

          Because $\vwtx^*$ appears immediately after $\vwtx$ in $\vwlist_i$,
          $\vwtx'$ either appears before $\vwtx$ or after $\vwtx^*$.
          By Claim~\ref{claim:wwinchain},
          $q$ has either $\vwtx' \rightsquigarrow \vwtx$
          or $\vwtx^* \rightsquigarrow \vwtx'$
          which, together with edge $\vrtx \to \vwtx^*$ from \textsc{\InferRWEdges},
          means $\vrtx \rightsquigarrow \vwtx'$.

        \item $\vwtx' \not\in \vwlist_i$.

          By invoking Claim~\ref{claim:wwtwochains} for $\vwlist_i$ and $\vwtx'$,
          $q$ has either (1) $\vtail_i \rightsquigarrow \vwtx'$
          which, together with $\vrtx \to \vwtx^*$ (from \textsc{\InferRWEdges})
          and $\vwtx^* \rightsquigarrow \vtail_i$ (Claim~\ref{claim:wwinchain}),
          means $\vrtx \rightsquigarrow \vwtx'$.
          Or else $q$ has (2) $\vwtx' \rightsquigarrow \vwtx$
          ($\vwtx'$ has a path to every transaction in $\vwlist_i$, and $\vwtx \in \vwlist_i$).
  \end{myitemize2}
  \end{myitemize2}
\end{myitemize2}
\end{proof}

By topologically sorting $q$, we get a serial schedule $\hat{s}$.
Next, we prove $h$ matches $\hat{s}$, hence $h$ is serializable (Definition~\ref{def:serializable}).

Since $h$ and $\hat{s}$ have the same set of transactions
(because $q$ has the same transactions as the known graph of $Q(h)$, and
thus also the same as $h$),
we need to prove
only that for every read that reads from a write in $h$,
the write is the most recent write to that read in $\hat{s}$.

First, for every pair of transactions $(\vwtx,\,\vrtx)$ such that $\vrtx$ reads a key $k$ from $\vwtx$
in $h$, $q$ has an edge $\vwtx \to \vrtx$
(added to the known graph in line~\ref{li:addwr}); thus
$\vrtx$ appears after $\vwtx$ in $\hat{s}$ (a topological sort of $q$).
Second, by invoking Claim~\ref{claim:wrandw} for $(\vwtx,\,\vrtx)$,
any other transaction writing to $k$ is either ``topologically prior'' to
$\vwtx$ or ``topologically subsequent'' to $\vrtx$.
This ensures that, the most recent write
of $\vrtx$'s read (to $k$) belongs to $\vwtx$ in $\hat{s}$,
hence $\vrtx$ reads the value of $k$ written by $\vwtx$ in $\hat{s}$ as it does in $h$.
This completes the proof.
\end{proof}

\newpage
\section{\fontsize{13}{13} \selectfont Garbage collection correctness proof}
\label{sec:appxb}

\begin{figure*}[p!]

\newcommand{\accept}{{\small\textsc{accept}}}
\newcommand{\reject}{{\small\textsc{reject}}}
\newcommand{\mtab}{\hspace{\algorithmicindent}}
\newcommand{\mmtab}{\mtab\mtab}
\newcommand{\mmmmtab}{\mmtab\mmtab}

\footnotesize
\rule{\linewidth}{.08em}

\begin{minipage}[t][][t]{.48\textwidth} 
\begin{algorithmic}[1]
  \Procedure{\VerifySerializability}{\t{history\_stream}}
  \State $i \gets 0$; $\vg \gets$ empty graph
  \State $\wwpairs \gets$ empty map $\{\langle \textrm{Key}, \textrm{Tx}\rangle \to \textrm{Tx}\}$
  \State $\readfrom \gets$ empty map $\{\langle \textrm{Key}, \textrm{Tx}\rangle \to
                                    \textrm{Set}\langle\textrm{Tx}\rangle\}$
  \While{\t{True}}
    \State $i \gets i + 1$
    \State $h_i \gets $ fetch a continuation from \t{history\_stream} \label{li2:historyinit}
    \State $\vg,\readfrom,\wwpairs \gets$
    \State \mmtab $\CreateKnownGraphTwo(\vg, \readfrom,\,\wwpairs,\,h_i)$ \label{li2:creategraph2}

    \State

    \State $\vg,\,\t{psccs} \gets \EncodeAndSolve(\vg,\,\readfrom,\,\wwpairs)$

    \State

    \State $\vg,\readfrom,\wwpairs \gets \GarbageCollection(\vg,\,\t{psccs})$
  \EndWhile

  \EndProcedure

  \State

  \Procedure{\CreateKnownGraphTwo}{\vg,\,\readfrom,\, \wwpairs,\,$h$}
    \For {transaction \vtx in $h$} 
        \State $\vg.\textrm{Nodes} \pluseq \vtx$
        \State $\vtx.\ffrozen = \t{False}$
        \For {read operation \rop in \vtx}
            \State \textbf{if} $\rop.\textrm{read\_from\_tx}$ not in $\vg$: \ \reject
            \State $\vg.\textrm{Edges} \pluseq (\rop.\textrm{read\_from\_tx},\ \vtx)$   \label{li2:edge1}
            \State $\readfrom[\langle \rop.\textrm{key},\,\rop.\textrm{read\_from\_tx} \rangle] \pluseq \vtx$
        \EndFor
        \For {all Keys $\vkey$ that are both read and written by $\vtx$}
           \State \rop $\gets$ the operation in $\vtx$ that reads $\vkey$
                \If {$\wwpairs[\langle \vkey,\,\rop.\textrm{read\_from\_tx} \rangle] \neq \t{null}$}
                  \State \reject
                \EndIf
                \State $\wwpairs[\langle \vkey,\,\rop.\textrm{read\_from\_tx} \rangle] \gets \vtx$
        \EndFor
    \EndFor
    \State
    \For {each client $c$}  \label{li2:clientstart}
      \State $\t{list}_c \gets$ an ordered list of trsanctions issued by $c$ in $\vg$
      \For {$i$ in $[0,\,\textrm{length}(\t{list}_c)-2]$}
        \State $\vg.\textrm{Edges} \pluseq (\t{list}_c[i],\t{list}_c[i+1])$ \label{li2:clientend} \label{li2:edge2}
      \EndFor
    \EndFor
    \State
    \State \Return $\vg,\,\readfrom,\,\wwpairs$
  \EndProcedure

  \State

  \Procedure{\EncodeAndSolve}{$\vg,\,\readfrom,\,\wwpairs$}
    \State $\vconstraints  \gets \GenConstraint(g,\,\readfrom,\,\wwpairs)$
    \State $\vconstraints,\,\vg \gets \PruneConstraintsTR(\vconstraints,\,\vg)$

    \State
    \State $\t{formula} \gets \EncodeSMT(\vconstraints,\,\vg)$
    \State $\t{solved} \gets $ MonoSAT.solve(\t{formula})
    \State \textbf{if} not \t{solved}: \reject \label{li2:reject1}

    \State
    \State $\t{psccs} \gets \GenIndependentClusters(\vconstraints,\,\vg)$
    \State \Return \vg, \t{psccs}
  \EndProcedure

  \State

  \Procedure{\GenIndependentClusters}{$\vconstraints,\,\vg$}
    \State $g' \gets g$
    \State $g'.\textrm{Edges} \gets g'.\textrm{Edges}\ \cup$ \{all edges in $\vconstraints$\}
    \State $\t{psccs} \gets \ComputeStronglyConnectedComponents(g')$
    \State \Return \t{psccs}
  \EndProcedure

  \State

  \Procedure{\GarbageCollection}{\vg, \t{psccs}}
    \State $\epochagree \gets \AssignEpoch(\vg)$
    \State $\SetFrozen(\vg,\,\epochagree)$
    \State $\SetObsolete(\vg,\, \epochagree)$
    \State $\SetRemovable(\vg,\, \t{psccs})$
    \State \Return $\SafeDeletionTwo(\vg,\,\readfrom,\,\wwpairs)$
  \EndProcedure

\algstore{linenum}
\end{algorithmic}

\end{minipage}
\hspace{1ex}
\begin{minipage}[t][][t]{.50\textwidth} 
\begin{algorithmic}[1]
\algrestore{linenum}

  \Procedure{\AssignEpoch}{\vg} \label{li2:defepoch}
  \State $\t{epoch\_num}  \gets 0$
  \State $\t{topo\_tx} \gets \textrm{TopologicalSort}(\vg)$

  \For {\vtx in \t{topo\_tx}}   \label{li2:startwfence}
    \If {\vtx writes to key \t{``EPOCH''}}
      \State $\vtx.\fepoch \gets \t{epoch\_num}$
      \State $\t{epoch\_num} \gets \t{epoch\_num} + 1$  \label{li2:endwfence}
    \EndIf
  \EndFor

  \For {\vtx in \t{topo\_tx}}  \label{li2:startrfence}
    \If {\vtx reads but not writes key \t{``EPOCH''}}
      \State $\vtx.\fepoch \gets \vtx.\textrm{read\_from\_tx}.\fepoch$ \label{li2:endrfence}
    \EndIf
  \EndFor

  \State

  \State $\epochagree \gets \inf$
  \For {each client $c$}  \label{li2:startnormal}
    \State $\t{list}_c \gets$ an ordered list of trsanctions issued by $c$ in $\vg$
    \State $\t{rlist}_c \gets$ reversed ordered list of $\t{list}_c$
    \State $\t{cur\_epoch} \gets \inf$
    \For {\vtx in $\t{rlist}_c$}
      \If {\vtx touches key \t{``EPOCH''}}
        \If {$\t{cur\_epoch} = \inf$}
          \State $\epochagree \gets \textrm{min}(\epochagree,\, \vtx.\fepoch)$ \label{li2:trackepochagree}
        \EndIf
        \State $\t{cur\_epoch} \gets$ $\vtx.\fepoch$
      \Else:
        \State $\vtx.\fepoch = (\t{cur\_epoch}=\inf\ ? \ \inf : \t{cur\_epoch} - 1)$ \label{li2:endnormal} \label{li2:assignnormal}
      \EndIf
    \EndFor
  \EndFor

  \State\Return $\epochagree$

  \EndProcedure

  \State

  \Procedure{\SetFrozen}{\vg, $\epochagree$} \label{li2:deffrozen}
    \State $\vfe \gets \epochagree - 2$

    \For {\vtx in \vg}
      \If {$\vtx.\fepoch \le \vfe$ \textbf{and}\\
        \mmtab\ \ all predecssors of \vtx in \vg have epoch $\le \vfe$}
        \State $\vtx.\ffrozen \gets \t{True}$
      \EndIf
    \EndFor
  \EndProcedure

  \State

  \Procedure{\SetObsolete}{\vg, $\epochagree$}
    \State $\vfe \gets \epochagree - 2$

    \For {\vtx in \vg}
    \If {$\vtx.\fepoch \le \vfe$}
      \State $\t{\obsolete} \gets \t{True}$
      \For {Key \vkey{} that $\vtx$ writes}
        \If {$\not\exists tx_j, \ s.t.\ \vtx \rightsquigarrow \vtx_j
             \land \vtx_j.\fepoch \le \vfe \land \vtx_j \textrm{ writes  } \vkey$}
        \State $\t{\obsolete} \gets \t{False}$; \textbf{break}
        \EndIf
      \EndFor
      \State $\vtx.\fobsolete \gets \t{\obsolete}$
    \EndIf
    \EndFor
  \EndProcedure

  \State

  \Procedure{\SetRemovable}{$\vg,\, \t{psccs}$}
  \For {\vtx in \vg}
    \If {$\vtx.\ffrozen = \t{True}$}
      \If {$\vtx.\fobsolete = \t{True}$ \textbf{or} \vtx is read-only}
        \State $\vtx.\fdelcandcand \gets \t{True}$
      \EndIf
    \EndIf
  \EndFor
  \For {\t{pscc} in \t{psccs}}
    \If {$\forall \vtx \in \t{pscc},\ \vtx.\fdelcandcand = \t{True}$}
      \For {\vtx in \t{pscc}}
        \State $\vtx.\fdelcand \gets \t{True}$
      \EndFor
    \EndIf
  \EndFor
  \EndProcedure

  \State

  \Procedure{\SafeDeletionTwo}{$\vg,\, \readfrom,\, \wwpairs$}
  \For {\vtx in \vg}
    \If {$\vtx.\fdelcand = \t{True}$ \textbf{and} $\vtx$ doesn't touch key \t{``EPOCH''}}
      \label{li2:realdelete} \label{li2:nodeletefence}
      \State $\vg.\textrm{Nodes} \subeq \vtx$
      \State $\vg.\textrm{Edges} \subeq \{\textrm{edges with \vtx as one endpoint}\}$
      \State $\readfrom \subeq \{\textrm{tuples containing \vtx}\}$
      \State $\wwpairs \subeq \{\textrm{tuples containing \vtx}\}$
    \EndIf
  \EndFor
  \State \Return $\vg,\readfrom,\wwpairs$
  \EndProcedure

\end{algorithmic}
\end{minipage}

\rule{\linewidth}{.08em}
\caption{\sys's algorithm for verification in rounds.}
\label{fig:algocodefull}
\end{figure*}
 
\subsection{Verification in rounds}

Besides the ``one-shot verification'' described in \S\ref{s:search}
and Appendix~\ref{sec:appxa},
\sys also works for online verification and does verification in rounds
(pseudocode is described in Figure~\ref{fig:algocodefull}).
In each round, \sys's verifier checks serializability on the transactions
that have been received.
In the following, we define terms used in
the context of verification in rounds:
\textit{complete history},
\textit{continuation},
\textit{\clientser}, and 
\textit{\depinfo}.

\heading{Complete history and continuation.}
A \textit{complete history} is a prerequisite of checking serializability.
If a history is incomplete and some of the transactions are unknown,
it is impossible to decide whether this history is serializable.

\begin{definition2}[Complete history]
A \emph{complete history} is a history where all read operations read from
the write operations in the same history.
\label{def:completehistory}
\end{definition2}

For verification in rounds, in each round,
\sys's verifier receives a set of transactions
that may read from the transactions in prior rounds.
We call such newly coming transactions a
continuation~\cite{hadzilacos89deleting}.

\begin{definition2}[Continuation]
A \emph{continuation} $r$ of a complete history $h$ is a set of transactions
in which all the read operations read from transactions in either $h$ or $r$.
\label{def:continuation}
\end{definition2}

We denote the combination of a complete history $h$ and its continuation $r$
as $h \circ r$. By Definition~\ref{def:completehistory}, $h \circ r$ is also a complete history.
Also, we call the transactions in future continuations of the current history as
  \textit{future transactions}.

In the following discussion, we assume that the transactions received in each
  round are continuations of the known history.
However, in practice, the received transactions may not form a complete history
and \sys's verifier has to adopt a preprocessing phase to filter out the
transactions whose predecessors are unknown and save them for future rounds
(for simplicity, such preprocessing is omitted in Figure~\ref{fig:algocodefull}
which should have happen in line~\ref{li2:historyinit}.)

\heading{\CF{\clientser}.}
As mentioned in \S\ref{s:fence},
transactions' serialization order in practice should respect their causality
which, in our context, is the transaction issuing order by clients (or session).
So, if a history satisfies serializability (Definition~\ref{def:serializable})
and the corresponding serial schedule preserves the transaction issuing order,
we say this history is \textit{\clientser}, defined below.

\begin{definition2}[\CF{\clientser} history]
A \emph{\clientser history}
is a history that matches a serial schedule $\hat{s}$,
such that $\hat{s}$ preserves the transaction issuing order for any client.
\label{def:clientser}
\end{definition2}

Notice that \sys requires that each client is single-threaded and blocking (\S\ref{s:background}).
So, for one client, its transaction issuing order is the order
seen by the corresponding history collector
(one client connects to one collector).
The verifier also knows such order by referring to the history fragments.

\heading{\CF{\depinfo}.}
In the following, we define a helper notion \textit{\depinfo}
that contains the information of a history that has parsed by \sys's algorithm.
An \depinfo $e$ of a history $h$ is a tuple $(\vg,\,\readfrom,\,\wwpairs)$
generated by \textsc{\CreateKnownGraphTwo}
in Figure~\ref{fig:algocodefull}, line~\ref{li2:creategraph2}.

Notice that in \sys's algorithm %
each round reuses the \depinfo $e=(\vg,\,\readfrom,\,\wwpairs)$
from the preceding round.
In the following, we use $\E(e, h)$ 
to represent
$\textsc{\CreateKnownGraphTwo}(\vg,\readfrom,\wwpairs,h)$.
And, we use $\E(h)$ as a shortened form of $\E(\emptyset,h)$.

\begin{fact}
For a complete history $h$ and its continuation $r$,
$\E(h \circ r) = \E(\E(h), r)$.
Because $\readfrom$ and $\wwpairs$ only depend on the information
carried by each transaction,
and this information is the same
no matter whether processing $h$ and $r$ together or separately.
For client ordering edges (Figure~\ref{fig:algocodefull},
line~\ref{li2:clientstart}--\ref{li2:clientend}),
since they are the ordering of transactions seen by the collectors,
the edges remain the same as well.

\label{fact:identicaldepinfo}
\end{fact}

\begin{definition2}[Deletion from an \depinfo]
A \emph{deletion} of a transaction $\vtx_i$ from an \depinfo $\E(h)$ is to
(1) delete the vertex $\vtx_i$ and edges containing $\vtx_i$
from the known graph $\vg$ in $\E(h)$; and
(2) delete tuples that include $\vtx_i$ from $\readfrom$ and $\wwpairs$.
\label{def:del}
\end{definition2}

We use $\E(h) \del \vtx_i$ to denote deleting $\vtx_i$
from \depinfo $\E(h)$.

\subsection{\CF{\polyg}, \mypolyg, pruned \polyg, and pruned \mypolyg}

Notice that an \depinfo contains all information from a history.
So, instead of building from a history,
both \polyg (\S\ref{subsec:bruteforce}) and \mypolyg (Definition~\ref{def:mypolyg})
can be built from an \depinfo.

Specifically, constructing a \polyg $(V,\,E,\,C)$
from an \depinfo $\E(h)$ works as follows (which is similar to what is in \S\ref{subsec:bruteforce}):
\begin{myitemize2}

  \item $V$ are all vertices in $\E(h).\vg$.

  \item $E = \{(\vtx_i,\,\vtx_j)\ |\ \langle \,\textrm{\_},\,\vtx_i,\,\vtx_j\, \rangle \in \E(h).\readfrom \}$;
  that is, $\vtx_i \wrarrow{x} \vtx_j$, for some $x$.

  \item $C = \{ \langle\,(\vtx_j,\, \vtx_k),\  (\vtx_k,\, \vtx_i)\, \rangle \mid
    (\vtx_i \wrarrow{x} \vtx_j)\, \land\,\\ 
    \hspace*{3em}(\vtx_k\textrm{\ writes to $x$}) \land \vtx_k \ne \vtx_i \land \vtx_k \ne \vtx_j\}$.
\end{myitemize2}
We denote the \polyg generated from \depinfo $\E(h)$ as $\Po(\E(h))$.

Since constructing an \depinfo is part of \sys's algorithm,
it is natural to construct a \mypolyg from an \depinfo,
which works as follow:
assign \mypolyg's known graph to be $\E(h).\vg$
and generate constraints by invoking
\[
\textsc{\GenConstraint}(\E(h).\vg,\,\E(h).\readfrom,\,\E(h).\wwpairs).
\]
We denote the \mypolyg generated from \depinfo $\E(h)$ as $\Qo(\E(h))$.

In order to test \clientserty, %
we add cliens' transaction issuing order to \polyg and \mypolyg
by inserting edges for transactions that are issued by
the same client. We call such edges \textit{client ordering edges} (short as CO-edges).
For each client, these CO-edges
point from one transaction to its immediate next transaction
(Figure~\ref{fig:algocodefull}, line~\ref{li2:clientstart}--\ref{li2:clientend}).

\begin{lemma2}
For a serializable history $h$,
a serial schedule that $h$ matches is some topological sort of an acyclic graph
that is compatible with the \polyg without CO-edges
(and \mypolyg without CO-edges)
of $h$; and topological sorting
an acyclic graph that is compatible with the \polyg without CO-edges
(and \mypolyg without CO-edges) of $h$
results in a serial schedules that $h$ matches.
\label{lemma:toposortschedule}
\end{lemma2}

\begin{proof}

First, we prove the Lemma for \polyg (then later \mypolyg).
In Papadimitriou's proof~\cite{papadimitriou79serializability} (\S3, Lemma~2),
when proving that $h$ is serializable $\Rightarrow$ \polyg is acyclic,
the proof constructs an acyclic compatible graph
according to a serial schedule, which means that this serial schedule
is a topological sort of the constructed comptible graph.
On the other hand, when proving that \polyg is acyclic $\Rightarrow$ $h$ is serializable,
the proof gets the serial schedule from topological sorting an ayclic compatible graph.

Similarly, for \mypolyg, in Thoerem~\ref{theorem:=>} (Appendix~\ref{sec:appxa}),
the proof constructs an acyclic compatible graph from a serial schedule;
and in Thoerem~\ref{theorem:<=}, the proof topologically sorts an acyclic compatible graph
to generate a serial schedule.
\end{proof}

In the following, all {\polyg}s $\Po(\E(h))$ and {\mypolyg}s $\Qo(\E(h))$
include client ordering edges by default.

\begin{lemma2}
Given a complete history $h$ and its \depinfo $\E(h)$,
the following logical expressions are equivalent:
\begin{myenumerate}[  {(}1{)}]
 \item history $h$ is \clientser.
 \item \polyg $\Po(\E(h))$ is acyclic.
 \item \mypolyg $\Qo(\E(h))$ is acyclic.
\end{myenumerate}
\label{lemma:polygraphser}
\end{lemma2}

\begin{proof}
First, we prove that $(1) \iff (2)$.

$(1) \implies (2)$:
Because $h$ is \clientser, there exists a serial schedule $\hat{s}$ that
$h$ matches and preserves the transaction issuing order of clients.
By Lemma~\ref{lemma:toposortschedule}, $\hat{s}$ is one of the topological sorts of some graph
$\hat{g}$ that is compatible with the \polyg without client ordering edges.
By adding client ordering edges to $\hat{g}$, we get $\hat{g+}$.
Graph $\hat{g+}$ is still compatible with $\Po(\E(h))$,
because edges of $\hat{g+}$ are a subset of the total ordering of $\hat{s}$
($\hat{s}$ preserves the clients' transaction issuing order).
Thus, $\hat{g+}$ is acyclic, hence $\Po(\E(h))$ is also acyclic.

$(2) \implies (1)$:
Because \polyg $\Po(\E(h))$ is acyclic, there exists a compatible graph
$\hat{g+}$ that is acyclic.
By removing all the client ordering edges
from $\hat{g+}$, we have $\hat{g}$ which
is compatible with the \polyg without client ordering edges.
Because $\hat{g}$ has the same nodes but less edges than $\hat{g+}$,
a topological sort $\hat{s}$ of $\hat{g+}$ is also a topological sort of $\hat{g}$.
By Lemma~\ref{lemma:toposortschedule}, $\hat{s}$ is a serial schedule that $h$ matches.
Because $\hat{g+}$ has client ordering edges,
$\hat{s}$ preserves transaction issuing order of clients,
hence $h$ is \clientser.

Similarly, we can prove $(1) \iff (3)$ by replacing \polyg (with and without
client ordering edges) with \mypolyg (with and without client ordering edges).
\end{proof}

\heading{Pruned (\sys) \polyg.}
Given a \mypolyg $\Qo(\E(h)) = (\vg,\,\vconstraints)$,
we call the \mypolyg after invoking $\textsc{\PruneConstraintsTR}(\vconstraints,\,\vg)$
(Figure~\ref{fig:algocodemain}, line~\ref{li:prunefunc}) as a \textit{pruned
\mypolyg}, denoted as $\Qp(\E(h))$.
Similarly, if we treat a constraint in a \polyg
(for example $\langle \vtx_i \to \vtx_j, \vtx_j \to \vtx_k \rangle$)
as a constraint in \mypolyg but with each edge set having only one edge
($\langle \{\vtx_i \to \vtx_j\}, \{\vtx_j \to \vtx_k\} \rangle$),
then we can apply $\textsc{\PruneConstraintsTR}$ to
a \polyg $\Po(\E(h))$
and get a \textit{pruned \polyg}, denoted as $\Pp(\E(h))$.

Note that $\Qo(\E(h))$ and $\Qp(\E(h))$ are
what \sys's algorithm actually creates;
$\Po(\E(h))$ and $\Pp(\E(h))$ are helper notions
for the proof only---they are not actually materialized.

\begin{lemma2}
$\Qo(\E(h))$ is acyclic $\iff$ $\Qp(\E(h))$ is acyclic,
and
$\Po(\E(h))$ is acyclic $\iff$ $\Pp(\E(h))$ is acyclic.
\label{lemma:prunedq}
\end{lemma2}

\begin{proof}
First, we prove $\Qo(\E(h))$ is acyclic $\iff$ $\Qp(\E(h))$ is acyclic.

\vspace{1ex}
\noindent
``$\Rightarrow$''.
To begin with, we prove that pruning one constraint $\langle A,B \rangle$
from $\Qo(\E(h))$ does not affect the acyclicity of the remaining \mypolyg.
If so, pruning multiple constraints on an acyclic \mypolyg still results in
an acyclic \mypolyg.

Now, consider the constraint $\langle A,B \rangle$ ($A$
and $B$ are edge sets) that has been pruned in $\Qo(\E(h))$,
and assume it gets pruned because of an edge
$(\vtx_i,\vtx_j) \in A$ such that $\vtx_j \rightsquigarrow \vtx_i$ in the known graph.

Because $\Qo(\E(h))$ is acyclic, there exists a compatible graph $\hat{g}$ that is acyclic.
For the binary choice of $\langle A,B \rangle$,
$\hat{g}$ must choose $B$;
otherwise $\hat{g}$ would have a cycle due to the edge $(\vtx_i,\vtx_j)$ in $A$
and $\vtx_j \rightsquigarrow \vtx_i$ in the known graph.
And, \textsc{\PruneConstraintsTR} (Figure~\ref{fig:algocodemain}, line~\ref{li:beginprune}--\ref{li:endprune})
does the same thing---add edges in $B$ to $\Qp(\E(h))$'s known graph,
when the algorithm detects edges in $A$ conflict with the known graph.
Hence, $\hat{g}$ is compatible with $\Qp(\E(h))$ and $\Qp(\E(h))$ is acyclic.

\vspace{1ex}
\noindent
``$\Leftarrow$''.
Because $\Qp(\E(h))$ is acyclic, there exists a compatible graph $\hat{g'}$ that is acyclic.
Consider all constraints in $\Qo(\E(h))$: for the pruned constraints,
$\hat{g'}$ contains edges from one of the two edge sets in the constraint;
for those constraints that is not pruned, $\Qp(\E(h))$ has them and
$\hat{g'}$ selects one edge set from each of them ($\hat{g'}$ is compatible with $\Qp(\E(h))$).
Thus, $\hat{g'}$ contains one edge set from all constraints in $\Qo(\E(h))$,
so it is compatible with $\Qo(\E(h))$.
Plus, $\hat{g'}$ is acyclic, hence $\Qo(\E(h))$ is acyclic.

Now we prove that $\Po(\E(h))$ is acyclic $\Leftrightarrow$ $\Pp(\E(h))$ is acyclic.
Because the constraint in a \polyg is a specialization of the constraint in a \mypolyg
(each edge set only contains one edge),
the above argument is still true
by replacing $\Qo(\E(h))$, $\Qp(\E(h))$ to $\Po(\E(h))$, $\Pp(\E(h))$ respectively.
\end{proof}

\begin{lemma2}
Given a history that is \clientser,
for any two transactions $\vtx_i$ and $\vtx_j$,
$\vtx_i \rightsquigarrow \vtx_j$ in the known graph of $\Pp(\E(h))$
$\iff$
$\vtx_i \rightsquigarrow \vtx_j$ in the known graph of $\Qp(\E(h))$
\label{lemma:sameknowngraph}
\end{lemma2}

\begin{proof} Because $h$ is \wellformed, $\E(h)$'s known graph is acyclic
and both $\Pp(\E(h))$ and $\Qp(\E(h))$'s known graphs
contain edges in $\E(h)$'s known graph.

\vspace{1ex}
\noindent
``$\Rightarrow$''.
We prove that for any edge $\vtx_a \to \vtx_b$ in path $\vtx_i \rightsquigarrow \vtx_j$
of $\Pp(\E(h))$ ($\vtx_a$ might be $\vtx_i$ and $\vtx_b$ might be $\vtx_j$),
there always exists $\vtx_a \rightsquigarrow \vtx_b$ in $\Qp(\E(h))$.
In $\Pp(\E(h))$ and $\Qp(\E(h))$'s known graph,
there are four types of edges.
Three of them---reading-from edges (Figure~\ref{fig:algocodefull}, line~\ref{li2:edge1}),
anti-dependency edge (Figure~\ref{fig:algocodemain}, line~\ref{li:addrw}),
and client order edges (Figure~\ref{fig:algocodefull}, line~\ref{li2:edge2})---are
captured by $\E(h)$'s known graph which shared by both $\Pp(\E(h))$ and $\Qp(\E(h))$.
Hence, if $\Pp(\E(h))$'s known graph has $\vtx_a \to \vtx_b$, $\Qp(\E(h))$
also has it.

Next, we prove that when edge $\vtx_a \to \vtx_b$ is the last type: edges added by
\textsc{\PruneConstraintsTR} (Figure~\ref{fig:algocodemain},
line~\ref{li:addprune1},\ref{li:addprune2}) in $\Pp(\E(h))$,
$Qp(\E(h))$'s known graph also has $\vtx_a \rightsquigarrow \vtx_b$.

Consider a constraint in $\Pp(\E(h))$ is $\langle \rone \to \wtwo, \wtwo \to \wone \rangle$
where $\wone$ and $\wtwo$ writes to the same key; $\rone$ reads this key from $\wone$.
For $\wone$ and $\wtwo$ in $\Qp(\E(h))$, because they write the same key,
they are either (1) in the same \chain, or else (2) $\wtwo$ and $\wtwo$ belong
to two \chains and there is a constraint about these two \chains.

Given $\vtx_a \to \vtx_b$ is added by pruning a constraint in $\Pp(\E(h))$,
there are two possibilities:
\begin{myitemize2}

\item $\vtx_a \to \vtx_b$ is $\rone \to \wtwo$, which means $\wone \rightsquigarrow \wtwo$
  (otherwise, the constraint would not be pruned). In $\Qp(\E(h))$, for above case (1),
  because $\wone \rightsquigarrow \wtwo$, $\rone$ reading from $\wone$
  has paths to the successive transactions (including $\wtwo$) in the \chain;
  for (2), because $\wone \rightsquigarrow \wtwo$, this constraint
  would be pruned and $\rone \rightsquigarrow \wtwo$.

\item $\vtx_a \to \vtx_b$ is $\wtwo \to \wone$, which means $\wtwo \rightsquigarrow \rone$.
  In $\Qp(\E(h))$, for (1), $\wtwo$ must appear earlier than $\wone$ in the
  \chain (hence $\wtwo \rightsquigarrow \rone$), because otherwise $\rone
  \rightsquigarrow \wtwo$, a contradiction;
  for (2), because $\wtwo \rightsquigarrow \rone$, the tail of \chain that $\wtwo$ is in
  has a path to the head of $\wone$'s \chain, hence $\wtwo \rightsquigarrow \wone$ in $\Qp(\E(h))$.

\end{myitemize2}

\vspace{1ex}
\noindent
``$\Leftarrow$''. Similarly, by swapping $\Pp(\E(h))$ and $\Qp(\E(h))$ in the above argument,
we need to prove that given a pruned constraint $\langle A,B \rangle$ in $\Qp(\E(h))$
which contains $\vtx_a \to \vtx_b$, there exists $\vtx_a \rightsquigarrow \vtx_b$
in $\Pp(\E(h))$.

Again, for a constraint $\langle A,B \rangle$ about two \chains $\vwlist_i$ and
$\vwlist_j$ in $\Qp(\E(h))$ ($\vhead_i$/$\vtail_i$ is the head/tail of
$\vwlist_i$; $\vrtx_i$ is a read transaction reads from $\vtail_i$).
There are two possibilities:
\begin{myitemize2}

\item $\vtx_a \to \vtx_b$ is $\vrtx_i \to \vhead_j$, which means $\vhead_i
\rightsquigarrow \vtail_j$. Consider the constraint $\langle \vrtx_i \to
\vhead_j, \vhead_j \to \vtail_i \rangle$ in $\Pp(\E(h))$.
Given that $h$ is serializable, two \chains must be schedule sequentially
and cannot overlap,
hence $\vtail_i \rightsquigarrow \vhead_j$.
Then, this constraint in $\Pp(\E(h))$ would be pruned and there is an edge
$\vrtx_i \to \vhead_j$.

\item $\vtx_a \to \vtx_b$ is $\vtail_i \to \vhead_j$, which means
$\vhead_i \rightsquigarrow \vtail_j$.
Call the second last transaction in $\vwlist_i$ $\vtx_k$.
Consider the constraint $\langle \vtail_i \to \vhead_j, \vhead_j \to \vtx_k \rangle$
($\vtail_i$ reads from $\vtx_k$).
Again, because $h$ is serializable, two \chains cannot overlap,
and $\vtx_k \rightsquigarrow \vhead_j$.
Thus, the constraint is pruned and $\Pp(\E(h))$ has $\vtail_i \to \vhead_j$.

\end{myitemize2}

\end{proof}

\heading{(Extended) \notwellformed history.}
Given that a \mypolyg $\Qo(\E(h))$ and a pruned \mypolyg $\Qp(\E(h))$ are equivalent in acyclicity,
we extend the definition of an \notwellformed history (Definition~\ref{def:notwellformed})
to use $\Qp(\E(h))$ which rules out more local malformations
that are not \clientser.

\begin{definition2}[An \notwellformed history]
\emph{An \notwellformed history} $h$ is a history that
either (1) contains a transaction that has multiple \adjacentwrites on one key,
or (2) has a cyclic known graph $\vg$ in $\Qp(\E(h))$.
\label{def:notwellformed2}
\end{definition2}

\begin{corollary2}
\sys rejects (extended) \notwellformed histories.
\label{corollary:algoreject}
\end{corollary2}

\begin{proof}
By Lemma~\ref{lemma:algoreject}, \sys rejects a history when
(1) it contains a transaction that has multiple \adjacentwrites one one key;
(2) If the known graph has a cycle in the known graph of $\Qp(\E(h))$,
\sys detects and rejects this history when checking acyclicity in the
constraint solver.
\end{proof}

\subsection{\CF{\polyscc}}
In this section, we define \textit{{\polyscc}s} (short as \pscc) which capture
the possible cycles that are generated from constraints.
Intuitively, if two transactions appears in one \pscc,
it is possible (but not certian) there are cycles between them;
but if these two transactions do not belong to the same \pscc,
it is impossible to have a cycle including both transactions.

\begin{definition2}[\CF{\polyscc}]
Given a history $h$ and its pruned \mypolyg $\Qp(\E(h))$,
the \emph{\polysccs} are the strongly connected components
of a directed graph that is the known graph
with all edges in the constraints added to it.
\label{def:polyscc}
\end{definition2}

\begin{lemma2}
In a history $h$ that is \wellformed,
for any two transactions $\vtx_i$ and $\vtx_j$ writing the same key,
if $\vtx_i \not\rightsquigarrow \vtx_j$ and
$\vtx_j \not\rightsquigarrow \vtx_i$ in the known graph of $\Qp(\E(h))$,
then $\vtx_i$ and $\vtx_j$ are in the same \pscc.
\label{lemma:samepscc}
\end{lemma2}

\begin{proof}
By Claim~\ref{claim:onechain}, each of $\vtx_i$ and $\vtx_j$ appears and only appears in one \chain
(say $\vwlist_i$ and $\vwlist_j$ respectively).
Because $\vtx_i \not\rightsquigarrow \vtx_j$ and
$\vtx_j \not\rightsquigarrow \vtx_i$, $\vwlist_i \neq \vwlist_j$.
\sys's algorithm generates a constraint for every pair of \chains on the same key
(Figure~\ref{fig:algocodemain}, line~\ref{li:chainpairs}),
so there is a constraint $\langle A,B \rangle$ for $\vwlist_i$ and $\vwlist_j$,
which includes $\vtx_i$ and $\vtx_j$.

Consider this constraint $\langle A,B \rangle$.
One of the two edge sets ($A$ and $B$) contains edges that establish a path
from the tail of $\vwlist_i$ to the head of $\vwlist_j$---either a direct edge
(Figure~\ref{fig:algocodemain}, line~\ref{li:noread}), or through a read
transaction that reads from the tail of $\vwlist_i$
(Figure~\ref{fig:algocodemain}, line~\ref{li:wend}).
Similarly, the other edge set establishes a path from the tail of $\vwlist_j$
to the head of $\vwlist_i$.
In addition, by Lemma~\ref{lemma:longestchain}, in each \chain, there is a path from
its head to its tail through the reading-from edges in the known graph
(Figure~\ref{fig:algocodemain}, line~\ref{li:addwr}).
Thus, there is a cycle involving all transactions of these two \chains.
By Definition~\ref{def:polyscc}, all the transactions in these two
\chains---including $\vtx_i$ and $\vtx_j$---are in one \pscc.
\end{proof}

\subsection{Fence transaction, epoch, \obsolete transaction, and frozen transaction}
\label{appxb:fence}

The challenge of garbage collecting transactions is that \clientserty
does not respect real-time ordering across clients and it is unclear to the verifier which
transactions can be safely deleted from the history (\S\ref{s:truncationhard}).
\sys uses \textit{fence transactions} and \textit{epochs}
which generate \textit{\obsolete transactions} and \textit{frozen transactions}
that address this challenge.

\heading{Fence transactions.}
As defined in \S\ref{s:fence}, \textit{fence transactions} are predefined
transactions that are periodically issued by each client and access a
predefined key called the \textit{epoch key}.
Based on the value read from the epoch key, a fence transaction
is either a \textit{write fence transaction} (short as \wfence) or a
\textit{read fence transaction} (short as \rfence):
\wfences read-and-modify the epoch key; and \rfences only read the epoch key.
In a complete history, we define that the fence transactions are \textit{\wellff} as follows.

\begin{definition2}[\CF{\wellff} fence transactions]
In a history, fence transactions are \emph{\wellff} when
(1) all write fence transactions are a sequence of consecutive writes to the
epoch key; and (2) all read fence transactions read from known write fence
transactions.
\label{def:wellfformedfence}
\end{definition2}

\begin{claim}
Given a complete history $h$ that is \wellformed,
fence transactions in $h$ are \wellff.
\label{claim:wellfformedfence}
\end{claim}

\begin{proof}
First, we prove that all \wfences are a sequence of consecutive writes.
Because the epoch key is predefined and reserved for fence transactions,
\wfences are the only transactions that update this key.
Given that \wfences read-modify-write the epoch key,
a \wfence read from either another \wfence or
the (abstract) initial transaction if the epoch key hasn't been created.

Given that $h$ is \wellformed, by Definition~\ref{def:notwellformed2}, there is no cycle in the known graph.
Hence, if we start from any \wfence and
repeatedly find the predecessor of current \wfence on the epoch key
(the predecessor is known because \wfences also read the epoch key),
we will eventually reach the initial transaction
(because the number of write fence transactions in $h$ is finite).
Thus, all \wfences and the initial transaction are connected
and form a tree (the root is the initial transaction).
Also, because $h$ is \wellformed, no write transaction has two successive writes on the same key.
So there is no \wfence that has two children in this tree,
which means that the tree is actually a list. And each node
in this list reads the epoch key from its preceding node and all of them write
the epoch key.
By Definition~\ref{def:consecutivewrites}, this list of \wfences is a sequence
of consecutive writes.

Second, because $h$ is a complete history, all \rfences read
from transactions in $h$. Plus, only \wfences update the epoch key,
so \rfences read from known \wfences in $h$.
\end{proof}

\heading{Epochs.}
\CF{\wellff} fence transactions cluster
normal transactions (transactions that are not fence transactions)
into \textit{epochs}.
Epochs are generated as follows (\textsc{\AssignEpoch} in
Figure~\ref{fig:algocodefull}, line~\ref{li2:defepoch}).
First, the verifier traverses the \wfences
(they are a sequence of consecutive writes)
and assigns them epoch numbers which are their positions in the write fence sequence
(Figure~\ref{fig:algocodefull}, line~\ref{li2:startwfence}--\ref{li2:endwfence}).
Second, the verifier assigns epoch numbers to \rfences which are the epoch numbers
from the \wfences they read from
(Figure~\ref{fig:algocodefull}, line~\ref{li2:startrfence}--\ref{li2:endrfence}).
Finally, the verifier assigns epoch numbers to normal transactions and uses the epoch number of their
successive fence transactions (in the same client) minus one
(Figure~\ref{fig:algocodefull}, line~\ref{li2:assignnormal}).

During epoch assigning process, the verifier keeps track of the largest epoch number
that all clients have exceeded, denoted as \textit{$\epochagree$}
(Figure~\ref{fig:algocodefull}, line~\ref{li2:trackepochagree}).
In other words, every client has issued at least one fence transaction that
has epoch number $\ge \epochagree$.

One clarifying fact is that the epoch assigned to each transaction by the
verifier is not the value (an integer) of the epoch key which is generated by
clients. The verifier doesn't need the help from clients to assign epochs.

In the following, we denote a transaction with an epoch number $t$ as
a transaction with $\epochis{t}$.

\begin{lemma2}
If a history $h$ is \wellformed,
then a fence transaction with a smaller epoch number has a
path to any fence transaction with a larger epoch number
in the known graph of $\E(h)$.
\label{lemma:fencepath}
\end{lemma2}

\begin{proof}
First, we prove that a fence transaction with $\epochis{t}$ has a path to
another fence transaction with $\epochis{t+1}$.
Given history $h$ is \wellformed, by Claim~\ref{claim:wellfformedfence},
fence transactions are \wellff, which means
all the \wfences are a sequence of consecutive writes (by Definition~\ref{def:wellfformedfence}).
Because the \wfence with $\epochis{t}$ ($t$ is its position
in the sequence) and the \wfence with $\epochis{t+1}$ are adjacent in the sequence,
there is an edge (generated from reading-from dependency) from \wfence
with $\epochis{t}$ to \wfence with $\epochis{t+1}$.

Now consider a \rfence with $\epochis{t}$ which reads the epoch key
from the \wfence with $\epochis{t}$.
Because \sys's algorithm adds anti-dependency edges which point from
one write transaction's succeeding read transactions to its \adjacentwrite on the same key
(Figure~\ref{fig:algocodemain}, line~\ref{li:addrw}), there is an edge from the \rfence with $\epochis{t}$
to the \wfence with $\epochis{t+1}$.
Plus, all \rfences with $\epochis{t+1}$ read from the \wfence with $\epochis{t+1}$,
hence fence transactions with $\epochis{t}$ have paths to fence transactions with $\epochis{t+1}$.

By induction,
for any fence transaction with $\epochis{t+\Delta}$ ($\Delta \ge 1$),
a fence transaction with $\epochis{t}$
has a path to it.
\end{proof}

\begin{claim}
For a history $h$ and any its continuation $r$,
if $h \circ r$ is \wellformed,
then any normal transaction with $\epochis{\le \epochagree - 2}$ has a path to any
future normal transaction in the known graph of $\E(h \circ r)$.
\label{claim:epochagree}
\end{claim}

\begin{proof}
Take any normal transaction $\vtx_i$ with $\epochis{t}$ ($t \le
\epochagree -2$) and call $\vtx_i$'s client $C_1$.
Because the epoch of a normal transaction equals the epoch number of
its successive fence transactions in the same client minus one (Figure~\ref{fig:algocodefull}, line~\ref{li2:assignnormal}),
there is a fence transaction $\vtx_{f}$ in $C_1$ with $\epochis{t+1}$ ($t+1 \le \epochagree -1$).
By Lemma~\ref{lemma:fencepath}, $\vtx_{f}$ has a path to any fence transaction
with $\epochis{\epochagree}$.
And, by the definition of $\epochagree$, all clients have
at least one fence transactions with $\epochis{\ge \epochagree}$ in $h$.
Thus, there is always a path from $\vtx_i$---through $\vtx_{f}$ and the last
fence transactions of a client in $h$---to the future
transactions in $r$.
\end{proof}

\heading{Frozen transaction.}
In this section, we define \textit{frozen transactions}.
A frozen transaction is a transaction that no future transaction can be
scheduled prior to these transactions in any possible serial schedule.
Intuitively, if a transaction is frozen, this
transaction can never be involved in any cycles containing future transactions.

\begin{definition2}[Frozen transaction]
For a history $h$ that is \wellformed,
a \emph{frozen transaction} is a transaction with $\epochis{\le \epochagree -2}$
and all its predecessors (in the known graph of $\E(h)$)
also with $\epochis{\le \epochagree - 2}$.
\label{def:frozen}
\end{definition2}

\heading{\CF{\obsolete} transaction}.
Recall that the challenge of garbage collection is that \sys's verifier does
not know whether a value can be read by future transactions.
In the following, we define \textit{\obsolete transactions}
which cannot be read by future transactions.

\begin{definition2}[\CF{\obsolete} transaction]
For a history $h$ that is \wellformed,
an \emph{\obsolete transaction} on a key $x$ is a transaction
with $\epochis{\le \epochagree -2}$
that writes to $x$ and has a successor (in the known graph of $\E(h)$)
that also has $\epochis{\le \epochagree -2}$ and writes to $x$.
\label{def:obsolete}
\end{definition2}

\begin{corollary2}
For an \obsolete transaction on a key $x$,
it's predecessors which has $\epochis{\epochagree - 2}$ and writes $x$
is also an \obsolete transaction on $x$.
\label{corollary:obsolete}
\end{corollary2}

\begin{proof}
Call the \obsolete transaction $\vtx_i$ and its predecessor with
$\epochis{\epochagree - 2}$ that writes $x$ as $\vtx_j$.
By Definition~\ref{def:obsolete}, $\vtx_i$ has a successor $\vtx_k$ that has
$\epochis{\epochagree - 2}$ and writes $x$.
Hence, as a predecessor of $\vtx_i$'s, $\vtx_j$ is also a predecessor of $\vtx_k$.
Thus, by Definition~\ref{def:obsolete}, $\vtx_j$ is an \obsolete transaction.
\end{proof}

\begin{claim}
For a history $h$ and any its continuation $r$
that satisfies $h \circ r$ is \wellformed,
no future transaction can read key $x$ from an \obsolete transaction on $x$.
\label{claim:obsoletenofutureread}
\end{claim}

\begin{proof}
Assume to the contrary that there
exists a future transaction $\vtx_k$ reading key $x$ from an \obsolete transaction $\vtx_i$.
By Definition~\ref{def:obsolete}, there must exist another transaction $\vtx_j$
with $\epochis{\le \epochagree - 2}$ that writes to $x$ and $\vtx_i \rightsquigarrow \vtx_j$.

Now, consider transactions ($\vtx_i$, $\vtx_j$, $\vtx_k$).
In a \polyg,  they form a constraint: $\vtx_i$ and
$\vtx_j$ both writes to key $x$; and $\vtx_k$ reads $x$ from $\vtx_i$.
The constraint is $\langle \vtx_k \to \vtx_j, \vtx_j \to \vtx_i \rangle$.
However, both options create cycles in the known graph of $\E(h \circ r)$
which is a contradiction to that $h \circ r$ is \wellformed
(Definition~\ref{def:notwellformed}):
(1) if choose $\vtx_k \to \vtx_j$,
because $\vtx_j$ is $\epochis{\le \epochagree - 2}$ and $\vtx_k$ is a future transaction in $r$,
by Claim~\ref{claim:epochagree}, $\vtx_j \rightsquigarrow \vtx_k$;
(2) if choose $\vtx_j \to \vtx_i$,
because of $\vtx_i \rightsquigarrow \vtx_j$, there is a cycle.
\end{proof}

\begin{claim}
For a history $h$ and any its continuation $r$
that satisfies $h \circ r$ is \wellformed,
if a transaction $\vtx_i$ reads key $x$ from an \obsolete transaction on $x$,
then $\vtx_i$ has paths to future transactions
in the known graph of $\Pp(\E(h \circ r))$.
\label{claim:readfromobsolete}
\end{claim}

\begin{proof}
By Definition~\ref{def:obsolete}, the obsolete transaction (call it $\vtx_j$)
has a successor that has $\epochis{\le \epochagree -2}$ and writes $x$.
Call this successor $\vtx_k$.

Because both $\vtx_j$ and $\vtx_k$ write $x$ and $\vtx_i$ reads from $\vtx_j$,
they form a constraint $\langle \vtx_i \to \vtx_k, \vtx_k \to \vtx_j \rangle$.
Since $\vtx_k$ is a successor of $\vtx_j$, this constraint is pruned by
\textsc{\PruneConstraintsTR} and $\vtx_i \to \vtx_k$ in $\Pp(\E(h \circ r))$.
Plus, by Claim~\ref{claim:epochagree}, $\vtx_k$ has paths to any future transactions,
hence so does $\vtx_i$.

\end{proof}

\subsection{\CF{\delcand} transactions and {\solvedcon}s}
In this section, we define \textit{\delcand transactions} which are deleted
from the \depinfo by \sys's algorithm (Figure~\ref{fig:algocodefull},
line~\ref{li2:realdelete}).

\begin{definition2} [Candidates to remove]
In a history that is \wellformed, 
a transaction is a \emph{candidate to remove}, when
\begin{myitemize2}
  \item it is a frozen transaction; and
  \item it is either a read-only transaction
    or an \obsolete transaction on all keys it writes.
\end{myitemize2}
\label{def:delcandcand}
\end{definition2}

\begin{claim}
For a history $h$ and any its continuation $r$
that satisfies $h \circ r$ is \wellformed,
a future transaction cannot read from a candidate to remove.
\label{claim:noreadfromfuture}
\end{claim}

\begin{proof}
By Definition~\ref{def:delcandcand}, a candidate to remove is either
a read-only transaction which do not have writes,
or else an \obsolete transaction on all keys it writes,
which by Claim~\ref{claim:obsoletenofutureread} cannot be read by future transactions.
\end{proof}

\begin{definition2} [\CF{\delcand}]
For a history that is \wellformed, a transaction is \emph{\delcand} when
it is a candidate to remove 
and all the transactions in the same
\pscc are also candidates to remove.
\label{def:delcand}
\end{definition2}

\heading{\CF{\solvedcon} and \unsolvedcon}.
As mentioned in \S\ref{subsec:bruteforce}, a constraint in a \polyg involves
three transactions: two write transactions ($\wone$, $\wtwo$) writing to the same key
and one read transaction ($\rone$) reading this key from $\wone$.
And, this constraint ($\langle \rone \to \wtwo, \wtwo \to \wone \rangle$) has two
ordering options, either (1) $\wtwo$ appears before both $\wone$ and $\rone$ in
a serial schedule, or (2) $\wtwo$ appears after them.
We call a constraint as a \textit{\solvedcon} when the known graph
has already captured one of the options.

\begin{definition2}[\CF{\solvedcon}]
For a history $h$ that is \wellformed,
a constraint $\langle \rone \to \wtwo, \wtwo \to \wone \rangle$
is a \emph{\solvedcon},
when the known graph of $\E(h)$ has either
(1) $\wtwo \rightsquigarrow \wone,\, \wtwo \rightsquigarrow \rone$,
or (2) $\wone \rightsquigarrow \wtwo, \,\rone \rightsquigarrow \wtwo$.
\label{def:solvedcon}
\end{definition2}

\begin{fact}
Eliminating {\solvedcon}s doesn't affect whether a \polyg is
acyclic, because the ordering of the three transactions in a \solvedcon has
been already captured in the known graph.
\label{fact:solvedcon}
\end{fact}

For those constraints that are not \solvedcons, we call them
\textit{\unsolvedcons}.
Notice that both \solvedcons and \unsolvedcons are defined
on \polyg (not \mypolyg).

\begin{lemma2}
Given a history $h$ and a removable transaction $t$ in $h$, for any its
  continuation $r$ that satisfies $h \circ r$ is \wellformed,
  there is no \unsolvedcon that includes both $t$ and a future
  transaction in $r$.
\label{lemma:solvedcon}
\end{lemma2}

\begin{proof}
Call a constraint $\langle \rone \to \wtwo, \wtwo \to \wone \rangle$
($\wone$ and $\wtwo$ write to the same key $x$;
$\rone$ reads $x$ from $\wone$)
where one of the three transactions is \delcand and another
is a future transaction.

In the following, by enumerating all combinations of
possibilities, we prove such a constraint is always a \solvedcon.

\begin{myitemize2}

\item The \delcand transaction is $\rone$.

Because $\rone$ is \delcand, it is a frozen transaction.
As a predecessor of $\rone$ ($\rone$ reads from $\wone$),
by Definition~\ref{def:frozen},
$\wone$ has $\epochis{\le \epochagree - 2}$.
Hence, the last transaction $\wtwo$ must be the future transaction.
By Claim~\ref{claim:epochagree}, both $\wone$ and $\rone$ have paths to
future transactions including $\wtwo$.
Thus, this constraint is a \solvedcon.

\item The \delcand transaction is $\wone$.

Because $\rone$ reads $x$ from $\wone$,
by Claim~\ref{claim:noreadfromfuture}, $\rone$ cannot be a future transaction.
Hence, the future transaction must be $\wtwo$.

Because $\wone$ has writes, by Definition~\ref{def:delcandcand},
$\wone$ is an \obsolete transaction on key $x$.
Plus $\rone$ reads $x$ from $\wone$, by Claim~\ref{claim:readfromobsolete},
$\rone$ has paths to future transactions including $\wtwo$. Thus, this
constraint is a \solvedcon.

\item The \delcand transaction is $\wtwo$.

Because $h$ is a complete history, it is impossible to have
$\rone \in h$ but the transaction it reads $\wone \in r$.
Also, because there must be one future transaction (in $r$),
$\rone \in h \land \wone \in h$ is impossible.
Hence, there are two possibilities:

  \begin{myitemize2}

  \item $\wone \in r \land \rone \in r$.

  By Claim~\ref{claim:epochagree}, $\wtwo$ has a path to future
  transactions including $\wone$ and $\rone$. Hence, the constraint is a
  \solvedcon.

  \item $\wone \in h \land \rone \in r$.

  Now, consider the relative position of $\wone$ and $\wtwo$.
  Because a future transaction $\rone$ reads from $\wone$, by Claim~\ref{claim:noreadfromfuture},
  $\wone$ is not a candidate to remove.
  Further, by Definition~\ref{def:delcand}, $\wone$ and $\wtwo$ do not belong to the same \pscc.
  Hence, by Lemma~\ref{lemma:samepscc}, either $\wone \rightsquigarrow \wtwo$ or
  $\wtwo \rightsquigarrow \wone$.

  Assume $\wone \rightsquigarrow \wtwo$.
  Because $\wtwo$ a frozen transaction ($\wtwo$ is \delcand),
  by Definition~\ref{def:frozen}, $\wtwo$ has $\epochis{\epochagree - 2}$.
  Thus, by Corollary~\ref{corollary:obsolete}, $\wone$ is an \obsolete transaction on $x$,
  but it is read by a future transaction $\rone$,
  a contradiction to Claim~\ref{claim:noreadfromfuture}.

  Above all, $\wtwo \rightsquigarrow \wone$.
  And, by Claim~\ref{claim:epochagree}, $\wtwo \rightsquigarrow \rone$.
  Hence, this constraint is a \solvedcon.
  \end{myitemize2}

\end{myitemize2}

\end{proof}

\begin{lemma2}
Given a history $h$ that is \wellformed and a continuation $r$,
for any transaction $\vtx_i$ that is \delcand,
if the known graph of
$\E(h \circ r) \del \vtx_i$ is acyclic,
then the known graph of $\E(h \circ r)$ is acyclic.
\label{lemma:theknowngraphdelcand}
\end{lemma2}

\begin{proof}
Assume to the contrary that the known graph of $\E(h \circ r)$ has a cycle.
Because $h$ is \wellformed, there is no cycle in the known graph of $\E(h)$,
hence the cycle must include at least one transaction (say $\vtx_j$) in $r$.
Also, this cycle must include $\vtx_i$, otherwise $\E(h \circ r) \del
\vtx_i$ is not acyclic.

Now, consider the path $\vtx_j \rightsquigarrow \vtx_i$ in the cycle.
If such path has multiple edges, say $\vtx_j \rightsquigarrow \vtx_k \to \vtx_i$,
then because $\vtx_i$ is a frozen transaction, by Definition~\ref{def:frozen},
transaction $\vtx_k$ has $\epochis{\epochagree - 2}$.
Hence, by Claim~\ref{claim:epochagree}, $\vtx_k$ has a path to the future transaction $\vtx_j$
which generates a cycle in $\E(h \circ r) \del \vtx_i$, a contradiction.

On the other hand, assume the path is an edge $\vtx_j \to \vtx_i$.
There are four types of edges in the known graph,
but $\vtx_j \to \vtx_i$ can be none of them.
In particular,
\begin{myitemize2}

\item Edge $\vtx_j \to \vtx_i$ cannot be a reading-from edge (Figure~\ref{fig:algocodefull}, line~\ref{li2:edge1}) because $h$ is a complete
history.

\item Edge $\vtx_j \to \vtx_i$ cannot be an anti-dependency edge (Figure~\ref{fig:algocodemain}, line~\ref{li:addrw}),
in which case $\vtx_j$ has to read from a predecessor of $\vtx_i$.
By Corollary~\ref{corollary:obsolete}, the predecessor being read
is an \obsolete transaction, a contradiction to Claim~\ref{claim:obsoletenofutureread}.

\item Edge $\vtx_j \to \vtx_i$ cannot be a client order edge
(Figure~\ref{fig:algocodefull}, line~\ref{li2:edge2}),
because there is at least one fence transaction issued after $\vtx_i$ in the same client.

\item Edge $\vtx_j \to \vtx_i$ cannot be an edge added by \textsc{\PruneConstraintsTR}
(Figure~\ref{fig:algocodemain}, line~\ref{li:addprune1},\ref{li:addprune2}).
Because, by Lemma~\ref{lemma:solvedcon}, there is no \unsolvedcon between
$\vtx_i$ and $\vtx_j$.

\end{myitemize2}
\end{proof}

\subsection{The main argument}

\begin{lemma2}
Given a history $h$ and any continuation $r$
that satisfies $h \circ r$ is \wellformed,
for any transaction $\vtx_i$ that is \delcand,
$\E(\E(h) \del \vtx_i, r) \iff \E(h \circ r) \del \vtx_i$
\label{lemma:samedepinfo}
\end{lemma2}

\begin{proof}
First, we prove $\E(\E(h) \del \vtx_i, r) \iff \E(\E(h), r) \del \vtx_i$,
which means that the final \depinfo remains the same
no matter when \sys's algorithm deletes $\vtx_i$---before or after processing $r$.

An \depinfo has three components: $\readfrom$, $\wwpairs$, and the known graph $\vg$.
For $\readfrom$,
because $\vtx_i$ is \delcand, by Claim~\ref{claim:noreadfromfuture},
no future transactions in $r$ can read from it.
So deleting $\vtx_i$ before or after processing $r$ does not change the $\readfrom$,
and also no reading-from edges (from $\vtx_i$ to transactions in $r$)
are added to the known graph (Figure~\ref{fig:algocodefull}, line~\ref{li2:edge1}).
Similarly, for $\wwpairs$, there is no read-modify-write transactions in $r$
that read from $\vtx_i$, hence $\wwpairs$ are the same
in $\E(\E(h) \del \vtx_i, r)$ and $\E(\E(h), r) \del \vtx_i$,
and no edges are
added during \textsc{\InferRWEdges}
(Figure~\ref{fig:algocodemain}, line~\ref{li:addrw}).
Also, because $\vtx_i$ has $\epochis{\epochagree-2}$, there must be a fence transaction
that comes after $\vtx_i$ from the same client, hence there is no client order edge
from $\vtx_i$ to transactions in $r$ (Figure~\ref{fig:algocodefull}, line~\ref{li2:edge2}).
Thus, the known graphs in both {\depinfo}s are also the same.

Finally, by lemma~\ref{fact:identicaldepinfo},
$\E(\E(h), r) \del \vtx_i \iff \E(h \circ r) \del \vtx_i$.
\end{proof}

\begin{lemma2}
In a history that is \clientser,
for any {\unsolvedcon} $\langle \rone \to \wtwo, \wtwo \to \wone \rangle$
that includes a \delcand transaction in $\Pp(\E(h)$), %
all three transactions ($\wone$, $\wtwo$, and $\rone$) are in the same \pscc.
\label{lemma:inpscc}
\end{lemma2}

\begin{proof}
Consider the relative position of $\wone$ and $\wtwo$ in the known graph
of $\Pp(\E(h))$.

\begin{myitemize2}

  \item Assume $\wtwo \rightsquigarrow \wone$. Because of $\wone \to \rone$
  ($\rone$ reads from $\wone$), $\wtwo \rightsquigarrow \rone$. By
  Definition~\ref{def:solvedcon}, the constraint is a \solvedcon, a
  contradiction.
  
  \item Assume $\wone \rightsquigarrow \wtwo$. By \textsc{\PruneConstraintsTR}
  (Figure~\ref{fig:algocodemain}, line~\ref{li:addprune1},\ref{li:addprune2}),
  $\rone \rightsquigarrow \wtwo$. Again, by Definition~\ref{def:solvedcon}, the
  constraint is a \solvedcon, a contradiction.
  
  \item Finally, $\wone \not\rightsquigarrow \wtwo$ and $\wtwo \not\rightsquigarrow \wone$.
  By Lemma~\ref{lemma:sameknowngraph}, $\wone$ and $\wtwo$ are concurrent in the known graph
  of $\Qp(\E(h))$ too. Thus, by Lemma~\ref{lemma:samepscc}, they are in the same \pscc.
\end{myitemize2}
\end{proof}

\begin{theorem2}
Given a history $h$ that is \clientser and a continuation $r$,
for any transaction $\vtx_i$ that is \delcand,
there is:

$\Pp( \E(h \circ r) \del \vtx_i ) $ is acyclic
$\Leftrightarrow$  %
$\Pp( \E(h \circ r) )$ is acyclic.
\label{theorem:safedeleteone}
\end{theorem2}

\begin{proof}
First, we prove that if $h \circ r$ is \notwellformed,
\sys rejects (so that neither $\Po( \E(h \circ r) \del \vtx_i ) $ nor $\Po(
\E(h \circ r) )$ is acyclic).
By Definition~\ref{def:notwellformed}, $h \circ r$ either (1)
contains a write transaction having multiple \adjacentwrites,
or else (2) has cycles in the known graph of $\E(h \circ r)$.

For (1), if the write transaction which has multiple \adjacentwrites
is $\vtx_i$, given that $h$ is \clientser, there is at least one
\adjacentwrite is in $r$. However, $\vtx_i$ is an \obsolete transaction,
hence the algorithm detects a violation in $\E(h \circ r) \del \vtx_i$
(Claim~\ref{claim:noreadfromfuture}).
On the other hand, if the write transaction is not $\vtx_i$,
such transaction is detected in $\E(h \circ r) \del \vtx_i$
the same way as in $\E(h \circ r)$.
For (2), by Lemma~\ref{lemma:theknowngraphdelcand},
$\E(h \circ r) \del \vtx_i$ also has cycles
which the algorithm will reject.

Now, we consider the case when $h \circ r$ is \wellformed.
Because $h \circ r$ is \wellformed,
there is no cycles in the known graph of $\Pp(\E(h \circ r))$.
In $\Pp(\E(h\circ r))$, by Lemma~\ref{lemma:inpscc},
all transactions in \unsolvedcons that involves $\vtx_i$
are in the same \pscc (call this \pscc $\vpscc_i$).
Because $\vtx_i$ in $\vpscc_i$ is \delcand, by Definition~\ref{def:delcand},
all transactions in $\vpscc_i$ are \delcand.
Hence, no transaction in $\vpscc_i$ is involved in {\unsolvedcon}s with
either future transactions in $r$ (Lemma~\ref{lemma:solvedcon})
or other transactions in $h$ (Lemma~\ref{lemma:inpscc}).

\vspace{1ex}
\noindent
``$\Rightarrow$''.
Next, we prove that $\Pp(\E(h \circ r))$ is acyclic by construction an acyclic
compatible graph $\hat{g}$.
By Fact~\ref{fact:solvedcon}, we only need to concern \unsolvedcons
that might generate cycles.
Consider the transactions in $h \circ r$ but not in $\vpscc_i$,
the \unsolvedcons are the same in both
$\Pp(\E(h\circ r))$ and $\Pp(\E(h \circ r) \del \vtx_i)$;
given that $\Pp(\E(h \circ r) \del \vtx_i)$ is acyclic, there exists
a combination of options for \unsolvedcons that makes $\hat{g}$ acyclic
in these transactions.
Now, consider transactions in $\vpscc_i$.
Because all transactions in $\vpscc_i$ are in $h$ and $h$ is \clientser,
there exists a combination of options for the \unsolvedcons in $\vpscc_i$
so that $\hat{g}$ has no cycle in $\vpscc_i$.
Finally, because there is no \unsolvedcon between $\vpscc_i$ and other transactions in $h \circ r$
(proved in the prior paragraph), $\hat{g}$ is acyclic.

\vspace{1ex}
\noindent
``$\Leftarrow$''.
Because $\Po(\E(h\circ r))$ is acyclic, there exists an acyclic compatible graph $\hat{g}$.
We can construct a compatible graph $\hat{g'}$ for $\Po(\E(h\circ r) \del \vtx_i)$
by choosing all constraints according to $\hat{g}$---choose the edges in constraints
that appear in $\hat{g}$.
Given that the known graph in $\Po(\E(h\circ r) \del \vtx_i)$ is a subgraph of
$\Po(\E(h\circ r))$'s, $\hat{g'}$ is a subgraph of $\hat{g}$.
Hence, $\hat{g'}$ is acyclic, and $\Po(\E(h\circ r) \del \vtx_i)$ is acyclic.
\end{proof}

In the following, we use $h_i$ to represent the transactions fetched in
$i_{\textrm{th}}$ round.
The first round's history $h_1$ is a complete history itself;
for the $i_{\textrm{th}}$ round ($i \ge 2$), $h_i$ is a continuation of the prior history
$h_1 \circ \cdots \circ h_{i-1}$.
We also use $d_i$ to denote the transactions deleted
in the $i_{\textrm{th}}$ round.

\begin{lemma2}
Given that history $h_1 \circ \cdots \circ h_i \circ h_{i+1}$ is \wellformed,
if a transaction is \delcand in $h_1 \circ \cdots \circ h_i$,
then it remains \delcand in $h_1 \circ \cdots \circ h_i \circ h_{i+1}$.
\label{lemma:stilldelcand}
\end{lemma2}

\begin{proof}
Call this \delcand transaction $\vtx_i$ and the \pscc it is in during round $i$
as $\vpscc_i$.
Because \sys's algorithm does not delete fence transactions
(Figure~\ref{fig:algocodefull}, line~\ref{li2:nodeletefence}),
the epoch numbers for normal transactions in round $i$
remain the same in round $i+1$.
Hence, the $\epochagree$ in round $i+1$ is greater than or equal to the one in
round $i$.
Thus, if \sys's algorithm (\textsc{\SetFrozen} and \textsc{\GenFrontier})
sets a transaction (for example $\vtx_i$) as a candidate to remove in round $i$,
it still is in round $i+1$.

Because history $h_1 \circ \cdots \circ h_i \circ h_{i+1}$ is \wellformed,
transactions in $\vpscc_i$ do not have cycles with transactions in $h_{i+1}$.
Also, by Lemma~\ref{lemma:inpscc}, transactions in $\vpscc_i$ do not have
\unsolvedcons with $h_{i+1}$.
Thus, $\vpscc_i$ remains to be a \pscc in round $i+1$.
Above all, by Definition~\ref{def:delcand}, $\vtx_i$ is \delcand in round $i+1$.
\end{proof}

\begin{theorem2}
\sys's algorithm runs for $n$ rounds and doesn't reject
$\iff$
history $h_1 \circ h_2 \cdots \circ h_n$ is \clientser.
\end{theorem2}

\begin{proof}
We prove by induction.

For the first round,
\sys's algorithm only gets history $h_1$ (line~\ref{li2:historyinit})
and constructs its \depinfo $\E(h_1)$ (line~\ref{li2:creategraph2}).
Because \textsc{\VerifySerializability} doesn't reject,
the pruned \mypolyg $\Qp(\E(h_1))$ is acyclic, and
\begin{align*}
&\Qp(\E(h_1)) \text{ is acyclic} &\\
  & \iff \Qo(\E(h_1)) \text{ is acyclic} & \text{[Lemma~\ref{lemma:prunedq}]}\\
  &\iff h_1 \text{ is \clientser} & \text{[Lemma~\ref{lemma:polygraphser}]}
\end{align*}

For round $i$, assume that
history $h_1 \circ h_2 \cdots \circ h_{i-1}$ is \clientser
and \sys's algorithm doesn't reject for the last $i-1$ rounds.
In round $i$,
\sys' algorithm first fetches $h_i$,
gets the \depinfo from the last round which is
$\E(h_1 \circ \cdots \circ h_{i-1}) \del (d_0 \cup \cdots \cup d_{i-1})$,
and constructs a pruned \mypolyg
$\Qp(\E(\E(h_1 \circ \cdots \circ h_{i-1}) \del (d_0 \cup \cdots \cup d_{i-1}), h_i))$.
In the following, we prove that \sys's algorithm doesn't reject (the pruned \mypolyg is acyclic)
if and only if $h1 \circ \cdots \circ h_i$ is \clientser.

\begin{align*}
&\Qp(\E(\E(h_1 \circ \cdots \circ h_{i-1}) \del (d_0 \cup \cdots \cup d_{i-1}), h_i))  \text{ is acyclic} \\
&\iff \Qp(\E(h_1 \circ \cdots \circ h_i) \del (d_0 \cup \cdots \cup d_{i-1}))  \text{ is acyclic} \\
&  \text{\rightline{[Lemma~\ref{lemma:samedepinfo}, \ref{lemma:stilldelcand}]}} \\
& \iff \Pp(\E(h_1 \circ \cdots \circ h_i) \del (d_0 \cup \cdots \cup d_{i-1}))  \text{ is acyclic} \\
&  \text{\rightline{[Lemma~\ref{lemma:prunedq},\ref{lemma:polygraphser}]}} \\
& \iff \Pp(\E(h_1 \circ \cdots \circ h_i)) \text{ is acyclic}  \\
& \text{\rightline{[Theorem~\ref{theorem:safedeleteone}, Lemma~\ref{lemma:stilldelcand}]}} \\
& \iff \Po(\E(h_1 \circ \cdots \circ h_i))  \text{ is acyclic} \\
&  \text{\rightline{[Lemma~\ref{lemma:prunedq}]}} \\
& \iff h_1 \circ \cdots \circ h_i \text{ is \clientser} \\
& \text{\rightline{[Lemma~\ref{lemma:polygraphser}]}}
\end{align*}

\end{proof}

\end{document}